\documentclass[12pt]{article}

\usepackage[letterpaper,margin=1in]{geometry}
\usepackage{mathptmx}
\usepackage[T1]{fontenc}
\usepackage[utf8]{inputenc}
\usepackage{setspace}
\usepackage{amsmath,amssymb,amsthm,mathtools}
\usepackage{bm}
\usepackage{graphicx}
\usepackage{booktabs}
\usepackage{multirow}
\usepackage{caption}
\usepackage[round]{natbib}
\usepackage[protrusion=true,expansion=false]{microtype}
\usepackage[dvipsnames]{xcolor}
\usepackage[colorlinks=true,linkcolor=NavyBlue,citecolor=NavyBlue,urlcolor=NavyBlue]{hyperref}
\usepackage{enumitem}
\usepackage{placeins}

\newtheorem{theorem}{Theorem}
\newtheorem{corollary}[theorem]{Corollary}
\newtheorem{assumption}{Assumption}
\theoremstyle{definition}
\newtheorem{remark}{Remark}

\newcommand{\R}{\mathbb{R}}
\newcommand{\E}{\mathbb{E}}
\newcommand{\Prob}{\mathbb{P}}

\newcommand{\N}{\mathcal{N}}
\newcommand{\logit}{\mathrm{logit}}

\newcommand{\KL}{\mathrm{KL}}

\newcommand{\Phat}{\widehat{P}}
\newcommand{\Ptil}{\widetilde{P}}
\newcommand{\bbeta}{\bm{\beta}}
\newcommand{\trans}{^{\!\top}}
\DeclareMathOperator*{\argmin}{arg\,min}
\DeclareMathOperator*{\argmax}{arg\,max}
\newtheorem{lemma}{Lemma}
\newcommand{\Pbb}{\mathbb{P}}
\newcommand{\F}{\mathcal{F}}
\newcommand{\1}{\mathbf{1}}
\newcommand{\norm}[1]{\left\lVert #1\right\rVert}
\DeclareMathOperator{\kl}{kl}
\DeclareMathOperator{\Bern}{Bernoulli}
\DeclareMathOperator{\ECE}{ECE}

\title{\bf Bayesian Predictive Synthesis for Dynamic Networks: Forecasting and Identifying Structural Mechanisms}

\author{
Marios Papamichalis\thanks{Human Nature Lab, Yale University, New Haven, CT 06511, \texttt{marios.papamichalis@yale.edu}} \and
Regina Ruane\thanks{Department of Statistics and Data Science, The Wharton School, University of Pennsylvania, 3733 Spruce Street, Philadelphia, PA 19104-6340, \texttt{ruanej@wharton.upenn.edu}} \and
Theofanis Papamichalis\thanks{Department of Economics, Yale University, 28 Hillhouse Ave, New Haven, USA, \texttt{theofanis.papamichalis@yale.edu}}
}

\date{\today}

\begin{document}
\maketitle
\singlespacing

\begin{abstract}
\noindent
Networks are shaped by competing structural mechanisms, such as communities, geometry, or hubs. In a dynamic network the most predictive mechanism can change, and a model tied to one mechanism, or to fixed weights, cannot adapt as the dominant structure shifts. We develop dynamic Bayesian predictive synthesis for networks, in which a mechanism is an agent forecasting the next snapshot's edges and a synthesis layer combines them with time-varying weights. At each step the method returns a calibrated edge forecast and inference on the mechanism weights, with intervals valid given the fitted agents, so it also reports which mechanism is most informative. Inference of this kind requires a sparse-safe parametrization and an identification theory, under which a single graph identifies and estimates the weights. A sharp threshold separates distinguishable from indistinguishable mechanisms, a change in the active mechanism is tracked at an optimal per-switch cost, and for a single snapshot the method reduces to calibrated link prediction. On real networks, simulations, and benchmarks, the synthesis gives accurate, calibrated forecasts and recovers the leading mechanism when agents are well separated.
\smallskip

\noindent\textbf{Keywords:} dynamic networks; Bayesian predictive synthesis; forecast combination; calibration; link prediction.
\end{abstract}

\section{Introduction}
\label{sec:intro}

A network observed at a single time is often modeled as a draw from a generative mechanism. Three
families are standard: stochastic block models for community structure
\citep{holland1983stochastic,karrer2011stochastic}, latent-position and random dot product graphs
for geometry \citep{hoff2002latent,rubin2022statistical}, and configuration or Chung--Lu models for
degree heterogeneity \citep{chung2002connected}. Two difficulties arise in practice. First, the
practitioner rarely knows which mechanism generated the data. Second, when the network changes over
time, the mechanism that best explains it changes as well. Contact networks alternate between
community structure during class and hub structure during breaks within a single day, and
co-authorship and online social networks reorganize as communities form and dissolve. A model committed to one mechanism cannot follow such a change. Combining several mechanisms with
weights set in advance does not help. Once the most predictive mechanism changes, those weights still favor
the mechanism that was dominant before, and the combination misrepresents the current structure. We
instead let the weights change with the network and estimate them from data. The weights become the
object of inference: we estimate them, with uncertainty, from a single network snapshot, and we track their
change over time. The weight on an agent is a residualized log-odds projection coefficient: a positive value means that agent carries predictive structure beyond the others, not that the network was generated by that mechanism, so the leading mechanism is a pool-dependent statement and adding or removing an agent can move the lead. Forecast combination for time series learns such weights from a long history.

The changing balance of mechanisms raises three questions. Can the weights on competing mechanisms be
learned from data, not fixed in advance? At what rate can they be learned from a single
network? And can the resulting forecast be trusted as a calibrated probability? Existing work does
not answer them. Forecast combination is developed for time series, not for graph-valued outcomes
with single-snapshot inference. Network estimation treats a single fixed graph, not weights that
follow a changing mechanism.

We extend dynamic Bayesian predictive synthesis to networks. Predictive synthesis
\citep{west1992modelling,mcalinn2019dynamic} specifies how a set of agent forecasts relates to the
outcome, and updates that specification as agents and outcomes are observed. Each mechanism is fit
to the history and forecasts the edge probabilities of the next network snapshot. A synthesis layer combines
these forecasts on the logit scale through time-varying coefficients. The coefficients are the state
of a dynamic linear model \citep{west1997bayesian}: they drift as a random walk and are updated
against the observed edges. At each step the method returns two things: a calibrated forecast of the
next edges, and an estimate of the weights with its uncertainty. With a single time point, it
reduces to calibrated link prediction.

A network on $n$ nodes has of order $n^2$ node pairs. Under a conditionally independent edge model,
each pair is a Bernoulli observation, and the Fisher information for the weights is of order
$n^2\rho_n$, where $\rho_n$ is the edge density. The weights are therefore estimable from one network
snapshot, at rate $(n^2\rho_n)^{-1/2}$, a rate set by the size of the network, not the length
of the series. In a large network $n^2\rho_n$ is in the thousands even at low density, while a change
in the most predictive mechanism takes a few snapshots, so the weights are estimated faster than they
change. The same quantity $n^2\rho_n$ also determines the contrast two mechanisms need to be
distinguished from one network snapshot, and the rate at which a change in mechanism is tracked.

Bayesian predictive synthesis grew out of the reconciliation of probability assessments
\citep{lindley1979reconciliation,west1984bayesian} and the modeling of probabilistic agent
opinion \citep{genest1985modeling,west1992modelling}, in which a decision maker
updates a prior over an outcome by treating the agents' predictive distributions as data through
a conditional synthesis density. \citet{mcalinn2019dynamic} made the synthesis dynamic for time
series by letting the synthesis coefficients follow a dynamic linear model, and
\citet{mcalinn2020multivariate} carried it to multivariate macroeconomic forecasting;
\citet{johnson2017bayesian} developed agent-specific calibration, \citet{tallman2024bayesian}
introduced decision-guided synthesis, and \citet{chernis2023predictive} replaced the linear
synthesis function with a nonparametric tree ensemble. The construction generalizes Bayesian
model averaging \citep{hoeting1999bayesian} and the stacking of predictive distributions
\citep{wolpert1992stacked,breiman1996stacked,yao2018using}, and is related to opinion pools
\citep{genest1986combining} and optimal prediction pools \citep{geweke2011optimal}. Every paper
in this lineage operates on scalar or vector time series: the outcome is a number or a
low-dimensional vector and the agents forecast it. None acts on a graph, where the outcome is an
adjacency matrix, the agents are generative network mechanisms, and the unit of information is the
dyad. We extend the dynamic synthesis of \citet{mcalinn2019dynamic} to networks. Incorporating the network
structure into the synthesis yields results with no time-series counterpart: the combination weights
are identified and estimated from a single graph, which requires a sparse-safe parametrization and an
identification theory the time-series literature does not need.

The work nearest to ours combines several network predictors for one prediction task.
\citet{ghasemian2020stacking} stack a large collection of link-prediction algorithms into a
near-optimal meta-learner on static networks, \citet{he2024sequential} extend feature-based
stacking to temporal networks, and \citet{zhang2025network} average latent-space models of
differing latent dimension by $K$-fold edge cross-validation \citep{li2020network}, with
asymptotic optimality for link prediction. These are combination methods for networks, and
we benchmark against them, but they differ from the present construction in what they return and
what they justify: they stack non-mechanistic predictor scores or average models within a single family
by a frequentist criterion, so in the convex-pooling form used here they produce a ranking, not a calibrated predictive
distribution, carry no calibration intercept, and stay within the convex hull of their
inputs, and they provide no inference on the contribution of any mechanism, no single-snapshot
identification theory, no separation rate, and no tracking guarantee. \citet{papamichalis2025graphon} combine network mechanisms at the population graphon; they use no dynamic state and give no inference on a finite graph.
\citet{papamichalis2026minimax} study, for one fixed graph, when two mechanisms are distinguishable, and \citet{papamichalis2025state} model time-varying spillovers on a network in state-space form.
The present paper differs in three ways. It estimates the combination weights, with uncertainty,
from a finite graph. It tracks a change in the active mechanism at an optimal per-switch cost. And it shows
when model averaging fails under misspecification.

The agents instantiate the standard families. The community agent is the stochastic block model
\citep{holland1983stochastic} in its mixed-membership \citep{airoldi2008mixed} and spectral
\citep{rohe2011spectral} forms; the latent-geometry agent is the latent space model
\citep{hoff2002latent} and the random dot product and generalized random dot product graphs
\citep{athreya2018statistical,rubin2022statistical}; the degree agent is the Chung--Lu
model \citep{chung2002connected}; and the local agent is the Adamic--Adar index
\citep{adamic2003friends}. We treat each as a one-step-ahead predictive distribution and combine
them, not selecting one \citep{saldana2017many,wang2023occam} or merging their parameters
into a single hybrid model.

A parallel response to dynamic networks fits a single dynamic generative model: dynamic latent
space models \citep{sewell2015latent,durante2014nonparametric}, dynamic stochastic block models
\citep{xu2014dynamic,matias2017statistical}, and temporal exponential random graph models
\citep{hanneke2010discrete,krivitsky2014separable}, surveyed alongside the static families by
\citet{goldenberg2010survey}. Each commits to one mechanism for the entire series. The present
method instead holds several mechanisms at once and lets a dynamic synthesis layer reweight them as
their relative fit changes, the regime-switching behavior the experiments target.

Two further threads fix the scale of the theory. \citet{kolaczyk2015question} show that a graph
carries between $O(n)$ and $O(n^2)$ units of information about a tie propensity according to its
density, and the dyadic-inference literature \citep{graham2020network,chiang2026double} treats the
dyad count as the unit of information; we use this to set the rate at which combination weights,
not model parameters, are learned from a single snapshot. Detection thresholds for a fixed
graph, community recovery at the Kesten--Stigum boundary
\citep{decelle2011asymptotic,mossel2015reconstruction,abbe2018community}, latent-geometry detection
\citep{bubeck2016testing}, and global structure tests \citep{gao2017testing,jin2018network}, have
as their snapshot-wise, mechanism-against-mechanism analogue our separation rate inside a tracking
problem, with graphon estimation rates \citep{gao2015rate,klopp2017oracle} entering through the
agents. Dynamic model averaging \citep{raftery2010online} tracks which model is active over time
but, as an averaging rule, inherits the misspecification collapse
\citep{berk1966limiting,kleijn2006misspecification} that Theorem~\ref{thm:t4} formalizes and the
experiments expose, while the link-prediction problem itself \citep{liben2003link,lu2011link}
is the $T=1$ task that our construction calibrates.

We develop dynamic Bayesian predictive synthesis for dynamic
networks: each candidate mechanism is used as an agent that forecasts the next snapshot's edges,
and a synthesis layer combines these forecasts with weights that drift over time, so the method
is at once a calibrated one-step-ahead edge forecaster, in the Bayesian predictive synthesis
tradition of \citet{mcalinn2019dynamic,mcalinn2020multivariate} and scored sequentially by a proper
scoring rule, and a reading of which mechanism is currently active.
A sparse-safe parametrization keeps it valid on sparse graphs, and with a single snapshot it
reduces to calibrated link prediction. The main contribution is a theory for these combination
weights, built on one fact: a graph on $n$ nodes at sparsity $\rho_n$ contains on the order of
$n^2\rho_n$ edges, so the weights are learned from a single snapshot, not from a long
series. The theory establishes five properties of these weights. First, the weights are identified and
estimable from one snapshot, with a normal limit; their intervals are valid conditional on the
fitted agents, and remain valid unconditionally for finite-dimensional mechanisms through a
cross-fold orthogonal correction that accounts for refitting the agents on the same graph. Second,
two mechanisms are distinguishable from one snapshot precisely when their edge probabilities differ
by more than a threshold that no test can improve upon. Third, a change in the active mechanism is
tracked at a cost incurred only at the changes, at a rate matching a minimax lower bound. Fourth,
when no agent is correct, the weights converge to a well-defined projection, the forecast remains
calibrated, and the synthesis improves on Bayesian model averaging, which under misspecification collapses onto one agent. Fifth, the static
link-prediction case carries a case-control correction.

We evaluate the method on five real dynamic networks, controlled simulations, and
link-prediction benchmarks, comparing against model averaging, stacking, and purpose-built dynamic
network models. The real networks are a financial correlation network, high-school and hospital-ward contact networks with known group labels, and two
timestamped networks in the tens of thousands of nodes. Two findings stand out. The synthesis produces
calibrated forecasts together with interpretable, time-varying weights that carry valid intervals;
a recalibrated latent-space competitor matches the probabilities only after a separate
recalibration step and returns no weights, so the contribution is calibrated forecasting with
mechanism inference.

\section{Dynamic networks and mechanism agents}
\label{sec:model}

We observe a sequence of undirected networks on a fixed node set $V=\{1,\dots,n\}$,
recorded as symmetric adjacency matrices $A_1,\dots,A_T \in \{0,1\}^{n\times n}$ with zero
diagonal. Conditionally on a symmetric
edge-probability matrix $P_t \in [0,1]^{n\times n}$, the upper-triangular entries of $A_t$
are independent Bernoulli draws, $A_{t,ij}\sim \mathrm{Bernoulli}(P_{t,ij})$ for $i<j$. We
work in the sparse regime
\begin{equation}
\label{eq:sparse}
P_t = \rho_n\, M_t, \qquad \rho_n \to 0, \qquad n\rho_n \ge c\log n,
\end{equation}
where $\rho_n$ is a global sparsity factor and $M_t$ is a bounded structural matrix; the
condition $n\rho_n\ge c\log n$ keeps the average degree above the connectivity threshold
\citep{le2017concentration}.

A \emph{mechanism agent} is a forecaster that maps the history $A_{1:t-1}$ to a one-step
predictive edge-probability matrix $\Phat^{(j)}_t$. We use three model-based agents.

\begin{itemize}[leftmargin=1.4em,itemsep=2pt]
\item \textbf{SBM (community).} A degree-corrected stochastic block model
\citep{karrer2011stochastic} fit by regularized spectral clustering
\citep{rohe2011spectral,lei2015consistency}, yielding block assignments and
Laplace-smoothed block connection probabilities; $\Phat^{\mathrm{SBM}}_{t,ij}$ is the
estimated probability for the blocks of $i$ and $j$.
\item \textbf{GRDPG (geometry).} A generalized random dot product graph fit by adjacency
spectral embedding with magnitude-truncated eigenvalues \citep{rubin2022statistical};
$\Phat^{\mathrm{GRDPG}}_{t,ij}= x_i\trans I_{p,q}\, x_j$ for embedded positions $x_i$ and
signature $(p,q)$.
\item \textbf{Chung--Lu (degree).} A sparse degree model
$\Phat^{\mathrm{CL}}_{t,ij}= d_i d_j / (2m)$ with degrees $d_i$ and edge count $m$
\citep{chung2002connected}, truncated to the sparse scale by the two-sided clip of
Section~\ref{sec:bps} so that high-degree pairs do not exceed it; this is closely related to the preferential-attachment per-dyad score.
\end{itemize}

For link prediction we additionally include a local per-dyad agent, Adamic--Adar
\citep{adamic2003friends,liben2003link}, $\mathrm{AA}(i,j)=\sum_{z\in N(i)\cap N(j)} 1/\log
d_z$, mapped to a probability by the fixed, label-free squashing $\Phat^{\mathrm{AA}}_{t,ij}=1-e^{-\mathrm{AA}(i,j)/\tau}$, with $\tau$ the mean Adamic--Adar score over the scored dyads. The three model-based
agents are those for which Theorem~\ref{thm:t1} supplies explicit rates; the local agent
is a nonparametric predictor that the oracle inequality covers without a parametric rate.

\section{Dynamic Bayesian predictive synthesis}
\label{sec:bps}

Let $\sigma(\cdot)$ denote the logistic function and $\logit$ its inverse. Stack the
agent predictives for a dyad $(i,j)$ at time $t$ into a feature vector $z_{t,ij} = (1,\,
\logit \Phat^{(1)}_{t,ij},\dots,\logit \Phat^{(J)}_{t,ij})\trans$. The synthesis predictive
is the affine-logit pool
\begin{equation}
\label{eq:synthesis}
\Ptil_{t,ij} = \sigma\!\big(\bbeta_t\trans z_{t,ij}\big),
\qquad
\bbeta_t = (\beta_{t0},\beta_{t1},\dots,\beta_{tJ})\trans \in \R^{J+1},
\end{equation}
where $\beta_{t0}$ is a calibration intercept and $\beta_{tj}$ is the time-$t$ weight on
agent $j$. Each $\beta_{tj}$ is a log-odds weight, not a mixture share: a positive value indicates that
agent $j$ aligns with the edge structure left unexplained by the other agents, and a negative value
that it is anti-aligned (Table~\ref{tab:weights}). The individual weights are interpretable when the agents are not aliased,
a condition made precise in Section~\ref{sec:theory}.

\begin{table}[t]
\centering
\caption{Interpretation of the synthesis weights, reported through the scale-invariant standardized contribution
$|\theta_{tj}|H_{t,jj}^{1/2}$; the condition number $\kappa(H_t)$ gates whether individual weights
are separately interpretable.}
\label{tab:weights}
\small
\begin{tabular}{@{}p{0.33\textwidth}p{0.60\textwidth}@{}}
\toprule
Weight pattern & Reading \\
\midrule
Large positive standardized contribution & agent supplies residual predictive structure beyond the others \\
Near zero & agent is redundant given the others \\
Negative & agent is anti-aligned with the edges, a sign of misspecified or opposite structure \\
Ill-conditioned $H_t$ (large $\kappa$) & individual weights are not separately interpretable \\
\bottomrule
\end{tabular}
\end{table}

The coefficient vector is the latent state of a dynamic linear model
\citep{west1997bayesian}:
\begin{equation}
\label{eq:dlm}
\bbeta_t = \bbeta_{t-1} + \omega_t, \quad \omega_t \sim \N(\bm 0, W_t),
\qquad
A_{t,ij} \mid \bbeta_t \sim \mathrm{Bernoulli}\!\big(\sigma(\bbeta_t\trans z_{t,ij})\big).
\end{equation}
The evolution variance $W_t$ is set by a single discount factor $\delta\in(0,1]$, which
controls how quickly weights may move: $\delta=1$ holds the weights fixed (suited to $T=1$),
while $\delta<1$ permits tracking. A single discount controls the whole filter, and the results below
are insensitive to its value: sweeping $\delta\in[0.7,0.95]$ leaves one-step-ahead switch recovery at a
single snapshot, the agent-refit lag, up to $\delta=0.95$, while the transductive recovery time, which
carries no refit lag, scales as $1/(1-\delta)$ exactly as the tracking bound predicts, so a default near
$0.9$ suffices and no per-network tuning of the discount is needed. The Bernoulli observation gives no closed-form update, so we use a Laplace (extended-Kalman) step.
The posterior mode is computed by damped Newton iteration with step-halving, which decreases the
negative log posterior monotonically, and the curvature at the mode is propagated as the posterior
covariance. The one-step forecast is the plug-in predictive
$\sigma(\widehat\bbeta_{t\mid t-1}^\top z_{t,e})$ at the filtered mean; the curvature enters the
weight intervals but not the point forecast. The agents are fit on a fold of dyads held out from the
fold used to update the synthesis, so no dyad is used both to fit an agent and to estimate the
weights. The forward pass at each time $t$ is: fit the agents on $A_{1:t-1}$ and form the features
$z_t$; compute the one-step-ahead forecast $\Ptil_t$ from the current state; and, on observing
$A_t$, filter $\bbeta_t$ by the discounted Bernoulli update.

\paragraph{Sparse-safe parametrization.}
In the sparse regime $\rho_n\to0$ a fixed affine-logit pool of raw logits need not preserve
the global density, so for the analysis (and in the implementation) we separate the density
from the structural combination. Writing $\rho_n$ for an estimate of the global edge density
(the observed dyad mean), set the centered agent features and the offset synthesis
\begin{equation}
\label{eq:offset}
u^{(j)}_{t,ij} = \logit \Phat^{(j)}_{t,ij} - \log\rho_n, \qquad
\Ptil_{t,ij} = \sigma\!\Big(\log\rho_n + \alpha_t + \textstyle\sum_{j=1}^J \theta_{tj}\,
u^{(j)}_{t,ij}\Big),
\end{equation}
so $\log\rho_n$ is a fixed offset carrying the density, $\alpha_t$ is the calibration
intercept, and $\theta_t=(\theta_{t1},\dots,\theta_{tJ})$ are the structural weights; the
state is $\bbeta_t=(\alpha_t,\theta_t)$. The clipping of $\Phat^{(j)}$ is at the sparse scale on both sides,
$[\rho_n\epsilon,\,\rho_n/\epsilon]$ intersected with $[\rho_n\epsilon,\,1-\epsilon]$, not a
fixed $[\epsilon,1-\epsilon]$. The two-sided sparse cap keeps the centered features bounded,
$|u^{(j)}_{t,e}|\le\log(1/\epsilon)+O(\rho_n)$, as Assumption~\ref{ass:agents} and the design
condition of Theorem~\ref{thm:t1} require; without an upper cap at the sparse scale a high-degree
Chung--Lu dyad with $d_id_j/(2m)=O(1)$ would give $u^{(j)}_{t,e}\asymp\log(1/\rho_n)\to\infty$, so
the cap, not merely the lower clip, is what enforces boundedness, and it does not distort the density
as $\rho_n\to0$. In the implementation the density level is carried entirely by the free intercept $\alpha_t$, which
the filter updates from past snapshots only, so the one-step-ahead forecast uses no information from
the snapshot being scored; the fixed offset $\log\rho_n$ of the theory and the filtered intercept
differ only by a constant that $\alpha_t$ absorbs, leaving the structural-weight inference unchanged. The theory
states a single scale $\rho_n$ for transparency; the statements and rates carry over unchanged to a
time-varying $\rho_{t,n}$ with $c\rho_n\le\rho_{t,n}\le C\rho_n$ uniformly in $t$, on replacing
$N_t\asymp n^2\rho_n$ by $N_t\asymp n^2\rho_{t,n}$. The theory below targets the population weight
$\bbeta^\circ_t$ in this parametrization.

\paragraph{Computation.}
The agents are standard static network models fit to the observed history; the dynamic linear
model acts only on the low-dimensional vector of weights that combines them, not on the network
itself. Because that weight state is Gaussian and linear, its recursion uses standard Kalman
filtering and Rauch--Tung--Striebel smoothing \citep{west1997bayesian,papamichalis2025state}.

\section{Theory}
\label{sec:theory}

Every result in this section is network-specific and is organized around one quantity, the
edge information in a single snapshot. With the sparse-safe synthesis \eqref{eq:offset},
write the dyad feature $z_{t,e}=(1,u^{(1)}_{t,e},\dots,u^{(J)}_{t,e})\trans$ for a dyad
$e=(i,j)$ and let
\begin{equation}
\label{eq:info}
N_t \;:=\; \sum_{e\in\mathcal E_t} q_{t,e}(1-q_{t,e})
\;\asymp\; |\mathcal E_t|\,\rho_n \;\asymp\; n^2\rho_n ,
\qquad
H_t \;:=\; N_t^{-1}\!\sum_{e\in\mathcal E_t} q_{t,e}(1-q_{t,e})\,z_{t,e}z_{t,e}\trans ,
\end{equation}
where $q_{t,e}=q_{t,e}(\bbeta^\circ_t)$ and $\mathcal E_t$ is the set of sampled dyads. $N_t$
is the scalar dyadic information scale and grows as the expected edge count of one graph, never as
the series length $T$; $H_t$ is the normalized information matrix, so $N_tH_t$ is the Fisher
information for the weights. The central
object is the population log-odds projection of the network onto the set of mechanisms,
\begin{equation}
\label{eq:popbeta}
\bbeta^\circ_t \;=\; \argmin_{\bbeta\in\Theta}\; R_t(\bbeta),
\qquad
R_t(\bbeta) \;=\; \frac{1}{|\mathcal E_t|}\sum_{e\in\mathcal E_t}
\KL\!\big\{\mathrm{Bern}(p_{t,e}),\,\mathrm{Bern}(q_{t,e}(\bbeta))\big\},
\end{equation}
and the theory is about estimating, tracking, calibrating, and interpreting
$\bbeta^\circ_t$. Three results are the headline. First, a single snapshot identifies and estimates the weights at the dyadic-information rate, with an aliasing condition for when individual weights are recoverable (Theorem~\ref{thm:t1}); this single-graph identification has no time-series counterpart and is the conceptual core. Second, two mechanisms separate above a contrast threshold, and a switch in the active mechanism is tracked with a bounded recovery delay that is minimax-optimal up to constants (Theorems~\ref{thm:t2} and \ref{thm:t3}, with the matching lower bound in Section~\ref{supp:s7} of the Supplement). Third, when no agent equals the truth the weights converge to a log-odds projection and the intercept enforces calibration, separating the synthesis from model averaging and stacking under misspecification in this setting (Theorem~\ref{thm:t4}). A corollary specializes the first result to one snapshot, recovering calibrated link prediction (Corollary~\ref{cor:t1}).

We write $a_{n,t}$ for the agents' estimation error on the information-weighted logit (score) scale,
$a_{n,t}^2=N_t^{-1}\sum_e q_{t,e}(1-q_{t,e})\,|\widehat z_{t,e}-z^\circ_{t,e}|^2$, the deviation of the fitted
dyad logits from their population targets; the probability-scale graphon rates cited in
Assumption~\ref{ass:agents} convert to this scale through the sparse lower bound $p_{t,e}\asymp\rho_n$,
which amplifies probability error by $1/\rho_n$ (Supplement~S2), and
$r_{\mathrm{Lap}}$ for the Laplace/extended-Kalman linearization error of the filter.

\begin{assumption}[Sparsity]
\label{ass:sparse}
The edge probabilities satisfy \eqref{eq:sparse} with the entries of $M_t$ bounded in
$[c_0,C_0]$ on their support scale, and the average degree exceeds the connectivity
threshold, $n\rho_n\ge c\log n$. The number of blocks $K$ and the embedding dimension
$d=p+q$ are fixed.
\end{assumption}

\begin{assumption}[Agent regularity and rates]
\label{ass:agents}
After two-sided sparse-scale clipping to $[\rho_n\epsilon,\,\rho_n/\epsilon]\cap[\rho_n\epsilon,\,1-\epsilon]$, each
centered agent feature $u^{(j)}_{t,e}$ is bounded, $|u^{(j)}_{t,e}|\le\log(1/\epsilon)+O(\rho_n)$; the
upper sparse cap is what enforces this, since a Chung--Lu hub with $\Phat^{(j)}_{t,e}=O(1)$ would
otherwise give $u^{(j)}_{t,e}\asymp\log(1/\rho_n)$. The model-based agents are consistent for their best-fitting
parameters. On the probability scale the dyadic mean-squared error of the fitted edge probabilities
is of order $\rho_n(K^2/n^2+\log K/n)$ (SBM, \citealp{gao2015rate,lei2015consistency}),
$\rho_n d/n\cdot\mathrm{polylog}\,n$ (GRDPG, \citealp{le2017concentration,rubin2022statistical}),
and $\rho_n/n$ (Chung--Lu). Because a centered logit feature equals the relative probability
deviation, the score-scale quantity $a_{n,t}^2$ defined above is larger by $\rho_n^{-2}$: thus
$a_{n,t}^2\asymp K^2/(n^2\rho_n)$ (SBM, the $\log K/n$ clustering term a loose worst-case bound that
does not bind for fixed separated blocks), $d/(n\rho_n)\cdot\mathrm{polylog}\,n$ (GRDPG), and
$1/(n\rho_n)$ (Chung--Lu).
\end{assumption}

\begin{assumption}[Design non-degeneracy]
\label{ass:design}
The normalized information matrix is well conditioned, $\lambda_{\min}(H_t)\ge
\underline\lambda>0$ uniformly in $t$. This is a quantitative strengthening of the condition that the agents' dyad-logit surfaces
$(1,u^{(1)}_{t,e},\dots,u^{(J)}_{t,e})$ are not affinely dependent across dyads, and is the
identifiability condition for $\bbeta^\circ_t$; its failure is the aliasing of
Theorem~\ref{thm:t1}.
\end{assumption}

\begin{assumption}[State evolution]
\label{ass:state}
The state $\bbeta_t=(\alpha_t,\theta_t)$ follows \eqref{eq:dlm} with discount $\delta$ and
evolution scale $\tau^2=1-\delta$. A comparator sequence $\bbeta^\circ_t$ has total movement
$\sum_{t=2}^T\|\bbeta^\circ_t-\bbeta^\circ_{t-1}\|^2\le \mathcal V$, realized by $S$ regime
switches of jump size at most $D$ in the piecewise-stable case.
\end{assumption}

\subsection{Single-snapshot dyadic-rate identification and aliasing}

One result separates network synthesis from time-series synthesis. A single graph identifies and
estimates the weights, at a rate set by the size of the graph, not the length of the series. The one-snapshot
estimator $\widehat\bbeta_t$ is the cross-fitted maximum a posteriori (or, with a flat prior,
maximum likelihood) logistic fit of the dyads $A_{t,e}$ on the offset design $z_{t,e}$.

\begin{theorem}[Single-snapshot identification, dyadic rate, and aliasing]
\label{thm:t1}
Fix $t$, condition on the sigma-field $\F_{t,n}$ generated by the history, the cross-fitted
agents, the sampled dyad set $\mathcal E_t$ (with $m_{t,n}=|\mathcal E_t|$), the offset
$\log\rho_n$, and the true conditional edge probabilities $p_{t,e}$, and assume the dyad set
is conditionally non-informative, so that $A_{t,e}\sim\mathrm{Bernoulli}(p_{t,e})$ are
independent given $\F_{t,n}$. Suppose the basic sparse-scale conditions hold: the parameter
space $\Theta$ is compact and convex, $\max_e\|z_{t,e}\|\le B$, and there are constants
$0<c_q<C_q$ with $c_q\rho_n\le q_{t,e}(\beta)\le C_q\rho_n\le\tfrac12$ uniformly on $\Theta$,
$p_{t,e}\le C_p\rho_n$, and $m_{t,n}\rho_n\to\infty$ (Supplement~S1, conditions
B1--B3).
\begin{enumerate}[label=(\alph*),leftmargin=1.9em,itemsep=2pt]
\item \textbf{(Identification and aliasing.)} Because $q_{t,e}(1-q_{t,e})>0$, the null space
of $H_t(\beta)$ equals $\mathcal N_t=\{a:a\trans z_{t,e}=0\ \forall e\in\mathcal E_t\}$ for
every $\beta$. Hence the conditional synthesis law $\beta\mapsto\{q_{t,e}(\beta)\}$ is
identified if and only if $H_t$ is full rank ($\mathcal N_t=\{0\}$); otherwise $\beta$ is
identified only modulo $\mathcal N_t$, and a linear functional $a\trans\beta$ is estimable iff
$a\in\mathrm{Range}(H_t)$. An exact affine dependence among $1,u^{(1)}_{t,\cdot},\dots,
u^{(J)}_{t,\cdot}$ aliases the corresponding separate agent weights.
\item \textbf{(Rate and conditional normality.)} Adding interior, Gram-nonsingularity, and
weak-penalty conditions (Supplement~S1, C1--C4), the one-snapshot estimator
$\widehat\bbeta_t$ satisfies, conditionally on $\F_{t,n}$,
\begin{equation}
\label{eq:t1rate}
\|\widehat\bbeta_t-\bbeta^\circ_t\|
= O_p\!\Big(\sqrt{\tfrac{J+1}{N_t}}\Big),
\quad N_t\asymp m_{t,n}\rho_n\asymp n^2\rho_n \ (\text{with } m_{t,n}\asymp n^2 \text{ for full-snapshot synthesis}),
\end{equation}
with the asymptotic linear representation $\sqrt{N_t}(\widehat\bbeta_t-\bbeta^\circ_t)=
H_t^{-1}N_t^{-1/2}\sum_e(A_{t,e}-p_{t,e})z_{t,e}+o_p(1)$ and therefore
\begin{equation}
\label{eq:t1clt}
\sqrt{N_t}\,\big(\widehat\bbeta_t-\bbeta^\circ_t\big)
\;\Rightarrow\;
\N\!\big(0,\;H_t^{-1}V_t H_t^{-1}\big),
\quad V_t=H_t \text{ under correct Bernoulli-logit specification.}
\end{equation}
The effective sample size is the edge count: every weight is estimable from one graph at the
dyadic rate $(n^2\rho_n)^{-1/2}$, not the series rate. If the fitted-agent population score is
within $O_p(a_{n,t})$ of an oracle score and the implemented filter is within
$O_p(r_{\mathrm{Lap}})$ of $\widehat\bbeta_t$, then the oracle-centered error is
$O_p(\sqrt{(J+1)/N_t}+a_{n,t}+r_{\mathrm{Lap}})$, and the oracle-centered limit holds with the
same covariance when $\sqrt{N_t}\{a_{n,t}+r_{\mathrm{Lap}}\}\to0$.
\end{enumerate}
\end{theorem}

The Fisher information for the weights is of order $n^2\rho_n$, so they are estimable from a single
network, without a time series. The case $T=1$, a single network on a fixed node set, is the static
link-prediction specialization of the same estimator. The rank condition in
Assumption~\ref{ass:design} is a diagnostic of identifiability: when it fails, two mechanisms are
aliased and their individual weights are not separately identified.

The limit of Theorem~\ref{thm:t1} is conditional on the fitted agents. When a flexible agent is
refit on the same snapshot, the conditional intervals can undercover. The following debiased
construction removes the bias. Split the synthesis dyads of
snapshot $t$ into $K$ folds $I_1,\dots,I_K$, estimate each agent surface on the dyads outside a
fold, and replace the synthesis score $s_e(\bbeta,\xi)=z_e(\xi)\{A_e-q_e(\bbeta,\xi)\}$ by a
fold-specific orthogonal score $\psi^\perp_{k,e}=s_e-c_{k,e}$, where the correction $c_{k,e}$ is
built from the nuisance influence function so that it annihilates the first-order effect of
agent-estimation error, $\partial_\xi\Psi^\perp_{n,k}(\bbeta^\circ,\xi^\circ)[h]=0$, while
preserving the projection target $\bbeta^\circ$. Let $\widetilde\bbeta$ solve the cross-fold
orthogonal estimating equation $\sum_{k}\sum_{e\in I_k}\psi^\perp_{k,e}(\bbeta,\widehat\xi^{-k})=0$,
with $\widehat\xi^{-k}$ measurable with respect to the out-of-fold data.

The synthesis is computed in two passes. The filtering pass of \eqref{eq:dlm} produces the
one-step-ahead forecast and the weight point estimates. The cross-fold orthogonal pass produces weight
estimates whose confidence intervals remain valid when the agents are refit on the same graph. The
forecasts below use the first pass; the weight intervals use the second.

This unconditional validity is the content of the cross-fold orthogonal limit. The orthogonal
estimator is asymptotically normal \emph{unconditionally} over the dyad draw, the random fold split,
and the agent refits under a score-scale agent-error rate of $o(N_t^{-1/4})$, far weaker than the
$o(N_t^{-1/2})$ the conditional limit needs. This second-order rate is met by agents whose parameter
count is fixed, so that each parameter is informed by order $n^2$ dyads: the SBM block probabilities
give $a_{n,t}^2\asymp K^2/(n^2\rho_n)$, hence $\sqrt{N_t}\,a_{n,t}^2\asymp K^2/(n\sqrt{\rho_n})\to0$.
It is not met in the sparse regime by agents whose $\Theta(n)$ latent parameters are each informed by
only $\Theta(n)$ dyads: the GRDPG positions and Chung--Lu degrees give $a_{n,t}^2\asymp d/(n\rho_n)$
and $1/(n\rho_n)$, so $\sqrt{N_t}\,a_{n,t}^2\asymp d/\sqrt{\rho_n}$ and $1/\sqrt{\rho_n}$ do not
vanish. The unconditional Wald intervals are therefore established for the finite-dimensional
community weight; for every agent the conditional limit of Theorem~\ref{thm:t1}, centered at the
fitted projection, supplies valid intervals given the fitted features. The construction is dyadic
double/debiased machine learning \citep{chiang2026double} specialized to a structure those methods do
not address: the nuisances are network models refit on the same graph on which the weights are estimated, and
the target is the vector of interpretable mechanism weights.

\subsection{A separation rate for distinguishing mechanisms}

The scale $n^2\rho_n$ also determines how far apart two mechanisms must be to be distinguished from a
single network snapshot. The natural object is a
\emph{separation rate} in graphon log-odds, not a density threshold; a hard $n^2\rho_n=O(1)$
regime is excluded by Assumption~\ref{ass:sparse} (which forces $n^2\rho_n\ge cn\log n\to
\infty$), so the boundary lives at vanishing contrast inside the assumed regime.

The argument rests on a perturbation lemma for sparse Bernoulli-logit models: on the sparse scale,
the divergence between two edge laws is controlled by the dyadic log-odds gap.

\begin{lemma}[Sparse Bernoulli-logit perturbations]
\label{lem:sparse}
Let $p,q\in(0,1)$ with $x=\logit q-\logit p$, and suppose $c_0\rho\le p,q\le C_0\rho\le\tfrac12$
and $|x|\le\eta_0$ for fixed constants. Then there are $0<c<C<\infty$, depending only on
$c_0,C_0,\eta_0$, such that the Kullback--Leibler divergence, the negative log Hellinger
affinity, and the chi-square divergence between $\mathrm{Bern}(q)$ and $\mathrm{Bern}(p)$ are
all between $c\,p(1-p)x^2$ and $C\,p(1-p)x^2$ (the chi-square only bounded above), with the
same bounds when $p,q$ are interchanged; and the per-dyad log-likelihood-ratio has variance at
most $C\,p(1-p)x^2$.
\end{lemma}

\begin{theorem}[Mechanism separation rate and tracking]
\label{thm:t2}
Let two mechanism configurations produce sparse edge laws $p_e,q_e\asymp\rho_n$ over
$m_{t,n}\asymp n^2$ conditionally independent dyads, with bounded dyadic log-odds gap
$x_e=\logit q_e-\logit p_e$, and aggregate Fisher signal $S_n=\sum_e p_e(1-p_e)x_e^2\asymp
n^2\rho_n\Delta_n^2$ under a uniform gap $|x_e|\asymp\Delta_n$. Let the synthesis track with
discount $\delta$.
\begin{enumerate}[label=(\alph*),leftmargin=1.9em,itemsep=2pt]
\item \textbf{(Achievability and tracking.)} If $S_n\to\infty$ (equivalently
$n^2\rho_n\Delta_n^2\to\infty$), the likelihood-ratio test separates the two configurations
from a single snapshot with total error tending to zero, and between switches the filtered
synthesis recovers the active weight at the per-snapshot squared rate $O((n^2\rho_n)^{-1})$,
the drift over the horizon contributing the discounted tracking regret
$O(\bar N\,\mathcal V/(1-\delta)^2)$ of Theorem~\ref{thm:t3}.
\item \textbf{(Impossibility.)} If $S_n=O(1)$ (equivalently $n^2\rho_n\Delta_n^2=O(1)$), the
two snapshot laws are mutually contiguous and no test, hence no measurable functional of the
data, separates them consistently: no test has power tending to one, and the minimax total testing
error is bounded below by a positive constant; in the limiting case $S_n\to0$ the power of every test
tends to its size.
\end{enumerate}
The critical separation rate is therefore $\Delta_n\asymp(n^2\rho_n)^{-1/2}$: two mechanisms
are distinguishable from one snapshot precisely when their graphon log-odds differ by more than
$(n^2\rho_n)^{-1/2}$, and affinely dependent mechanisms ($\Delta_n=0$) are never distinguishable
at any density.
\end{theorem}

\begin{remark}[Temporal tracking versus single-graph distinguishability]
\label{rmk:vs-minimax}
Theorem~\ref{thm:t2} concerns \emph{tracking a moving target} and is distinct from asking
whether two mechanisms can be distinguished in one fixed graph, the static question studied by
\citet{papamichalis2026minimax}. The two boundaries share the scale
$n^2\rho_n$ for the same elementary reason, that a snapshot's order-$n^2$ dyads carry Fisher information of order $n^2\rho_n$, but the operative
quantity here is the discounted accumulated separation
$n^2\rho_n\,\Delta_n^2\,(1-\delta^h)/(1-\delta)$ over an $h$-snapshot window, which over the $O(1)$
snapshots that localization takes reduces to $\asymp h\,n^2\rho_n\Delta_n^2$ (the $1/(1-\delta)$
canceling, in agreement with the local rate of Theorem~\ref{thm:t3}) and tends to
$n^2\rho_n\Delta_n^2/(1-\delta)$ in a stable regime; this discount interaction
has no counterpart in the static problem.
\end{remark}

\subsection{Dynamic tracking, switch recovery, and mechanism localization}

Because the weights are learned per snapshot, the filter tracks a switching mechanism by
paying only for the switches, and, when the margin condition holds, it identifies the active
mechanism once the discounted information exceeds the margin between mechanisms. The following
theorem establishes the bounded recovery delay and the margin condition under which the active
mechanism is identified, and is tested by the weight-trajectory experiments of
Section~\ref{sec:experiments}.

\begin{theorem}[Prequential tracking, switch recovery, and localization]
\label{thm:t3}
Let $\widehat\bbeta_{t\mid t-1}$ be the one-step-ahead filtered state. Under
Assumptions~\ref{ass:sparse}--\ref{ass:state} with cross-fitting,
\begin{equation}
\label{eq:regret}
\begin{split}
\sum_{t=1}^{T} |\mathcal E_t|\big\{R_t(\widehat\bbeta_{t\mid t-1})-R_t(\bbeta^\circ_t)\big\}
\;\le\;
C\Big[&(J{+}1)\,\lambda_T\big(1+\log T+T(1-\delta)\big)\\
&+ \frac{\bar N}{(1-\delta)^2}\!\sum_{t=2}^T\|\bbeta^\circ_t-\bbeta^\circ_{t-1}\|^2
+ \bar N\!\sum_{t=1}^T \varepsilon_t^2\Big],
\end{split}
\end{equation}
where $\lambda_T=\log(2(J{+}1)T/\alpha)$, $\bar N\asymp N_t\asymp n^2\rho_n$ is the per-snapshot
information, and the implementation error obeys $\varepsilon_t\le a_{n,t}+r_{\mathrm{Lap}}$. The left
side is the cumulative excess one-step log-loss: $|\mathcal E_t|\{R_t(\cdot)-R_t(\bbeta^\circ_t)\}$ is
the excess expected log-loss of snapshot $t$ in the per-dyad normalization of \eqref{eq:popbeta}.
For piecewise-stable comparators with $S$ switches of size at most $D$ the movement is
$\sum_{t}\|\bbeta^\circ_t-\bbeta^\circ_{t-1}\|^2\le SD^2$, so its contribution is at most
$C\bar N SD^2/(1-\delta)^2$: the synthesis charges for switches, not for every time point. The
discounted-filter estimation term $(J{+}1)\lambda_T(1+\log T+T(1-\delta))$ and the
agent-and-linearization floor $\bar N\sum_t\varepsilon_t^2$ are lower order when the agents are
well estimated and the Laplace step is accurate. Suppose an active
mechanism $r_t\in\{\mathrm{SBM},\mathrm{GRDPG},\mathrm{CL},\mathrm{AA}\}$ has margin
$\theta^\circ_{t,r_t}-\max_{j\ne r_t}\theta^\circ_{t,j}\ge\kappa>0$, and let
$N_t^\delta=\sum_{s\le t}\delta^{t-s}N_s$ be the discounted effective information. Then
\begin{equation}
\label{eq:localize}
\|\widehat\bbeta_{t\mid t}-\bbeta^\circ_t\|_\infty
= O_p\!\Big(\sqrt{\tfrac{\log(JT)}{N_t^\delta}} + \frac{\sum_{s\le t}\delta^{t-s}N_s
\|\bbeta^\circ_s-\bbeta^\circ_t\|}{N_t^\delta} + a_{n,t}+r_{\mathrm{Lap}}\Big),
\end{equation}
and whenever the right side is below $\kappa/2$, $\Prob\{\argmax_j\widehat\theta_{t,j}=r_t\}
\to1$. Localization is defined in coefficient space: the agent features $u^{(j)}_{t,e}$ share the
centered-logit scale and are bounded in the same range, so the margin compares commensurable
quantities; for a scale-invariant statement one may compare the signed standardized contributions
$\theta_{t,j}\,H_{t,jj}^{1/2}$, to which the same margin argument applies. After a switch the recovery delay is of order
$h=O\big(\log(JT/\alpha)/(N\kappa^2)\;\vee\;\log(D/\kappa)/|\log\delta|\big)$.
\end{theorem}

\begin{remark}[Model averaging has switch inertia]
\label{rmk:inertia}
Bayesian model averaging weights by cumulative marginal likelihood, so after a regime of
length $L$ its log-odds for the old mechanism are of order $L$ and it needs an
$\Omega(L)$-long window to switch; the discounted filter forgets in $O(1/(1-\delta))$ steps,
independent of $L$. The wrong-mechanism window of averaging grows with how long the previous
mechanism was dominant, whereas the synthesis's recovery delay does not.
\end{remark}

Theorem~\ref{thm:t3} is an upper bound: the discounted filter localizes the active mechanism
within $O(\log(JT/\alpha)/(N_t^\delta\kappa^2))$ post-switch snapshots and pays $O(SD^2)$ for
$S$ switches. A matching local minimax lower bound, stated and proved in Section~\ref{supp:s7} of the
Supplement, shows these rates are minimax-optimal up to constants, not merely attained: on the
one-switch sparse family over a uniformly non-aliased design, no procedure, even one told the
switch time, the agent arrays, the information level $N$, and the margin $\kappa$, localizes the
post-switch mechanism in fewer than $\log J/(n^2\rho_n\kappa^2)$ snapshots, any $\alpha$-reliable
declaration has expected delay at least $c\,\log(eJ/\alpha)/(n^2\rho_n\kappa^2)$, and the cumulative
reliable delay and switch regret obey matching $S\log(eJ/\alpha)/(n^2\rho_n\kappa^2)$ and $S\log J$
floors. The lower-bound constant comes from a two-point/Fano reduction and the upper-bound constant
from a Freedman inequality, so the matching is in rate, not in constant.

For $\delta$ bounded away from $1$, the discount affects only the constant, since the discounted post-switch information
$N(1-\delta^h)/(1-\delta)\asymp hN$ coincides with the undiscounted $hN$ over the $O(1)$ snapshots
localization takes, so no $1/(1-\delta)$ factor enters the rate.

\subsection{Universal misspecification: projection, calibration, and competitor gaps}

When no agent equals the truth, the weights converge to a log-odds projection, the intercept
enforces calibration through a score equation, and the synthesis improves on model averaging by a
computable gap.

\begin{theorem}[Projection, calibration, and BMA/stacking gaps]
\label{thm:t4}
Suppose no agent equals the truth.
\begin{enumerate}[label=(\alph*),leftmargin=1.9em,itemsep=2pt]
\item \textbf{(Projection and sparse limit.)} The one-snapshot weight converges at
$(n^2\rho_n)^{-1/2}$ to $\bbeta^\circ_t$, the log-odds projection \eqref{eq:popbeta} of the
true mechanism onto the set of agents. In the sparse graphon limit
$\rho_n^{-1}\Phat^{(j)}_t\to W^{(j)}_t$ and $\rho_n^{-1}p_t\to W^\circ_t$, the offset
synthesis tends to the log-linear form $\rho_n^{-1}q_t(\bbeta)\to
e^{\alpha_t}\prod_j(W^{(j)}_t)^{\theta_{tj}}=:W_t(\bbeta)$, and $\bbeta^\circ_t$ solves the
Bernoulli/KL (Poisson-type) variational problem $\argmin_{\bbeta}\int[W_t(\bbeta)-W^\circ_t\log
W_t(\bbeta)]$, not an $L^2$ projection.
\item \textbf{(Calibration.)} The score equations $\sum_e(p_{t,e}-q^\circ_{t,e})=0$ and
$\sum_e(p_{t,e}-q^\circ_{t,e})u^{(j)}_{t,e}=0$ hold: the intercept forces zero average
miscalibration and the slopes make the residual orthogonal to the mechanism features. Full
reliability additionally requires the approximation error of the set
$\inf_{\bbeta}\|p_t-q_t(\bbeta)\|_{2,t}$ to be small; the finite-sample bound is
$\mathrm{ECE}_b(\widehat q_t)\le C[\inf_{\bbeta}\|p_t-q_t(\bbeta)\|_{2,t}+\rho_n(\sqrt{(J{+}1)/N_t}
+\sqrt{b/N_t}+a_{n,t}+r_{\mathrm{Lap}})]$, where $b$ is the number of equal-mass calibration bins (distinct from the block count $K$); the estimation, binning, and agent terms enter on the sparse probability scale, hence the factor $\rho_n$.
\item \textbf{(Competitor gaps.)} Against Bayesian model averaging, which collapses onto the
single Kullback--Leibler-closest agent (or, under ties, the Kullback--Leibler-minimizing set) \citep{berk1966limiting}, with $e_j$ the coefficient vector $(\alpha_t,\theta_t)=(0,\mathbf 1_j)$ for which the offset gives $q_t(e_j)=\Phat^{(j)}_t$ exactly,
\begin{equation}
\label{eq:gap}
R_t(\mathrm{BMA})-R_t(\mathrm{BPS}) \;\to\; \min_j R_t(e_j)-\inf_{\bbeta\in\Theta}R_t(\bbeta)
\;\ge\;0,
\end{equation}
strict whenever the affine-logit projection strictly improves on every single agent (a Kullback--Leibler/Bregman
Pythagorean relation for the affine-natural-parameter family of~(a) at an interior projection, and a
convexity inequality at the boundary). Against convex stacking we claim no universal dominance, the families
differ; but when the truth lies at curvature-distance $d_t$ outside the convex probability
hull of the agents and the affine-logit approximation error $\epsilon_t$ is at most $\tfrac12 c\,d_t^2$, the projection improves on convex stacking, $\inf_{w\in\Delta_J}R_t(\sum_j w_j\Phat^{(j)}_t)-\inf_{\bbeta}
R_t(\bbeta)\ge \tfrac12 c\,d_t^2$.
\end{enumerate}
\end{theorem}

\begin{remark}[What the calibration guarantee gives]
\label{rmk:calib}
Theorem~\ref{thm:t4}(b) gives average calibration, through the intercept score equation, and
orthogonality of the residual to the agent features, through the slope equations; the empirical
calibration error on held-out dyads is bounded by these plus the approximation error of the set of
agents, and after a case-control evaluation the same offset returns population calibration
(Corollary~\ref{cor:t1}). None of these is full per-bin reliability, which requires
$\inf_{\bbeta}\|p_t-q_t(\bbeta)\|_{2,t}$ to be small and is not guaranteed for a fixed set of agents.
\end{remark}

\subsection{Static link prediction: case-control correction and graceful failure}

The $T=1$ specialization is static link prediction. Two network-specific points make it precise:
the held-out evaluation is a case-control sample, and agents that carry no signal fail
gracefully.

\begin{corollary}[$T=1$ oracle, case-control offset, graceful failure]
\label{cor:t1}
For $T=1$ with train/\allowbreak validation/\allowbreak test dyad splits and cross-fitted agents, the test risk
obeys the oracle inequality $R_{\mathrm{test}}(\widehat q_{\mathrm{BPS}})\le
\inf_{\theta\in\Theta}R_{\mathrm{test}}(q_\theta)+O_p(\sqrt{(J{+}1)/m_{\mathrm{val}}}+a_n)$,
and adding an agent cannot raise the oracle risk, strictly lowering it when the new agent
carries residual log-score information beyond the existing span. If sampling depends only on the label, $S\perp X\mid Y$, with positives sampled with
probability $s_1$ and negatives with $s_0$ (an equal-size negative set is the case $s_1\gg
s_0$), then $\logit\Prob(Y{=}1\mid X,S{=}1)=\logit\Prob(Y{=}1\mid X)+\log(s_1/s_0)$, so the
population-calibrated forecast is $\widehat p_{\mathrm{pop}}=\sigma(\logit\widehat
p_{\mathrm{cc}}-\log(s_1/s_0))$; the intercept $\alpha$ absorbs this constant, so the
synthesis is automatically calibrated to the evaluation sample and a single offset recovers
the population scale. Finally, if the agent features carry no ranking signal, $Y\perp z$, then
$\theta^\circ=0$ and the synthesis reduces to an intercept-only base-rate forecast: no score ranks the dyads better than chance while the expected calibration error tends to zero.
\end{corollary}

\begin{remark}[Identifiability of the predictive; why stacking can out-rank but mis-calibrate]
\label{rmk:final}
The synthesis acts on edge probabilities, invariant to the SBM permutation and the GRDPG
$O(p,q)$ symmetry, so the predictive and (under Assumption~\ref{ass:design}) the weight
$\bbeta^\circ_t$ are identified even when agent parameters are not. Theorem~\ref{thm:t4}(b) decouples ranking from calibration: the ordering of dyads is invariant to
the intercept, while calibration is not. A convex combination, which carries no free intercept, can
therefore rank well yet remain miscalibrated.
\end{remark}

\subsection{Estimation defaults}
\label{sec:practice}

The procedure below summarizes the estimation: on each snapshot the agents are fit on held-out dyad folds, their forecasts are mapped to centered-logit features, the discounted Laplace filter updates the mechanism weights, and the cross-fold construction returns their intervals. Table~\ref{tab:defaults} lists the defaults used throughout, each paired with the diagnostic that signals when it should be revisited.

\begin{center}
\fbox{\begin{minipage}{0.9\textwidth}
\textbf{Procedure: dynamic BPS on a sequence of networks.}\\[2pt]
\emph{Input:} snapshots $\{A_t\}_{t=1}^T$; agent families; discount $\delta$; logit clip $\epsilon$; dyad folds $K$.\\[1pt]
For each snapshot $t=1,\dots,T$:
\begin{enumerate}[leftmargin=1.5em,itemsep=1pt,topsep=2pt]
\item Split the dyads into $K$ folds; on each held-out fold fit the agents on the complementary dyads and the discounted past, and read off their one-step edge probabilities.
\item Form the centered-logit features $u^{(j)}_{t,e}=\logit\widehat p^{(j)}_{t,e}-\overline{\logit\widehat p^{(j)}_t}$, with the probabilities clipped to the sparse scale $[\rho_n\epsilon,\,\rho_n/\epsilon]\cap[\rho_n\epsilon,\,1-\epsilon]$.
\item Update the weights by the discounted Laplace filter with forgetting factor $\delta$: predict $\widehat\bbeta_{t\mid t-1}$, then correct with snapshot $t$.
\item Report the forecast $\widehat q_{t,e}=\sigma(\widehat\alpha_t+\sum_j\widehat\theta_{t,j}u^{(j)}_{t,e})$, the leading agent $\argmax_j\widehat\theta_{t,j}\,H_{t,jj}^{1/2}$ (signed, as in Theorem~\ref{thm:t3}), the condition number $\kappa(H_t)$, and the cross-fold weight intervals.
\end{enumerate}
\emph{Output:} calibrated one-step forecasts, weight trajectories with intervals, and the per-snapshot leading mechanism with its aliasing diagnostic.
\end{minipage}}
\end{center}

\begin{table}[t]
\centering
\small
\caption{Estimation defaults used throughout, with the diagnostic that signals when each should be
revisited.}
\label{tab:defaults}
\begin{tabular}{@{}p{0.19\textwidth} l p{0.52\textwidth}@{}}
\toprule
Quantity & Default & Revisit when \\
\midrule
Discount $\delta$ & $0.9$ & switches are slower or faster than the $1/(1-\delta)$ window \\
Logit clip $\epsilon$ & $10^{-3}$ & many agent probabilities pile at $0$ or $1$ \\
Dyad folds $K$ & $5$ & the held-out fold is too small to fit the agents \\
Validation dyads & $20\%$ held out & the calibration intercept is unstable across folds \\
Weight interpretability & $\kappa(H_t)$ moderate & $\kappa(H_t)$ large: report the forecast, not the split of credit \\
Localization & margin $\kappa$ comfortable & small margin: name no leader; the forecast stays calibrated \\
\bottomrule
\end{tabular}
\end{table}

\section{Numerical studies}
\label{sec:experiments}

We report results on five real dynamic networks in the main text (an S\&P 500 correlation network,
SocioPatterns high-school and hospital-ward contact networks, and the arXiv HEP-PH and Enron
networks at scale) and a sixth, the Bitcoin-OTC trust network, in the Supplement, together
with controlled switching simulations. The forecast is evaluated by negative log-likelihood (NLL) and two calibration diagnostics: the Kolmogorov--Smirnov distance of
the randomized probability integral transform to uniformity (PIT-KS) and the expected calibration
error (ECE). The bootstrap intervals below are descriptive summaries. The competitors are Bayesian
model averaging, convex stacking, equal-weight pooling, and each single agent. We report, at each snapshot, the leading structural agent, the one with the
largest standardized contribution $|\theta_{t,j}|H_{t,jj}^{1/2}$; this magnitude is the largest residual predictive contribution, while the signed margin of Theorem~\ref{thm:t3} controls localization of a positively aligned mechanism. A positive
weight means an agent carries predictive structure beyond the others, not that the network was
generated by that mechanism, and individual weights are separately interpretable only when the
information matrix is well conditioned, its condition number $\kappa(H_t)$ reported as the aliasing
diagnostic (Section~\ref{sec:e_t4}).

The studies follow the paper's results: a financial correlation network for calibration, regime
tracking, and localization (Section~\ref{sec:sp}); two label-carrying contact networks for weight
validation (Section~\ref{sec:hs}); controlled simulations that test the single-snapshot theory under
known truth (Section~\ref{sec:controlled}, Theorems~\ref{thm:t1}--\ref{thm:t4}); and two networks
in the tens of thousands of nodes for scale and out-of-sample behavior (Section~\ref{sec:scale}).

Two points frame the comparison. First, the synthesis is evaluated as a calibrated density forecaster
that also returns mechanism inference; where a competitor can be recalibrated to match its forecast
scores (Section~\ref{sec:recalib}), the distinguishing output is the calibrated probabilities
delivered jointly with interpretable, uncertainty-quantified mechanism weights in a single pass,
not a higher raw score. Second, the time-varying weights do not lower the average forecast
score on a stable panel, as the $\delta=1$ ablation of Section~\ref{sec:recalib} shows; their role is
switch recovery and mechanism localization, together with the pooling of recent snapshots that
sharpens the weights where a single snapshot is too sparse.

\subsection{Dynamic networks: tracking and localization}

\paragraph{A financial network from S\&P 500 returns.}
\label{sec:sp}
Asset-return correlation networks change their dominant structure across market regimes. In calm
markets, returns cluster by sector, a community structure. In stress, a single market factor
dominates and most pairs move together, a low-rank structure. We construct a binary dynamic network
from these correlations. From daily closing prices of the $n=470$ S\&P 500 constituents with full
history from February 2013 to February 2018, we form log-returns and take non-overlapping quarterly
windows. In each window we compute the cross-correlation matrix and threshold $|\mathrm{corr}|>0.5$
to a binary graph, giving $T=19$ snapshots on the complete-case panel of constituents with full
price histories.

The density of the thresholded graph varies widely, from $0.018$ with $1{,}986$ edges in the 2017
calm market to $0.50$ with $55{,}214$ edges in the 2015--16 stress episode. The per-snapshot
information $n^2\rho_n$ therefore ranges over a factor of twenty-eight within one dataset. The dense stress
snapshots, with density up to $0.50$, sit outside the sparse regime $\rho_n\to0$ of \eqref{eq:sparse}; the
sparse-regime guarantees therefore attach to the sparse calm snapshots, while the dense snapshots are a
robustness check that the regime tracking persists beyond the theory's sparsity range. We take the
regime label from the data, as the share of the leading eigenvalue of the correlation matrix, which
is high in stress and low in calm; this recovers the documented market history from the correlation
threshold alone, with no fitting to the volatility series.

\begin{table}[t]
\centering
\small
\caption{One-step-ahead forecast scores on three networks. Lower NLL/PIT-KS/ECE is better; best in each block in bold.}
\label{tab:forecasts}
\begin{tabular}{lccc}
\toprule
Method & NLL & PIT-KS & ECE \\
\midrule
\multicolumn{4}{l}{\textit{S\&P 500 correlation} ($n=470$, $T=19$)}\\
\textbf{Dynamic BPS} & \textbf{0.653} & \textbf{0.092} & \textbf{0.123} \\
Dynamic SBM (smoothed)   & 0.829 & 0.245 & 0.284 \\
Dynamic GRDPG (smoothed) & 0.860 & 0.174 & 0.225 \\
BMA            & 1.047 & 0.249 & 0.303 \\
DMA (forgetting) & 1.047 & 0.249 & 0.303 \\
Stacking       & 0.940 & 0.222 & 0.284 \\
Equal weight   & 0.947 & 0.215 & 0.281 \\
SBM agent      & 1.053 & 0.251 & 0.305 \\
GRDPG agent    & 1.959 & 0.199 & 0.269 \\
Chung--Lu agent& 1.332 & 0.241 & 0.309 \\
\midrule
\multicolumn{4}{l}{\textit{High-school contact} ($n=327$, $T=41$)}\\
\textbf{Dynamic BPS} & \textbf{0.613} & \textbf{0.107} & \textbf{0.140} \\
BMA          & 2.065 & 0.390 & 0.453 \\
Stacking     & 2.015 & 0.394 & 0.454 \\
Equal weight & 2.110 & 0.397 & 0.456 \\
SBM agent    & 2.066 & 0.391 & 0.453 \\
GRDPG agent  & 3.727 & 0.360 & 0.432 \\
Chung--Lu agent & 3.106 & 0.455 & 0.484 \\
\midrule
\multicolumn{4}{l}{\textit{Hospital ward contact} ($n=75$, $T=44$)}\\
\textbf{Dynamic BPS} & \textbf{0.558} & \textbf{0.119} & \textbf{0.139} \\
BMA          & 1.604 & 0.359 & 0.411 \\
Stacking     & 1.404 & 0.309 & 0.375 \\
Equal weight & 1.419 & 0.292 & 0.356 \\
SBM agent    & 1.596 & 0.352 & 0.410 \\
GRDPG agent  & 3.126 & 0.269 & 0.316 \\
Chung--Lu agent & 2.420 & 0.290 & 0.349 \\
\bottomrule
\end{tabular}
\end{table}

Table~\ref{tab:forecasts} reports the forecast scores. The synthesis is among the best-calibrated and
lowest-log-score methods compared here (NLL $0.653$, PIT-KS $0.092$, ECE $0.123$), against NLL $0.83$ to $1.96$ and
ECE $0.23$ to $0.31$ for the classical combiners and the smoothed generative models. The comparison
includes two purpose-built dynamic competitors, a discount-smoothed dynamic stochastic block model
and a dynamic latent-space forecaster, each fit on the discounted-average adjacency, and a
recency-weighted sequential stacking \citep{he2024sequential} (NLL $0.94$). The convex combiners
cannot leave the agents' hull or carry a calibration intercept, so they rank but do not calibrate.

The dynamic latent-space model orders candidate edges somewhat more sharply but is badly
miscalibrated (NLL $0.860$, ECE $0.225$, against $0.653$ and $0.123$). This is the intended
trade-off: a latent-geometry model orders candidate edges well but is overconfident, whereas the
synthesis gives up a little ranking for calibrated probabilities. Resampling the eighteen one-step
differences, the synthesis improves NLL over the dynamic block model by $0.176$ ($95\%$ CI
$[0.107,0.240]$) and over the dynamic latent-space
model by $0.206$ ($[0.101,0.321]$).

Calibration carries a decision consequence. Under an asymmetric cost, with a missed link costing $r$
times a false alarm and the threshold set from $\tau=1/(1+r)$ with no labeled tuning set, the raw
latent-space forecasts cost $16\%$ more at $r=2$, $34\%$ at $r=5$, and $45\%$ at $r=10$ than the
synthesis, because the sharper ranking is read at the wrong scale; model averaging costs more still.
Recalibrating the competitor closes this gap (Section~\ref{sec:recalib}).

Figure~\ref{fig:sp}
reports the one-step log-score and the filtered synthesis weights over the five-year window. The
upper panel tracks the one-step log-score: the synthesis
stays low across the whole five years while BMA and stacking deteriorate, especially in the
2017 calm, where the synthesis NLL falls below $0.5$ while BMA exceeds $1.4$. The lower
panel shows the filtered weights reorganizing with the regime: in the sparse,
sector-structured 2017 calm the synthesis places most weight on the latent-geometry (GRDPG) agent, whose
weight rises toward $0.44$, while in the dense, factor-driven 2014--16 stress it shifts
toward the block and degree agents. The weights are learned per quarter, exactly as
Theorem~\ref{thm:t1} permits given the large per-snapshot edge counts, and they move as the
market regime moves.

\begin{figure}[t]
\centering
\includegraphics[width=\textwidth]{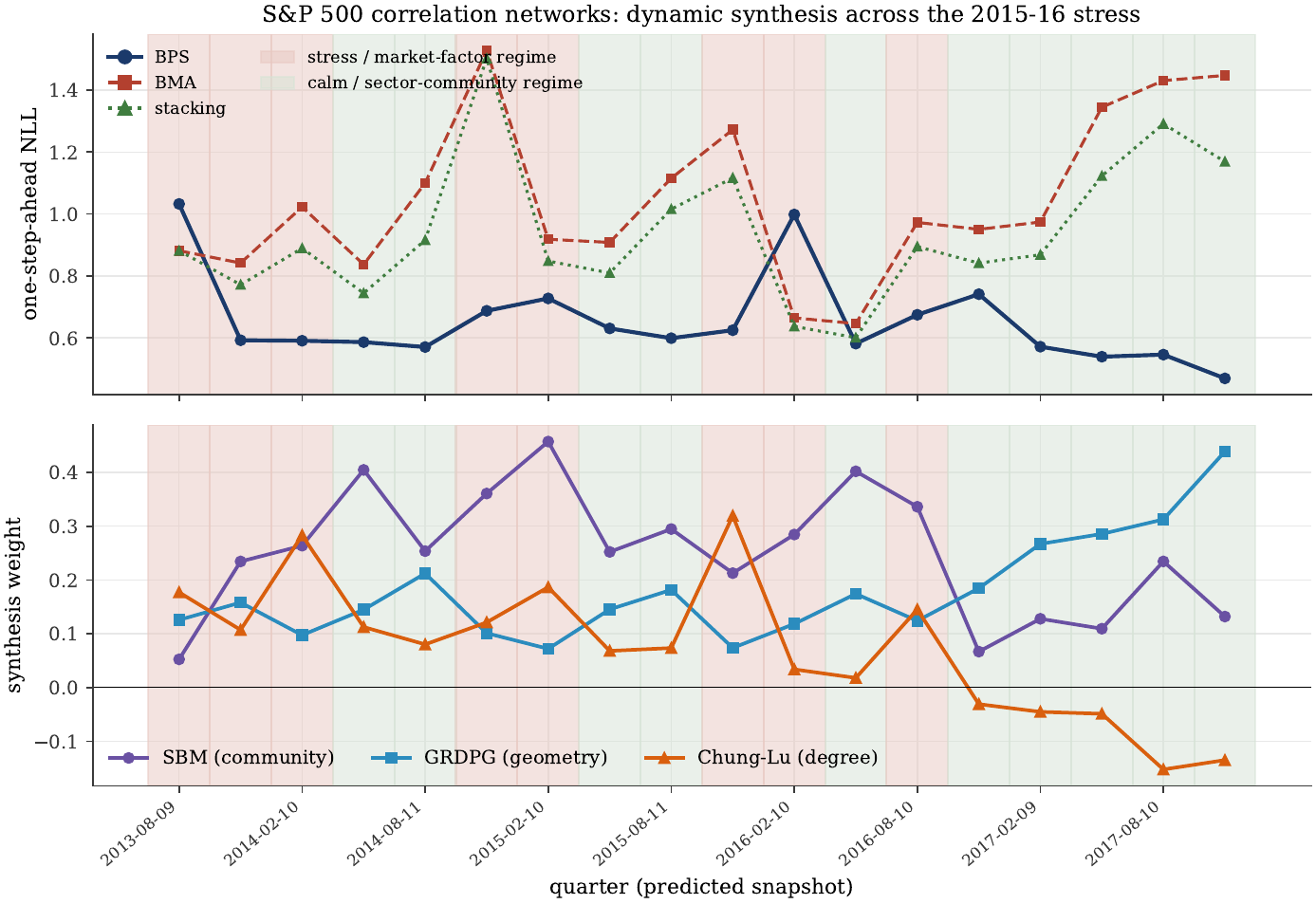}
\caption{S\&P 500 correlation networks, 2013--2018. (Top) One-step-ahead predictive
log-score; the synthesis stays low across the calm-to-crisis-to-calm cycle while BMA and
stacking deteriorate. (Bottom) Filtered synthesis weights reorganizing with the market
regime: latent geometry in the sparse 2017 calm, blocks and degree in the dense 2014--16
stress. Shading marks the data-derived regime (market-mode share above $0.35$).}
\label{fig:sp}
\end{figure}

Figure~\ref{fig:lpdr} shows the cumulative one-step-ahead log predictive score of the synthesis
relative to each combiner, accumulated quarter by quarter
\citep{mcalinn2019dynamic,mcalinn2020multivariate}. The synthesis has a positive cumulative advantage by the end of the sequence, $7.1$ nats over model
averaging and $5.2$ over stacking, though individual quarters can favor a competitor.

\begin{figure}[t]
\centering
\includegraphics[width=0.74\textwidth]{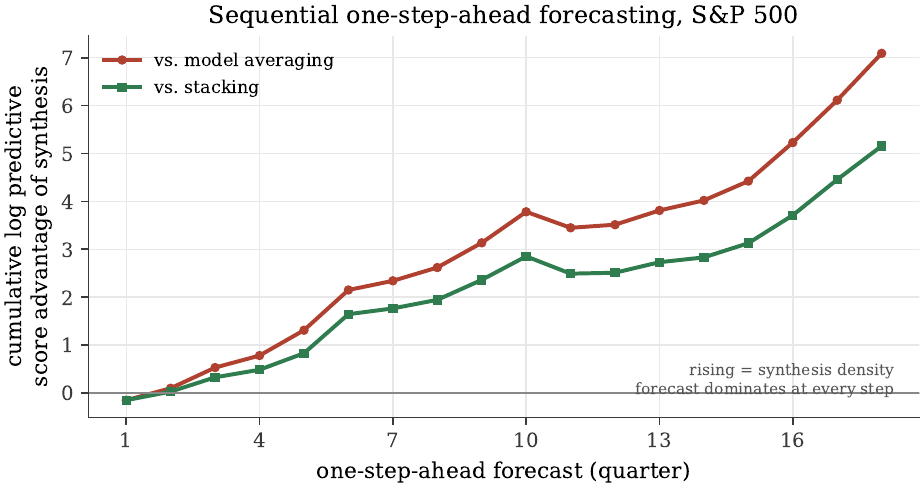}
\caption{Sequential one-step-ahead forecast performance on the S\&P network, in the form used to
compare density forecasts in Bayesian predictive synthesis \citep{mcalinn2020multivariate}: the
cumulative log predictive score of the synthesis relative to model averaging and to stacking,
accumulated over the eighteen quarterly forecasts. A rising curve means the synthesis places higher
predictive density on the realized snapshot at that step; upward segments favor the synthesis and
downward segments the comparator.}
\label{fig:lpdr}
\end{figure}

\label{sec:recalib}
We test whether a competitor can match the synthesis by recalibration. Splitting the forecast rounds
into a calibration fold (earlier snapshots) and an evaluation fold (later snapshots), we fit isotonic
regression on each competitor's calibration-fold predictions and apply it on the evaluation fold.
Recalibration sharply improves the latent-space competitor: its evaluation-fold ECE falls from
$0.244$ to $0.047$, its NLL from $0.832$ to $0.559$, and its $r=2$ decision cost from $0.569$
to $0.401$, while the already-calibrated synthesis improves more modestly (ECE $0.106$ to $0.049$). The
recalibrated competitor then matches the synthesis on every forecast score. The synthesis is
therefore not uniquely the best forecaster; its advantage is that it delivers calibration together
with the mechanism weights, their intervals
(Theorem~\ref{thm:crossfold-orthogonal-unconditional-bps}), and the localization of
Theorem~\ref{thm:t3}, none of which a recalibrated embedding returns.

Refitting the synthesis with the calibration intercept constrained to zero leaves the dense S\&P
scores essentially unchanged (NLL $0.653$ to $0.646$) but materially degrades the forecast on the sparse Bitcoin-OTC network
of the Supplement, so the free calibration layer matters most when the graph is
sparse. Fixing the weights ($\delta=1$) leaves the average forecast scores unchanged, so the
time-varying weights are needed for tracking and localization, not for an average-score gain on a
stable panel.

The dynamics also help estimation, in the sparse regime where the single-snapshot theory is weakest.
In a stable regime the discounted filter pools recent snapshots to an effective information $\asymp
N_t/(1-\delta)$, lowering the variance of the weight estimate below an independent single-snapshot
fit. On a stable three-block model held fixed across snapshots, the filtered community-weight estimate
has about half the standard deviation of the single-snapshot fit (Figure~\ref{fig:pooling}): at
$N=n^2\rho_n\approx250$ the single-snapshot standard deviation is $0.25$ and the pooled one $0.11$, and
both shrink as the snapshots densify. The time-varying machinery therefore restores estimable weights
where one snapshot is too sparse, and adds little where each snapshot is already dense.

\begin{figure}[t]
\centering
\includegraphics[width=0.6\textwidth]{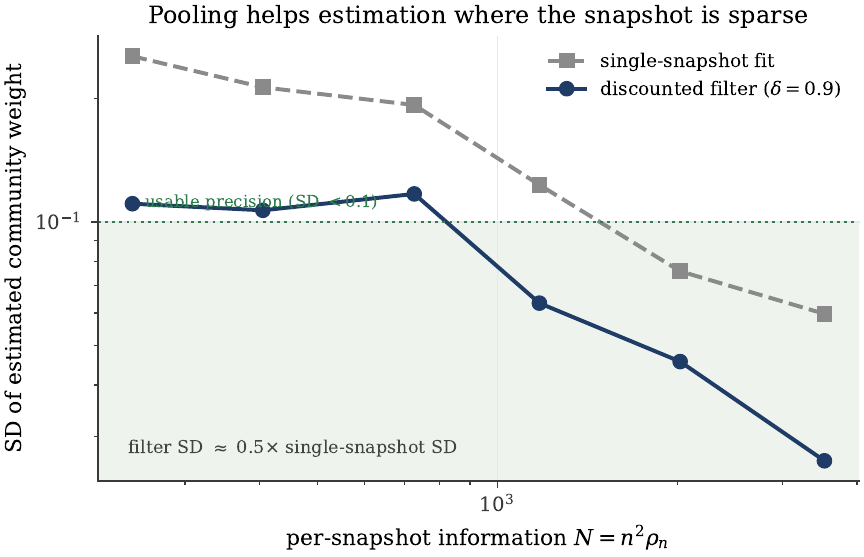}
\caption{Pooling helps estimation where the snapshot is sparse. On a stable three-block model held
fixed across $T=12$ snapshots, the across-replication standard deviation of the estimated community
weight is plotted against the per-snapshot information $N=n^2\rho_n$ for the discounted filter
($\delta=0.9$) and an independent single-snapshot fit. The filter standard deviation is about half
the single-snapshot value throughout; the absolute gap is decisive in the sparse regime (small $N$),
where the single-snapshot estimate exceeds the usable-precision band, and closes as the snapshots
densify.}
\label{fig:pooling}
\end{figure}

\label{sec:localize}
By Theorem~\ref{thm:t3}, the filtered weights localize the active mechanism, which is testable
against an external label. From the same daily returns, but independent of the network construction,
we compute the realized volatility of the equal-weight market index in each quarter, a standard
stress indicator. It agrees with the network's top-eigenvalue share on the regime (correlation
$0.93$), and both peak in the 2015--16 stress episode. The volatility series, computed from the returns, is the external validator; the eigenvalue share, read off the correlation
matrix itself, is an internal consistency check. The synthesis weights track this regime: the
latent-geometry weight is anticorrelated with stress ($-0.78$ against the external volatility, $-0.91$
against the internal eigenvalue share) and is largest in the five calm quarters of late 2016--2017; in stress the weight
reallocates to the block and degree agents, the block weight leading the most stressed quarters and
the degree weight rising most strongly with volatility ($+0.66$). The synthesis recovers the regime
with no knowledge of the volatility series.

Holding graph density fixed at $0.10$ in every quarter, by thresholding each correlation matrix at the
quantile that yields that density, leaves the synthesis advantage and the weight reallocation intact
(NLL $0.622$ against $1.14$ for averaging; community weight $0.40$ in stress against $0.27$ in calm).

\subsection{Two contact networks with known group labels}
\label{sec:hs}

We use two dynamic networks that carry known group labels, in different domains. The first is the SocioPatterns high-school face-to-face contact
network \citep{fournet2014contact}: wearable sensors recorded close-proximity contacts among
$327$ students of nine classes over five school days in December 2013 at twenty-second resolution.
The second is the SocioPatterns hospital-ward network \citep{vanhems2013estimating}: contacts among
$75$ doctors, nurses, administrative staff, and patients in a Lyon ward over four days, carrying the
four role labels. Aggregating each into hourly snapshots on its fixed node set gives $T=41$ and
$T=44$ snapshots in the sparse regime, on which we run the same one-step-ahead protocol and
competitors as above. The class and role labels are held out from every agent and used only for
validation.

On both networks the synthesis is the best-calibrated forecaster by a wide margin
(Table~\ref{tab:forecasts}): PIT-KS $0.107$ and ECE $0.140$ on the high school and
$0.119$ and $0.139$ on the hospital, against $0.29$ to $0.49$ for the convex combiners, which once
again remain badly miscalibrated. Against model averaging the synthesis
improves log-score by $1.45$ on the high school ($95\%$ CI $[1.31,1.57]$) and $1.05$ on the hospital
($[0.90,1.20]$).

The mechanism weights, which no competitor returns with calibrated uncertainty, are validated against the labels and identify a
\emph{different} mechanism in each domain (Figures~\ref{fig:hs} and \ref{fig:hosp}). On
the high school the community (SBM) weight is largest, near $0.42$ with an interval that separates from
the others at the settled snapshots; the community mechanism is finite-dimensional, so this is the
unconditionally valid interval of Theorem~\ref{thm:crossfold-orthogonal-unconditional-bps}, and the high
school is the case in which the leading mechanism carries the unconditional band. The aggregate contact graph is assortative by class (attribute
assortativity $0.65$ on the held-out labels), with the community agent's blocks recovering the nine
classes almost exactly, adjusted Rand index $0.993$ and normalized mutual information $0.994$,
computed without the labels; within the school day the community weight rises and falls
with the within-class share of contacts (correlation $0.49$, $p=0.0015$), tracking the daily
alternation between within-class lessons and between-class mixing. On the hospital the largest
weight is instead the degree (Chung--Lu) mechanism, near $0.50$, whose leading weight is infinite-dimensional and so carries the conditional interval, and the community weight is small.
The ward is not assortative by role: its attribute assortativity on the held-out labels is slightly
negative ($-0.12$), so edges run across roles, not within them and there is no community
structure for the SBM to lead on. What organizes the ward is contact volume, and the role labels
confirm it, the staff being the hubs (mean aggregate degree $39$ for nurses and $36$ for doctors
against $21$ for patients). The synthesis thus reports community structure where it exists and hub
structure where that is what dominates, in each case the reading an external label endorses.

These readings, with the S\&P regime tracking of Section~\ref{sec:sp}, are collected in
Table~\ref{tab:weight-validation}: the same synthesis names a different mechanism in each network, a
held-out signal the agents never see endorses each one, and no competitor returns a comparable mechanism weight with uncertainty.

\begin{table}[t]
\centering
\small
\caption{The synthesis weights identify a different mechanism in each network, each confirmed by a
held-out signal the agents never see. No competitor (Bayesian model averaging, stacking, equal-weight
pooling) returns mechanism weights.}
\label{tab:weight-validation}
\begin{tabular}{@{}l p{0.22\textwidth} p{0.42\textwidth}@{}}
\toprule
Network ($n$) & Leading mechanism & Held-out validation \\
\midrule
High school (327) & community / SBM ($\bar\theta{=}0.42$) & assortative by the $9$ classes ($0.65$); blocks recover the classes (ARI $0.993$, NMI $0.994$); weight tracks within-class contact share ($r{=}0.49$, $p{=}0.0015$) \\[3pt]
Hospital ward (75) & degree / Chung--Lu ($\bar\theta{=}0.50$) & not assortative by the $4$ roles ($-0.12$); staff are the hubs (mean degree: nurses $39$, doctors $36$, vs patients $21$) \\[3pt]
S\&P 500 (470) & geometry / GRDPG, regime-tracking & realized volatility (external): geometry weight $r{=}-0.78$ ($-0.91$ vs eigenvalue share); block and degree rise in stress, degree $r{=}+0.66$ \\
\bottomrule
\end{tabular}
\end{table}

\begin{figure}[t]
\centering
\includegraphics[width=\textwidth]{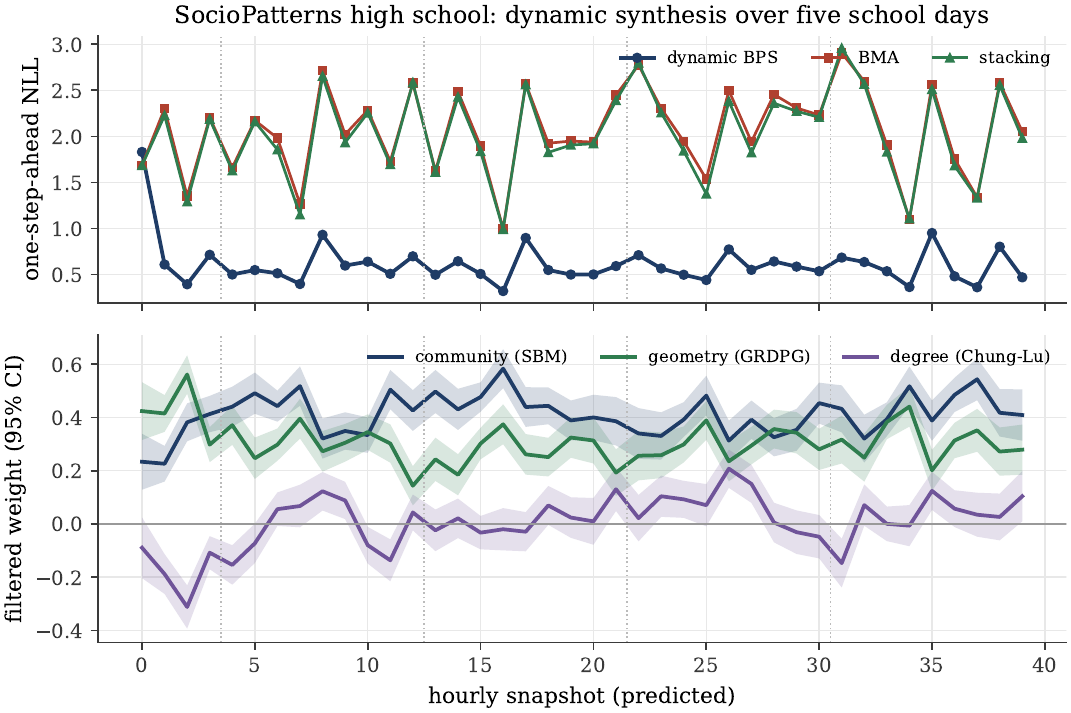}
\caption{SocioPatterns high-school contact network, five school days
(Theorem~\ref{thm:crossfold-orthogonal-unconditional-bps}). Top: one-step-ahead log-score; the synthesis stays low and calibrated
while model averaging and stacking are far higher. Bottom: the mechanism weight trajectories with $95\%$ bands that are
unconditional (Theorem~\ref{thm:crossfold-orthogonal-unconditional-bps}) for the leading
finite-dimensional community weight and conditional on the fitted features
(Theorem~\ref{thm:t1}) for the geometry and degree weights; the community (SBM) mechanism leads, consistent with its near-perfect recovery
of the nine classes (adjusted Rand index $0.993$, normalized mutual information $0.994$). Dotted
lines mark school-day boundaries.}
\label{fig:hs}
\end{figure}

\begin{figure}[t]
\centering
\includegraphics[width=\textwidth]{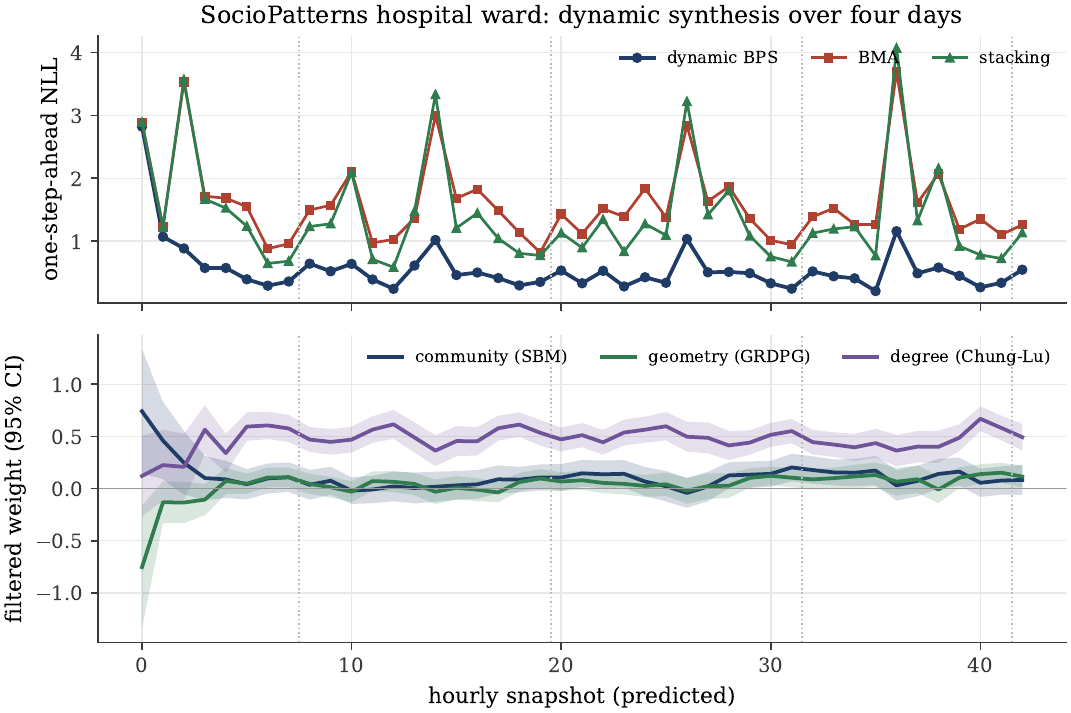}
\caption{SocioPatterns hospital-ward contact network, four days (Theorem~\ref{thm:crossfold-orthogonal-unconditional-bps}). Top:
one-step-ahead log-score; the synthesis stays low and calibrated. Bottom: filtered mechanism weights
with $95\%$ bands conditional on the fitted features (Theorem~\ref{thm:t1}); the unconditional
guarantee of Theorem~\ref{thm:crossfold-orthogonal-unconditional-bps} covers the finite-dimensional
community weight, not the degree weight that is largest here, so the band shown for the leading
mechanism is the conditional interval. The degree (Chung--Lu) mechanism leads, reflecting a ward
organized around high-contact staff (mean aggregate degree $39$ for nurses and $36$ for doctors against $21$
for patients), not assortative communities, which the community weight correctly does not
find. Dotted lines mark day boundaries.}
\label{fig:hosp}
\end{figure}

\subsection{Controlled simulations}
\label{sec:controlled}

\label{sec:e1}
We simulate $T=12$ snapshots on $n=300$ nodes in three regimes of four snapshots each: a
community regime (a three-block SBM), a hub regime (a Chung--Lu power-law degree model),
and a geometry regime (a two-dimensional GRDPG). This is a controlled setting in which the
data-generating mechanism changes abruptly, and it is the natural test of whether the
synthesis can track. Table~\ref{tab:e1} reports averages over the eleven one-step
forecasts.

\begin{table}[t]
\centering
\caption{Regime-switching simulation ($n=300$, $T=12$, three regimes). Averages over the
eleven one-step-ahead forecasts.}
\label{tab:e1}
\begin{tabular}{lccc}
\toprule
Method & NLL & PIT-KS & ECE \\
\midrule
\textbf{Dynamic BPS} & \textbf{0.742} & \textbf{0.074} & \textbf{0.108} \\
BMA            & 1.443 & 0.401 & 0.433 \\
Stacking       & 1.383 & 0.391 & 0.426 \\
Equal weight   & 1.378 & 0.389 & 0.426 \\
SBM agent      & 1.427 & 0.403 & 0.432 \\
GRDPG agent    & 2.338 & 0.365 & 0.417 \\
Chung--Lu agent& 1.430 & 0.400 & 0.431 \\
\bottomrule
\end{tabular}
\end{table}

The synthesis roughly halves the NLL of every competitor (NLL $0.74$
versus $1.38$--$2.34$) and improves calibration by a factor of five (PIT-KS $0.074$ versus
$0.39$--$0.40$). The
per-regime synthesis log-scores are stable ($0.83$ community, $0.70$ hub, $0.71$ geometry),
showing that no single regime is responsible for the gain.

The filtered log-score and weights reflect the regime structure. The one-step log-score spikes at each
regime change (snapshots $t=4$ and $t=8$) and re-adapts within the regime, whereas BMA and
stacking remain high throughout, and the filtered weights reallocate in step: in the community
regime the SBM weight is large ($\approx 0.9$); at the switch to the hub regime the Chung--Lu
weight rises from $0.12$ to $0.75$ while SBM is discounted; and the weights move again at the
geometry switch.

Theorem~\ref{thm:t3} sharpens this into a localization claim: when a mechanism is active with
a positive margin, the largest filtered weight should name it. Scoring the argmax of the
filtered structural weights against the active regime across the eleven forecasts, the
synthesis localizes correctly in $6$ of $11$ snapshots, and the pattern is exactly what the
margin condition predicts. In the community regime the SBM weight is largest at all three
forecasts ($3/3$); in the hub regime Chung--Lu is largest at three of four, the single miss
being the first snapshot after the switch, which is the recovery delay $h$ of the theorem
made visible. In the geometry regime the GRDPG weight is not the largest in any of the four geometry-regime forecasts. The
latent-geometry signal in this simulation is weakly separated from degree heterogeneity, so the
margin $\kappa$ is small and the localization condition of Theorem~\ref{thm:t3} fails. This is consistent with the theorem, which guarantees localization only above the margin; and the design-conditioning
diagnostic of Section~\ref{sec:e_t4} is what tells a practitioner in advance which regimes
are separable. Calibration, by contrast, does not depend on the margin and remains strong
throughout (PIT-KS $0.074$).

Theorem~\ref{thm:t3} guarantees localization only above the margin, so the geometry failure above
should be a small-margin artifact, not a structural limit. We confirm this with the same
three-regime design at a large, well-separated per-dyad contrast ($\kappa=2.2$), with the three
signals made mutually distinct: assortative blocks for community, latent positions on the unit
sphere for geometry (unit norm, so the geometry signal carries no degree information and the earlier
aliasing is removed), and an expected-degree product for degree. Over thirty replications, at the
settled snapshot of each regime the largest synthesis weight lands on the truly active mechanism in
every case, a localization hit rate of $1.00$ for all three mechanisms, geometry included
(Figure~\ref{fig:wellsep}); the only misses fall in the one or two snapshots immediately after a
switch, the recovery delay $h$ of the theorem, which lowers the all-snapshot rate to $0.58$ and
$0.60$ for the two regimes that follow a switch. Localization therefore succeeds for every mechanism
once its margin clears the threshold, exactly as Theorem~\ref{thm:t3} predicts, and the weak-margin
geometry failure is the predicted small-$\kappa$ regime, not a limitation of the synthesis.

\begin{figure}[t]
\centering
\includegraphics[width=0.62\textwidth]{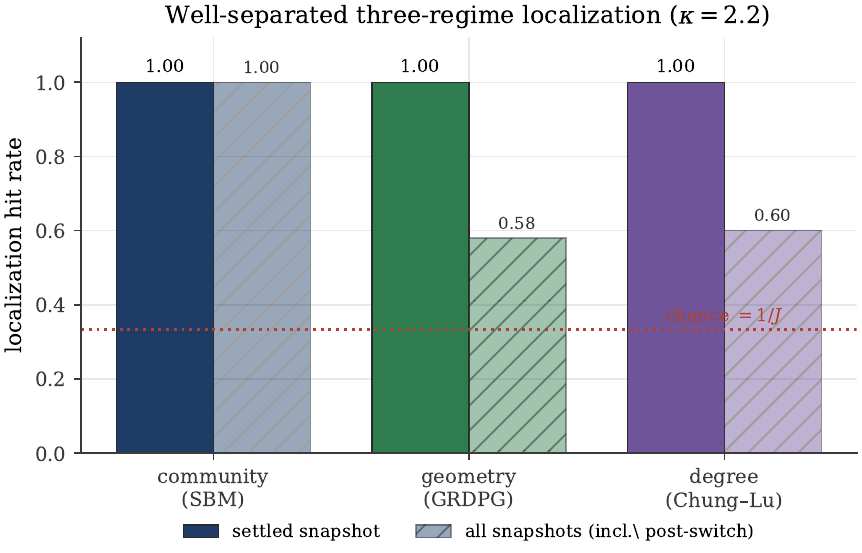}
\caption{Well-separated three-regime localization ($\kappa=2.2$, $n=180$, thirty replications). At
the settled snapshot of each regime the largest synthesis weight lands on the active mechanism in
every replication (hit rate $1.00$ for community, geometry, and degree, solid bars); the all-snapshot
rate (hatched) is lower only because of the one-to-two-snapshot recovery delay after each switch. The
weak-margin geometry failure in the preceding three-regime simulation is thus the small-margin regime of
Theorem~\ref{thm:t3}, not a structural limitation of the synthesis.}
\label{fig:wellsep}
\end{figure}

\label{sec:inertia}
Remark~\ref{rmk:inertia} makes the gap over averaging rules quantitative: after a regime of length $L$, vanilla Bayesian
model averaging has accumulated $\Theta(L)$ log-odds of evidence for the old mechanism and
needs an $\Omega(L)$-long window to reverse, whereas the discounted synthesis forgets in
$O(1/(1-\delta))$ snapshots independent of $L$. Dynamic model averaging \citep{raftery2010online} introduces a forgetting factor to reduce this
inertia, so we include it with forgetting matched to the synthesis discount, alongside vanilla
averaging. A network switches once from a community regime (an SBM) to a hub regime (a
Chung--Lu degree model) after $L$ snapshots, and we record, for each method, the recovery
delay, the number of post-switch snapshots until its largest combination weight points to the
new mechanism. Varying $L$ from $2$ to $16$ and averaging over eight replications,
Figure~\ref{fig:theorypanel} (center) reports the recovery delay. The synthesis recovers in a constant single
snapshot at every $L$ (delay slope $0.00$). Vanilla averaging takes longer the longer the old
regime lasted, its delay rising from near zero at $L=2$ to nineteen snapshots at $L=16$ (slope
$1.39$), one extra snapshot of lag per snapshot of history. Dynamic model averaging falls
between the two: its forgetting factor bounds the inertia, so the delay plateaus near six
snapshots, not growing without limit (slope $0.43$), but it still does not match the
synthesis, because it remains an averaging rule that must accumulate $O(1/(1-\delta))$ snapshots
of new-regime evidence to overcome the evidence it has retained, whereas the synthesis carries
an unconstrained logit state with an intercept and relocates at once.

\label{sec:regret}
Theorem~\ref{thm:t3} bounds the filter's cumulative one-step-ahead regret against the per-snapshot
oracle by \eqref{eq:regret}, whose only series-growing term is the comparator movement, itself
controlled by the number of switches. At a fixed horizon the cumulative regret rises monotonically
with the number of switches, from $0.68$ at none to $1.62$ at six (slope $0.16$ per switch, tracking
a movement budget that climbs from $0.5$ to $13$), while holding switches at zero and quadrupling the
horizon leaves it essentially flat ($0.69$ at $T=10$, $0.74$ at $T=44$): the filter charges per
switch, not per time point, exactly as \eqref{eq:regret} predicts.

\paragraph{The localization threshold.}
\label{sec:localize-threshold}
By Theorem~\ref{thm:t3}, the filter localizes the active mechanism only above a contrast threshold:
the localization delay is of order $\log J/(N_t\kappa^2)$, so a settled regime is identifiable once
$\kappa$ exceeds $\sqrt{\log J/N_t}$. We confirm this on a synthetic sequence that alternates a
community and a degree regime at a controlled contrast $\kappa$, with the geometry agent present as a
distractor. The settled localization hit rate is near the chance level $1/J$ for small $\kappa$ and
reaches one for $\kappa\ge0.75$. The crossing lies a constant factor above the information floor
$\sqrt{\log J/(N_t\rho(1-\rho))}\approx0.11$, the one-snapshot agent-estimation term $a_{n,t}$ of
\eqref{eq:regret}. The per-snapshot average saturates lower, near $0.63$, because each switch costs a
snapshot or two of relocalization.

\paragraph{The dyadic rate.}
\label{sec:e_t1}
Theorem~\ref{thm:t1} predicts that the synthesis weights are estimable from a single snapshot with
standard error $(n^2\rho_n)^{-1/2}$. Subsampling induced subgraphs of the densest S\&P snapshot
(November 2015, $55{,}214$ edges) and refitting the one-snapshot synthesis on each, the empirical
standard error of the fitted weights falls along the dyadic rate (Figure~\ref{fig:theorypanel}, left):
a fitted log-log slope of $-0.535$ against the theoretical $-1/2$, the weight standard deviation
shrinking from $0.18$ to $0.018$ as the edge information grows sixty-six-fold. The same decay appears
on the full snapshots, where the latent-geometry weight is pinned to standard error $0.03$ on the
stressed November 2015 graph but carries errors five to twenty times larger on the sparse November
2017 graph ($1{,}986$ edges). The rate is a property of one graph, not of the series.

\begin{figure}[t]
\centering
\includegraphics[width=\textwidth]{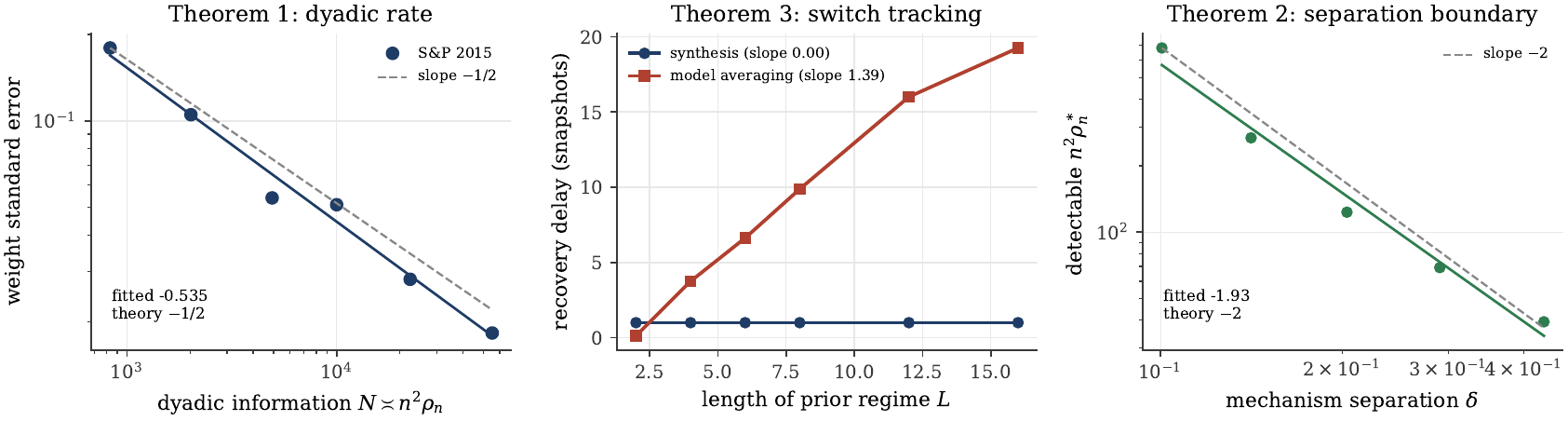}
\caption{The single-snapshot theory confirmed on the same data. (Left) Theorem~\ref{thm:t1}: the
weight standard error falls along the dyadic rate, fitted slope $-0.535$ against the predicted
$-1/2$ (S\&P 2015 stress snapshot, log-log against subsampled dyadic information
$N\asymp n^2\rho_n$). (Center) Theorem~\ref{thm:t3}: the recovery delay after a switch is flat in
the prior-regime length $L$ for the synthesis (slope $0.00$) while model averaging grows linearly
(slope $1.39$). (Right) Theorem~\ref{thm:t2}: the detectable dyadic information $n^2\rho_n^\ast$
scales as $\delta^{-2}$ in the separation $\delta$ (fitted $-1.93$), the boundary below which two
mechanisms cannot be distinguished.}
\label{fig:theorypanel}
\end{figure}

\paragraph{Coverage and normality of the weights.}
\label{sec:e_t1b}
The rate is one half of Theorem~\ref{thm:t1}; the other half is the conditional normal limit
\eqref{eq:t1clt} with the sandwich variance, which is what licenses confidence intervals for
the weights. We test it in a controlled design where the population weight is known. Two
fixed agent surfaces, a three-block community logit and a degree logit, are combined at a
known weight $\bbeta^\circ=(\alpha,\theta_{\mathrm{comm}},\theta_{\mathrm{deg}})$ with a
$\log\rho_n$ offset, and we sample many graphs from the resulting edge law, refit the dyadic
logistic regression on each, and form the inverse-information Wald intervals. Over $400$
replications at $n=300$ ($N_t\approx3{,}900$), the empirical coverage of the nominal $95\%$
intervals is $96.3\%$ for the intercept, $95.8\%$ for the community weight, and $95.8\%$ for
the degree weight, and the standardized estimates $(\widehat\beta-\beta^\circ)/\mathrm{se}$
are indistinguishable from a standard normal (Kolmogorov--Smirnov statistics $0.049$, $0.022$,
$0.029$). The standard error of the community weight falls along $N_t^{-1/2}$ with a fitted
log-log slope of $-0.500$, the dyadic rate exactly. The intercept here is centered at $\log\rho_n+\alpha$, a direct illustration of
the sparsity offset: without it the intercept silently absorbs the density. With the agent
surfaces known, the conditional limit of Theorem~\ref{thm:t1} holds with valid coverage at the
predicted rate.

This coverage is conditional on the agents, and Theorem~\ref{thm:t1} is explicit that the
unconditional limit requires $\sqrt{N_t}\,a_{n,t}\to0$, a condition the SBM agent rate violates
in the denser part of the sparsity regime. We measure the cost directly. Refitting the
community and degree agents on each sampled graph and forming the same naive inverse-information intervals, the coverage of the community weight
falls below nominal and the shortfall grows with the edge information: from $82\%$ at
$N_t\approx2{,}700$ to $71\%$ at $N_t\approx43{,}000$, while the standardized bias rises from
$0.5$ to $1.1$ standard errors over the same range. The fixed-agent coverage stays at nominal
throughout, so the gap is entirely the agent-estimation error entering the weight limit at
first order. The naive single-snapshot intervals are thus valid when the agents are known or
estimated at a negligible rate and undercover when a flexible agent is refit on the same dense
graph; the cross-fold orthogonal estimator of Theorem~\ref{thm:crossfold-orthogonal-unconditional-bps} closes this gap, with an
unconditional limit valid under an $o(N_t^{-1/4})$ agent rate.

\paragraph{The cross-fold orthogonal correction.}
\label{sec:e_t1c}
Theorem~\ref{thm:crossfold-orthogonal-unconditional-bps} asserts that a cross-fold Neyman-orthogonal score restores an
unconditional normal limit, and valid intervals, under a mild agent rate; we confirm this on a
design where the orthogonal correction is exactly constructed. A known three-block community
surface carries the target weight, while the degree mechanism is the refit nuisance, estimated
on each graph by a regularized Fourier sieve in a node-pair covariate, and the blocks carry
distinct degree levels, so the target and the nuisance are correlated and a mis-estimated
nuisance contaminates the target. The naive plug-in, which treats the fitted sieve as known and
inverts the second-stage information, undercovers the community weight: the empirical coverage
of the nominal $95\%$ interval lies between $82\%$ and $94\%$, with a standardized bias as large
as one standard error where the nuisance bias is strongest. The cross-fold orthogonal estimator,
which partials the target against the nuisance basis out of fold and forms the orthogonal
sandwich, holds nominal coverage of $96\%$ to $98\%$ with a standardized bias below $0.15$ across
the whole range of edge information, from $N_t\approx2{,}700$ to $N_t\approx43{,}200$. The
orthogonal coverage neither drifts nor degrades as the graph grows, which is the unconditional
limit of Theorem~\ref{thm:crossfold-orthogonal-unconditional-bps} made visible: the debiasing yields valid confidence
statements for an individual mechanism weight at the price of one cross-fold construction.

\paragraph{The separation rate.}
\label{sec:e_t2}
Theorem~\ref{thm:t2} predicts a separation rate $\Delta_n\asymp(n^2\rho_n)^{-1/2}$ in graphon
log-odds, the boundary between distinguishable and indistinguishable mechanisms. Sweeping the log-odds
contrast $\delta$ and locating the boundary $n^2\rho_n^\ast$ at which power crosses $0.9$, the
boundary falls from $683$ to $39$ as $\delta$ grows from $0.10$ to $0.43$ (Figure~\ref{fig:theorypanel},
right), a fitted log-log slope of $-1.93$ against the predicted $-2$. Every real snapshot, the S\&P
graphs ($1{,}986$ to $55{,}214$ edges) and the co-occurrence graphs ($700$ to $1{,}000$), lies well
inside the detectable region, the boundary sitting at vanishing contrast within the sparsity regime. The matching impossibility below the rate is
Theorem~\ref{thm:t2} itself, by the sparse-logit lemma and the two-point method.

\paragraph{Dominance over model averaging under misspecification.}
\label{sec:e_t3}
Under misspecification a coherent Bayesian model-averaging posterior concentrates on the single
Kullback--Leibler-closest agent as edge information accumulates, its weight vector tending to
$(1,0,0)$: this is Berk's collapse, the correct behavior of a posterior but the wrong default for
combination, since it discards the other agents instead of reweighting them. The synthesis instead
retains and reweights all agents, and Theorem~\ref{thm:t4} makes the resulting log-score gap
\eqref{eq:gap} strictly positive whenever the truth is not a single agent. In this misspecified
setting the gap is positive on every study, ranging from $0.39$ on the S\&P network to $3.33$ on
the power-grid mesh; Figure~\ref{fig:t4} in the Supplement plots the gap and the entropy collapse.

\paragraph{Mechanism aliasing.}
\label{sec:e_t4}
By Theorem~\ref{thm:t1}, two mechanisms are indistinguishable from a snapshot exactly when the
dyad-feature information matrix $H_t$ is singular. Its condition number $\kappa(H_t)$ is an observable
diagnostic of how close two mechanisms are to aliasing. The controls behave as predicted: a degree-corrected block
model, whose community and degree agents are distinct, is far better conditioned
($\kappa\approx340$) than the near-aliased control, whereas a structureless Erd\H{o}s--R\'enyi graph, where both agents
collapse to a constant and alias, is roughly six times worse. The real networks span four
orders of magnitude, from the dense S\&P stress snapshot ($\kappa\approx6$), where distinct
community, geometry, and degree signals make the weights sharply separable and underwrite the
precise estimates of Section~\ref{sec:e_t1}, to the extremely sparse ca-GrQc
($\kappa\approx7.8\times10^4$), which approaches aliasing, so that individual weights there are
only weakly identified even where the pooled forecast is accurate.

\subsection{Large dynamic networks}
\label{sec:scale}

The networks so far have a few hundred to a few thousand active nodes. To probe scale and
out-of-sample behavior we apply the same synthesis model, with the data construction and evaluation adapted for scale (symmetrization, balanced dyad sampling, and recent-history pruning, described below), to two
large public timestamped networks \citep{rossi2015network} on their entire node sets: the arXiv
HEP-PH citation network \citep{leskovec2007graph} ($n=28{,}093$, $T=85$ monthly snapshots), the
Enron e-mail network \citep{klimt2004enron} ($n=86{,}978$, $T=46$). The HEP-PH citation and
Enron e-mail logs are natively directed; because the model and theory of
Section~\ref{sec:model} are defined for undirected graphs, we symmetrize each to fit the undirected
observation model, placing an undirected edge between two nodes joined by at least one citation,
message, or interaction within a snapshot and discarding the directional signal. The full dyad set is too large to score at these node counts. Each forecast is therefore evaluated on
a balanced dyad sample of all edges and an equal number of randomly sampled non-edges. The sparse offset maps the sample probabilities to the population scale as in
Corollary~\ref{cor:t1}, and all methods are scored on this common balanced sample. The three
mechanism agents are fit on the discounted-average adjacency through $t-1$, so they are the dynamic block,
latent-geometry, and degree predictors used as competitors in the S\&P study
(Section~\ref{sec:sp}) and match the history-based agent definition of Section~\ref{sec:model}.
At this scale the discounted-average adjacency is pruned to its recent-history window so that the
largest graph fits in memory.

Two findings carry over to this scale (Table~\ref{tab:scale}, Figure~\ref{fig:scale}). First, the
synthesis is the best-calibrated method on both networks, with PIT deviation $0.031$ to $0.061$
and calibration error $0.041$ to $0.077$, against $0.41$ to $0.50$ for every classical combiner and
single agent; on the log-score it is lowest among the compared methods on both networks, by
$2.5$ to $3.5$ nats per dyad. Second, the weights remain tightly estimated at this scale: they favor degree and community
on the citation network, with the geometry weight near zero, and are positive across all three mechanisms on the e-mail network.
The two networks fall in different regimes: on the citation network the synthesis is calibrated and predictive, whereas the e-mail
network sits in the near-base-rate regime of Corollary~\ref{cor:t1}, where the edge probabilities are nearly uniform and the synthesis delivers calibration without discrimination.

\begin{table}[t]
\centering
\caption{Two large public timestamped networks, run with the same one-step-ahead pipeline on their
entire node sets; the three agents are fit on the discounted-average adjacency (the dynamic
predictors). Averages over the one-step-ahead
forecasts. Boldface marks the best in each column. The synthesis is the best-calibrated method on both networks and attains the
lowest log-score among the classical combiners and single agents compared here.}
\label{tab:scale}
\small
\begin{tabular}{lccc}
\toprule
Method & NLL & PIT-KS & ECE \\
\midrule
\multicolumn{4}{l}{\textit{arXiv HEP-PH} ($n=28{,}093$, $T=85$)}\\
Dynamic BPS & \textbf{0.413} & \textbf{0.031} & \textbf{0.041} \\
BMA & 3.032 & 0.457 & 0.486 \\
Stacking & 2.937 & 0.448 & 0.484 \\
Equal weight & 3.081 & 0.449 & 0.487 \\
SBM agent & 3.030 & 0.457 & 0.486 \\
GRDPG agent & 3.855 & 0.424 & 0.479 \\
Chung--Lu agent & 3.945 & 0.470 & 0.495 \\
\midrule
\multicolumn{4}{l}{\textit{Enron e-mail} ($n=86{,}978$, $T=46$)}\\
Dynamic BPS & \textbf{0.597} & \textbf{0.061} & \textbf{0.077} \\
BMA & 4.508 & 0.451 & 0.488 \\
Stacking & 4.125 & 0.447 & 0.488 \\
Equal weight & 4.163 & 0.457 & 0.491 \\
SBM agent & 4.672 & 0.495 & 0.499 \\
GRDPG agent & 4.510 & 0.431 & 0.479 \\
Chung--Lu agent & 4.615 & 0.474 & 0.496 \\
\bottomrule
\end{tabular}
\end{table}

\begin{figure}[t]
\centering
\includegraphics[width=\textwidth]{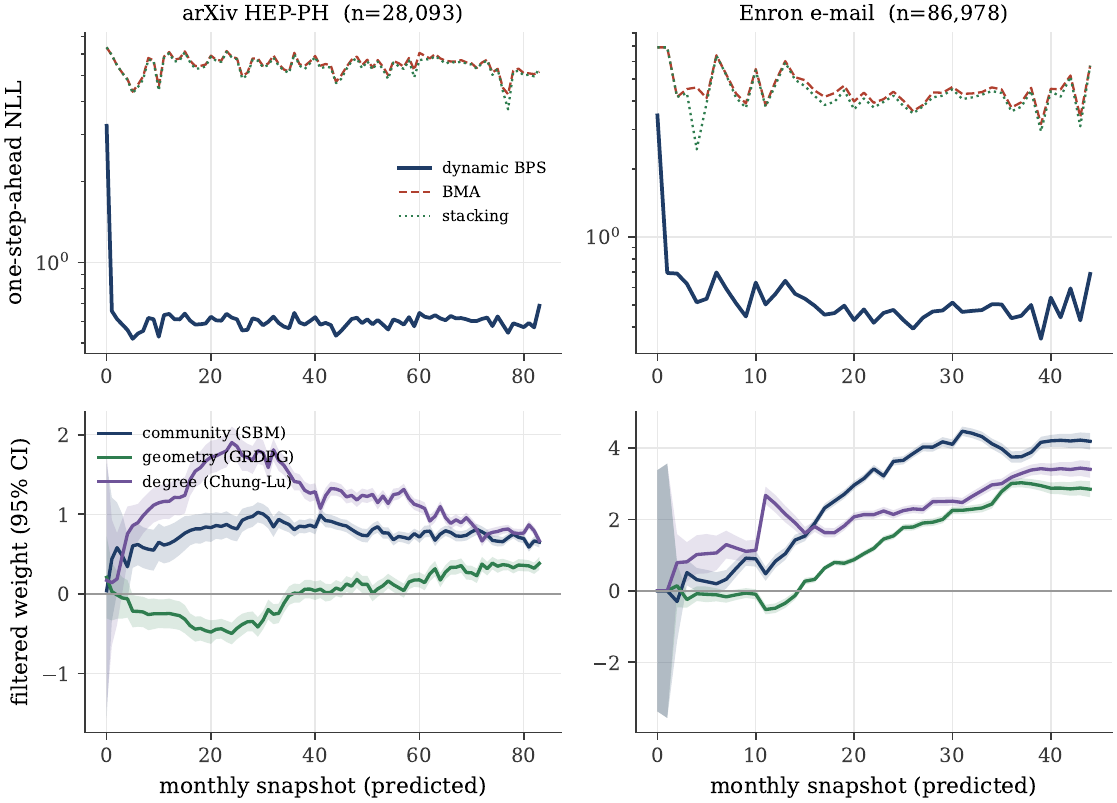}
\caption{Two large dynamic networks (arXiv HEP-PH, Enron), agents fit on the
discounted-average adjacency. Top: one-step-ahead log-score on a logarithmic scale; the synthesis stays near $1$ while
every classical combiner sits at $3$ to $5$. Bottom: filtered mechanism weights with $95\%$
confidence bands, estimable with tight intervals at these node counts; the degree and community weights lead on
the citation network and all three mechanisms carry positive weight on the e-mail network.}
\label{fig:scale}
\end{figure}

We also evaluate the one-step-ahead point forecasts by mean squared forecast error (MSFE$_{1:T}$, the Brier score for a binary edge on the balanced evaluation set of all edges and an equal-sized random sample of non-edges). The synthesis attains $0.112$ and $0.165$ on the citation and e-mail networks; the classical combiners lie near $0.47$ to $0.50$ on both, worse by $185\%$ to $345\%$, so the calibration advantage is also a point-forecast advantage.

\FloatBarrier
\section{Discussion and limitations}
\label{sec:discussion}

We developed dynamic Bayesian predictive synthesis for dynamic networks. Each generative mechanism,
such as community structure, latent geometry, or degree heterogeneity, is treated as an agent that
forecasts the next snapshot's edges, and a synthesis layer combines these forecasts with weights that
form the state of a dynamic linear model. The weights are the object of inference: the method
jointly estimates which mechanism is most informative at each time and produces a calibrated forecast,
specializing to calibrated link prediction at a single snapshot. Because a network on $n$
nodes has of order $n^2$ node pairs, one snapshot estimates the weights precisely even when the graph
is sparse, and faster than the mechanisms change. Conditional on the fitted agents the intervals hold
at the dyadic rate; for finite-dimensional agents a cross-fold orthogonal construction removes the
refitting bias and makes them unconditional, while for high-dimensional latent-position and degree
agents they remain conditional. Two mechanisms separate only above a sharp contrast threshold; a
change in the active mechanism is tracked at a cost paid only at the change, at a minimax-optimal
rate; and when no agent is correct the weight converges to the projection of the active mechanism onto
the set of agents while the forecast stays calibrated, where model averaging collapses onto a single
mechanism.

The empirical studies support these results. On real dynamic networks the estimated weights track
regime changes identified independently of the model; the synthesis remains calibrated where model
averaging collapses and pooling miscalibrates; and its advantage is largest at the regime changes
the theory identifies. The contribution is a single model that produces calibrated forecasts
without a separate recalibration step, together with interpretable time-varying mechanism weights
that carry conditional intervals, unconditional for finite-dimensional mechanisms, and identify mechanism changes. A forecaster that is not
mechanism-based can be recalibrated to match the probabilities, but does not provide the weights or
their uncertainty.

The results rest on two modeling choices: a finite, fixed set of agents, and an undirected binary
observation model with conditionally independent dyads. Enlarging the set of agents leaves the
filter unchanged but slows the estimation rate to $\{(n^2\rho_n)/J_t\}^{-1/2}$ and raises the
separation threshold. The conditionally independent dyad assumption holds only approximately for the
main-text networks, so the weight inference is validated on matched simulations and reported with a
dependence-robust variance.

We propose three extensions. First, weighted or count-valued edges replace the Bernoulli model with
a Gaussian or Poisson one. Directed and signed networks require asymmetric observation models and are
left to future work: the citation and e-mail networks studied here are natively directed and are
symmetrized to fit the undirected model, so a native-directed analysis, which must also model the
reciprocity-induced dependence between the two orientations of a dyad, is a genuine extension and not
a relabeling. The projection and tracking arguments extend to any observation model with a smooth
log-likelihood, but the directed case is not immediate. Second, the
Laplace filter can be replaced by particle filtering when a network snapshot is too sparse for the
Gaussian approximation. Third, the dynamic block-model and latent-space forecasters used here as
benchmarks \citep{xu2014dynamic,matias2017statistical,sewell2015latent,durante2014nonparametric} can
be folded into the set of agents.

\section*{Acknowledgments}
The authors thank colleagues for helpful discussions. Code to reproduce all results is
provided with this manuscript.

\section*{Disclosure Statement}
The authors report there are no competing interests to declare.

\section*{Data Availability Statement}
All networks analyzed in this article are publicly available benchmark datasets. The main text uses a
financial correlation network built from public daily S\&P 500 prices (Section~\ref{sec:sp}); the
SocioPatterns high-school and hospital-ward face-to-face contact networks
\citep{fournet2014contact,vanhems2013estimating}; and two large public timestamped networks
\citep{rossi2015network}, the arXiv HEP-PH citation network \citep{leskovec2007graph} and the Enron
e-mail network \citep{klimt2004enron}. The Supplement additionally analyzes the Bitcoin-OTC trust network
\citep{kumar2016edge,leskovec2016snap} and standard static link-prediction benchmarks, including the
ca-GrQc and ca-CondMat collaboration networks \citep{leskovec2007graph}, the political-blogs network
\citep{adamic2005political}, and the Western-states power-grid network \citep{watts1998collective}.
Code reproducing all simulations, tables, and figures, together with the scripts that construct each
snapshot set from its public source, is openly available at
\url{https://osf.io/fer69/overview?view_only=6f79f15628084adcb253f98b127d4649}.

\bibliographystyle{plainnat}
\bibliography{references}

\clearpage
\setcounter{section}{0}
\setcounter{subsection}{0}
\setcounter{equation}{0}
\setcounter{theorem}{0}
\setcounter{lemma}{0}
\setcounter{remark}{0}
\setcounter{assumption}{0}
\setcounter{definition}{0}
\setcounter{figure}{0}
\setcounter{table}{0}
\renewcommand{\thesection}{S\arabic{section}}
\renewcommand{\thesubsection}{S\arabic{section}.\arabic{subsection}}
\renewcommand{\theequation}{S\arabic{equation}}
\renewcommand{\thetheorem}{S\arabic{theorem}}
\renewcommand{\thelemma}{S\arabic{lemma}}
\renewcommand{\theremark}{S\arabic{remark}}
\renewcommand{\thefigure}{S\arabic{figure}}
\renewcommand{\thetable}{S\arabic{table}}
\part*{Supplementary Material}
\addcontentsline{toc}{part}{Supplementary Material}
\renewcommand{\thepart}{}

\singlespacing

This supplement contains the complete statements and proofs of the results of
Section~4 of the main text and of the two further inferential and tracking results. Section~\ref{supp:t1} proves the single-snapshot identification,
aliasing, and dyadic central limit theorem (Theorem~1). Section~\ref{supp:t2} proves the
sparse Bernoulli-logit perturbation lemma and the separation rate (Theorem~2).
Section~\ref{supp:t3} proves the point-prequential tracking, switch-recovery, and
localization theorem (Theorem~3). Section~\ref{supp:t4} proves the projection, calibration,
model-averaging collapse, and stacking-gap theorem (Theorem~4). Section~\ref{supp:t5} proves
the static cross-fitted synthesis corollary (Corollary~5). Section~\ref{supp:s6} proves the cross-fold orthogonal unconditional limit (Theorem~1c), and Section~\ref{supp:s7} proves the local minimax lower bound for mechanism tracking (Theorem~3b). Numbering of equations and
constants is local to each section; the constant $C$ may change from line to line.

\paragraph{Notation.}
Notation matches the main text. The population log-odds projection of the network onto the
set of mechanism agents, written $\bbeta^\circ_t$ in the main text, is denoted $\bar\bbeta_{t,n}$
here; its oracle, error-free-agent counterpart is $\bbeta^{\mathrm{or}}_{t,n}$. The synthesis
dimension is $d=J+1$; coefficient vectors are written in non-bold type; and a subscript
$n$ or $t,n$ makes the dependence on the snapshot size explicit. The information scale
$N_{t,n}$, the normalized information matrix $H_{t,n}$, the dyad feature $z_{t,e}=
(1,u^{(1)}_{t,e},\dots,u^{(J)}_{t,e})\trans$, the sparse offset $\log\rho_n$, and the four
mechanism agents (block, geometry, degree, local) are exactly the objects of Section~4.

\section{Single-snapshot synthesis: identification, aliasing, and dyadic limit}
\label{supp:t1}

\begin{theorem}[Single-snapshot synthesis: identification, aliasing, and dyadic limit]
\label{thm:t1supp}
Fix \(t\) and let \(n\to\infty\).  Let \(\mathcal E_{t,n}\) be the dyads used
in the time-\(t\) synthesis update and write
\[
        m_{t,n}=|\mathcal E_{t,n}|.
\]
For \(e\in\mathcal E_{t,n}\), put
\[
        Y_{t,e}=A_{t,e},
        \qquad
        z_{t,e}
        =
        (1,u^{(1)}_{t,e},\ldots,u^{(J)}_{t,e})^\top
        \in\mathbb R^d,
        \qquad
        d=J+1 .
\]
For the rate and central limit theorem below, \(d\) is fixed.  Let
\(\ell_{\rho,n}=\log\rho_n\) and define the sparse-offset synthesis probabilities
\[
        \eta_{t,e}(\bbeta)
        =
        \ell_{\rho,n}+\bbeta^\top z_{t,e},
        \qquad
        q_{t,e}(\bbeta)
        =
        \sigma\{\eta_{t,e}(\bbeta)\},
        \qquad
        \sigma(x)=\frac{e^x}{1+e^x}.
\]
Let \(\mathcal F_{t,n}\) be the sigma-field generated by the history, the
cross-fitted agents, the sampled dyad set \(\mathcal E_{t,n}\), the offset
\(\ell_{\rho,n}\), and the true conditional edge-probability array
\(p_{t,e}=\mathbb P(Y_{t,e}=1\mid \mathcal F_{t,n})\).
Assume that the dyad set is conditionally non-informative for the edge
outcomes: conditional on \(\mathcal F_{t,n}\),
\(Y_{t,e}\sim {\rm Bernoulli}(p_{t,e})\), \(e\in\mathcal E_{t,n}\), independently over \(e\).
Outcome-dependent case-control sampling is excluded
from this theorem unless the corresponding retrospective logit correction is
included in the likelihood.  Assume the following basic conditions.
\begin{enumerate}
\item[(B1)] The parameter space \(\Theta\subset\mathbb R^d\) is compact,
convex, and full-dimensional.
\item[(B2)] The fitted-agent design is uniformly bounded:
\(\max_{e\in\mathcal E_{t,n}}\|z_{t,e}\|\le B<\infty\).
\item[(B3)] The sparse scale is preserved uniformly on \(\Theta\): there exist
constants \(0<c_q<C_q<\infty\) such that, for all sufficiently large \(n\),
\(c_q\rho_n\le q_{t,e}(\bbeta)\le C_q\rho_n\le \tfrac12\)
for all \( e\in\mathcal E_{t,n},\ \bbeta\in\Theta\).
The true probabilities are on the same sparse scale,
\(0\le p_{t,e}\le C_p\rho_n\), and \(m_{t,n}\rho_n\to\infty\).
\end{enumerate}
Define the conditional fitted-agent population risk
\[
        R_{t,n}(\bbeta)
        =
        \sum_{e\in\mathcal E_{t,n}}
        \left[
          \log\{1+\exp(\eta_{t,e}(\bbeta))\}
          -
          p_{t,e}\eta_{t,e}(\bbeta)
        \right],
\]
and the pseudo-true set
\(\mathcal B_{t,n}=\arg\min_{\bbeta\in\Theta}R_{t,n}(\bbeta)\).
For any \(\bbeta\in\Theta\), let
\(w_{t,e}(\bbeta)=q_{t,e}(\bbeta)\{1-q_{t,e}(\bbeta)\}\),
\(N_{t,n}(\bbeta)=\sum_{e}w_{t,e}(\bbeta)\), and
\(H_{t,n}(\bbeta)=N_{t,n}(\bbeta)^{-1}\sum_{e}w_{t,e}(\bbeta)z_{t,e}z_{t,e}^\top\).
\begin{enumerate}
\item[(a)] Let
\(\mathcal N_{t,n}=\{a\in\mathbb R^d: a^\top z_{t,e}=0 \ \forall e\in\mathcal E_{t,n}\}\).
Because \(w_{t,e}(\bbeta)>0\),
\({\rm Null}\{H_{t,n}(\bbeta)\}=\mathcal N_{t,n}\) for every \(\bbeta\in\Theta\).
Hence the conditional synthesis law
\(\bbeta\mapsto\{q_{t,e}(\bbeta):e\in\mathcal E_{t,n}\}\)
is identified on \(\Theta\) if and only if \(\mathcal N_{t,n}=\{0\}\),
equivalently if and only if \(H_{t,n}(\bbeta)\) is full rank for one, and hence
for every, \(\bbeta\in\Theta\).  If \(\mathcal N_{t,n}\neq\{0\}\), then \(\bbeta\)
is identified only modulo \(\mathcal N_{t,n}\), and a linear functional
\(a^\top\bbeta\) is identifiable if and only if
\(a\in\mathcal N_{t,n}^\perp={\rm Range}\{H_{t,n}(\bbeta)\}\).
Thus exact affine dependence among \(1,u^{(1)}_{t,\cdot},\ldots,u^{(J)}_{t,\cdot}\)
aliases the corresponding separate agent weights.
\end{enumerate}
For the estimation and CLT statements, add:
\begin{enumerate}
\item[(C1)] The pseudo-true set has a uniformly interior representative
\(\bar\bbeta_{t,n}\in\mathcal B_{t,n}\) with
\({\rm dist}(\bar\bbeta_{t,n},\partial\Theta)\ge c_\Theta>0\); hence
\(\Psi_{t,n}(\bar\bbeta_{t,n})=0\), where
\(\Psi_{t,n}(\bbeta)=\sum_{e}\{p_{t,e}-q_{t,e}(\bbeta)\}z_{t,e}\).
\item[(C2)] The fixed-design Gram matrix
\(G_{t,n}=m_{t,n}^{-1}\sum_e z_{t,e}z_{t,e}^\top\) satisfies
\(\lambda_{\min}(G_{t,n})\ge g>0\).  Write \(N_{t,n}=N_{t,n}(\bar\bbeta_{t,n})\),
\(H_{t,n}=H_{t,n}(\bar\bbeta_{t,n})\).
\item[(C3)] With
\(V_{t,n}=N_{t,n}^{-1}\sum_e p_{t,e}(1-p_{t,e})z_{t,e}z_{t,e}^\top\),
assume \(H_{t,n}\to H_t\) and \(V_{t,n}\to V_t\) with \(H_t\) nonsingular.
\item[(C4)] Let \(\hat\bbeta_{t,n}\) minimize \(\ell_{t,n}(\bbeta)+P_{t,n}(\bbeta)\),
where \(\ell_{t,n}(\bbeta)=\sum_e[\log\{1+\exp(\eta_{t,e}(\bbeta))\}-Y_{t,e}\eta_{t,e}(\bbeta)]\)
and \(P_{t,n}\) is an \(\mathcal F_{t,n}\)-measurable, twice continuously
differentiable negative log-prior or penalty.  For the MLE, \(P_{t,n}\equiv0\); for a
weak MAP estimator, \(\|\nabla P_{t,n}(\bar\bbeta_{t,n})\|=o(\sqrt{N_{t,n}})\) and
\(\sup_{\bbeta\in\Theta}\|\nabla^2P_{t,n}(\bbeta)\|_{\rm op}=o(N_{t,n})\).
\end{enumerate}
Then, conditionally on \(\mathcal F_{t,n}\):
\begin{enumerate}
\item[(b)]
\(\|\hat\bbeta_{t,n}-\bar\bbeta_{t,n}\|=O_{\mathbb P}(\sqrt{(J+1)/N_{t,n}})\).
Under (B3), \(N_{t,n}\asymp m_{t,n}\rho_n\); hence for a full undirected snapshot with
\(m_{t,n}\asymp n^2\), \(\|\hat\bbeta_{t,n}-\bar\bbeta_{t,n}\|=O_{\mathbb P}(\sqrt{(J+1)/(n^2\rho_n)})\).
\item[(c)] The estimator has the conditional asymptotic linear representation
\[
        \sqrt{N_{t,n}}(\hat\bbeta_{t,n}-\bar\bbeta_{t,n})
        =
        H_{t,n}^{-1}N_{t,n}^{-1/2}
        \sum_{e\in\mathcal E_{t,n}}(Y_{t,e}-p_{t,e})z_{t,e}+o_{\mathbb P}(1),
\]
and consequently
\(\sqrt{N_{t,n}}(\hat\bbeta_{t,n}-\bar\bbeta_{t,n})\Rightarrow N(0,H_t^{-1}V_tH_t^{-1})\).
If the Bernoulli-logit model is correctly specified,
\(p_{t,e}=q_{t,e}(\bar\bbeta_{t,n})\) for every \(e\), then \(V_{t,n}=H_{t,n}\), so the
limit reduces to \(N(0,H_t^{-1})\).
\end{enumerate}
If the arrays \(\mathcal F_{t,n},\{z_{t,e}\},\{p_{t,e}\},\mathcal E_{t,n}\) are random, the
convergence is quenched: for every deterministic sequence satisfying the assumptions the
conditional law converges as stated, and if the assumptions hold with probability tending to
one and \(H_{t,n},V_{t,n}\) converge in probability, then the conditional bounded-Lipschitz
distance to the stated Gaussian limit converges to zero in probability.
Finally, let \(\bbeta^{\mathrm{or}}_{t,n}\) denote an oracle pseudo-true target based on population,
error-free agent surfaces, and suppose
\(N_{t,n}^{-1}\sup_{\bbeta\in\Theta}\|\Psi^{\rm fit}_{t,n}(\bbeta)-\Psi^{\rm or}_{t,n}(\bbeta)\|
=O_{\mathbb P}(a_{n,t})\)
and \(-N_{t,n}^{-1}\dot\Psi^{\rm or}_{t,n}(\bbeta)\succeq\lambda_{\rm or}I_d\) on the segment
between \(\bar\bbeta_{t,n}\) and \(\bbeta^{\mathrm{or}}_{t,n}\).
Then \(\|\bar\bbeta_{t,n}-\bbeta^{\mathrm{or}}_{t,n}\|=O_{\mathbb P}(a_{n,t})\). If
\(\|\tilde\bbeta_{t,n}-\hat\bbeta_{t,n}\|=O_{\mathbb P}(r_{\rm Lap})\), then
\[
        \|\tilde\bbeta_{t,n}-\bbeta^{\mathrm{or}}_{t,n}\|
        =
        O_{\mathbb P}\big(\sqrt{(J+1)/N_{t,n}}+a_{n,t}+r_{\rm Lap}\big).
\]
An oracle-centered CLT with the same covariance holds under
\(\sqrt{N_{t,n}}\{a_{n,t}+r_{\rm Lap}\}\to0\); without it the valid normal approximation is
centered at \(\bar\bbeta_{t,n}\), not \(\bbeta^{\mathrm{or}}_{t,n}\).
\end{theorem}

\begin{proof}
All probability statements in the main part of the proof are conditional on
\(\mathcal F_{t,n}\).  Thus the dyad set, the fitted-agent features, the
sparse offset, and the true probability array are fixed.  We suppress \(t\)
from the notation when no confusion is possible.
Conditional on \(\mathcal F_{t,n}\), the variables \(\{Y_e:e\in\mathcal E_n\}\)
are independent Bernoulli with success probabilities \(p_e\).  Since
the dyad sample is conditionally non-informative, the ordinary Bernoulli
likelihood is the correct conditional likelihood.
We first prove identification.  Because \(\sigma\) is strictly increasing,
\(q_e(\bbeta)=q_e(\bbeta')\) for every \(e\in\mathcal E_n\) if and only if
\((\bbeta-\bbeta')^\top z_e=0\) for every \(e\).
For any \(\bbeta\in\Theta\) and \(a\in\mathbb R^d\),
\[
        a^\top H_n(\bbeta)a
        =
        N_n(\bbeta)^{-1}\sum_{e\in\mathcal E_n}w_e(\bbeta)(a^\top z_e)^2 .
\]
Since \(w_e(\bbeta)>0\), this equals zero if and only if \(a^\top z_e=0\) for every \(e\).
Therefore \({\rm Null}\{H_n(\bbeta)\}=\mathcal N_n\), and the map
\(\bbeta\mapsto\{q_e(\bbeta)\}\) is injective on \(\Theta\) iff \(\mathcal N_n=\{0\}\),
i.e. iff \(H_n(\bbeta)\) is full rank.  If \(a\in\mathcal N_n\), then
\(q_e(\bbeta+sa)=q_e(\bbeta)\) for every feasible \(s\) and \(e\), so the likelihood and risk
are constant along feasible null-space directions, whence \(\bbeta\) is identified only modulo
\(\mathcal N_n\).  A linear functional \(b^\top\bbeta\) is constant on
\(\bbeta+\mathcal N_n\) iff \(b\perp\mathcal N_n\), i.e. \(b\in\mathcal N_n^\perp=
{\rm Range}\{H_n(\bbeta)\}\).  An affine dependence among the agent dyad-logit surfaces is a
nonzero \(c\) with \(c^\top z_e=0\) for every \(e\), i.e. \(c\in\mathcal N_n\); the
corresponding separate weights are aliased.  This proves (a).
We now prove the rate and the CLT.  By (B3),
\(q_e(\bbeta)\{1-q_e(\bbeta)\}\asymp\rho_n\) uniformly over \(e,\bbeta\), so
\(N_n=\sum_e \bar q_e(1-\bar q_e)\asymp m_n\rho_n\) with \(\bar q_e=q_e(\bar\bbeta_n)\); thus
\(N_n\to\infty\), and if \(m_n\asymp n^2\) then \(N_n\asymp n^2\rho_n\).
Let \(U_n(\bbeta)=\sum_e\{Y_e-q_e(\bbeta)\}z_e\) be the empirical score, so
\(\nabla\ell_n(\bbeta)=-U_n(\bbeta)\).  The conditional population score is
\(\Psi_n(\bbeta)=\sum_e\{p_e-q_e(\bbeta)\}z_e\), and \(\Psi_n(\bar\bbeta_n)=0\), so
\(U_n(\bar\bbeta_n)=\sum_e(Y_e-p_e)z_e\) with conditional covariance
\(\sum_e p_e(1-p_e)z_ez_e^\top=N_nV_n\).  Because \(p_e\le C_p\rho_n\), \(\|z_e\|\le B\), and
\(N_n\asymp m_n\rho_n\), the \(V_n\) are uniformly bounded, so for fixed \(d\),
\(\|U_n(\bar\bbeta_n)\|=O_{\mathbb P}(\sqrt{N_n})\).
The score derivative is \(\dot U_n(\bbeta)=-\sum_e q_e(\bbeta)\{1-q_e(\bbeta)\}z_ez_e^\top\), so
the information is \(I_n(\bbeta)=\sum_e q_e(\bbeta)\{1-q_e(\bbeta)\}z_ez_e^\top\).  By (C2),
\(G_n\succeq gI_d\), and by (B3) there is \(c>0\) with
\(q_e(\bbeta)\{1-q_e(\bbeta)\}\ge c\rho_n\), so
\(I_n(\bbeta)\succeq c\rho_n m_n G_n\succeq cg\,m_n\rho_n I_d\succeq\lambda N_nI_d\)
uniformly over \(\bbeta\in\Theta\), and \(I_n(\bar\bbeta_n)=N_nH_n\).
Let \(Q_n(\bbeta)=\ell_n(\bbeta)+P_n(\bbeta)\).  By (C4),
\(\nabla^2Q_n(\bbeta)=I_n(\bbeta)+\nabla^2P_n(\bbeta)\succeq\tfrac\lambda2 N_nI_d\) uniformly,
with probability tending to one, so \(Q_n\) is strongly convex with curvature of order \(N_n\).
Writing \(h=\bbeta-\bar\bbeta_n\), Taylor's theorem gives, for some \(\bbeta_h\) on the segment,
\[
        Q_n(\bar\bbeta_n+h)-Q_n(\bar\bbeta_n)
        \ge
        -\|h\|\{\|U_n(\bar\bbeta_n)\|+\|\nabla P_n(\bar\bbeta_n)\|\}
        +\tfrac\lambda4 N_n\|h\|^2 .
\]
Taking \(\|h\|=M\sqrt{d/N_n}\) and \(M\) large makes the right side positive on the boundary
of the ball with probability tending to one; since \(\bar\bbeta_n\) is interior and \(Q_n\) is
strongly convex, the minimizer lies inside, so
\(\|\hat\bbeta_n-\bar\bbeta_n\|=O_{\mathbb P}(\sqrt{d/N_n})\), which is (b) because \(d=J+1\).
As \(\|\hat\bbeta_n-\bar\bbeta_n\|=o_{\mathbb P}(1)\) and \(\bar\bbeta_n\) is interior,
\(\hat\bbeta_n\) is interior with probability tending to one, so
\(U_n(\hat\bbeta_n)=\nabla P_n(\hat\bbeta_n)\).  Expanding \(U_n\) about \(\bar\bbeta_n\),
\[
        U_n(\hat\bbeta_n)
        =
        U_n(\bar\bbeta_n)-N_nH_n(\hat\bbeta_n-\bar\bbeta_n)+R_n,
\]
where, since the logistic third derivative is bounded by a constant times
\(\sum_e q_e(\bbeta)\{1-q_e(\bbeta)\}\|z_e\|^3=O(N_n)\),
\(\|R_n\|=O_{\mathbb P}(N_n\|\hat\bbeta_n-\bar\bbeta_n\|^2)=O_{\mathbb P}(1)=o_{\mathbb P}(\sqrt{N_n})\).
Also \(\|\nabla P_n(\hat\bbeta_n)\|\le\|\nabla P_n(\bar\bbeta_n)\|+
\sup_{\bbeta}\|\nabla^2P_n(\bbeta)\|_{\rm op}\|\hat\bbeta_n-\bar\bbeta_n\|=o_{\mathbb P}(\sqrt{N_n})\).
Combining,
\[
N_nH_n(\hat\bbeta_n-\bar\bbeta_n)=U_n(\bar\bbeta_n)+o_{\mathbb P}(\sqrt{N_n}),
\]
so
\[
\sqrt{N_n}(\hat\bbeta_n-\bar\bbeta_n)=H_n^{-1}N_n^{-1/2}U_n(\bar\bbeta_n)+o_{\mathbb P}(1),
\]
which is the asserted linear representation since \(U_n(\bar\bbeta_n)=\sum_e(Y_e-p_e)z_e\).
For the normal limit, fix \(a\in\mathbb R^d\) and set
\(X_{e,n}=N_n^{-1/2}a^\top z_e(Y_e-p_e)\).  These are conditionally independent, mean zero,
with \(|X_{e,n}|\le\|a\|B/\sqrt{N_n}\to0\), so Lindeberg's condition holds, and their variance
sums to \(a^\top V_na\).  If \(V_n\to V_t\), Lindeberg--Feller and the Cram\'er--Wold device give
\(N_n^{-1/2}\sum_e(Y_e-p_e)z_e\Rightarrow N(0,V_t)\); since \(H_n\to H_t\) nonsingular,
Slutsky yields \(\sqrt{N_n}(\hat\bbeta_n-\bar\bbeta_n)\Rightarrow N(0,H_t^{-1}V_tH_t^{-1})\).
Under correct specification \(p_e=\bar q_e\), so \(V_n=H_n\) and the covariance is \(H_t^{-1}\).
When the arrays are random, the argument applies to every deterministic realization satisfying
the assumptions, giving the quenched statement.
For the oracle target, \(\Psi^{\rm fit}_n(\bar\bbeta_n)=0\) and
\(\Psi^{\rm or}_n(\bbeta^{\mathrm{or}}_{t,n})=0\), so
\(\|\Psi^{\rm or}_n(\bar\bbeta_n)\|=O_{\mathbb P}(N_na_{n,t})\); the mean-value theorem and the
oracle nonsingularity bound give \(\|\bar\bbeta_n-\bbeta^{\mathrm{or}}_{t,n}\|=O_{\mathbb P}(a_{n,t})\).
The triangle inequality with \(\|\tilde\bbeta_n-\hat\bbeta_n\|=O_{\mathbb P}(r_{\rm Lap})\) and
the rate in (b) gives the displayed three-term bound, and under
\(\sqrt{N_n}\{a_{n,t}+r_{\rm Lap}\}\to0\) the oracle-centered CLT has the same covariance.
\end{proof}

\section{Sparse Bernoulli-logit perturbations and the separation rate}
\label{supp:t2}

The lemma below is the engine for the separation rate of Theorem~2. It quantifies, on the
sparse scale, how the Kullback--Leibler divergence, the Hellinger affinity, the $\chi^2$
divergence, and the log-likelihood-ratio variance between two Bernoulli laws scale with the
dyadic log-odds gap. The separation rate then follows by chi-square tensorization across
independent dyads together with the two-point method.

\begin{lemma}[Sparse Bernoulli logit perturbations]
\label{lem:sparse-bernoulli-logit}
Let \(p,q\in(0,1)\), and write
\(x=\operatorname{logit}(q)-\operatorname{logit}(p)\).
Assume that, for fixed constants \(0<c_0<C_0<\infty\) and \(\eta_0<\infty\),
\[
c_0\rho \le p,q\le C_0\rho\le \frac12,
\qquad
|x|\le \eta_0 .
\tag{L0}
\]
Then there exist constants \(0<c<C<\infty\), depending only on \(c_0,C_0,\eta_0\), such that
\[
c\,p(1-p)x^2
\le
D\{\operatorname{Bern}(q)\|\operatorname{Bern}(p)\}
\le
C\,p(1-p)x^2,
\tag{L1}
\]
\[
c\,p(1-p)x^2
\le
-\log\!\left[\sqrt{pq}+\sqrt{(1-p)(1-q)}\right]
\le
C\,p(1-p)x^2,
\tag{L2}
\]
\[
\chi^2\{\operatorname{Bern}(q),\operatorname{Bern}(p)\}
=
\frac{(q-p)^2}{p(1-p)}
\le
C\,p(1-p)x^2,
\tag{L3}
\]
and the same bounds hold with \(p\) and \(q\) interchanged.  Moreover, if
\(Z=Y\log\tfrac qp+(1-Y)\log\tfrac{1-q}{1-p}\) with \(Y\sim\operatorname{Bern}(q)\), then
\[
\operatorname{Var}(Z)\le C\,p(1-p)x^2 .
\tag{L4}
\]
\end{lemma}

\begin{proof}
Since \(q=\sigma\{\operatorname{logit}(p)+x\}=\frac{pe^x}{1-p+pe^x}\),
\[
q-p=\frac{p(1-p)(e^x-1)}{1-p+pe^x}.
\]
Under (L0), \(p\asymp q\asymp \rho\) and \(1-p\asymp 1-q\asymp1\); and because \(x\in[-\eta_0,\eta_0]\),
the map \(x\mapsto(e^x-1)/x\) (extended to \(1\) at \(0\)) is bounded above and below by positive
constants. Hence
\[
(q-p)^2\asymp p^2(1-p)^2x^2 .
\tag{L5}
\]
For Bernoulli laws,
\(D\{\operatorname{Bern}(q)\|\operatorname{Bern}(p)\}=q\log\tfrac qp+(1-q)\log\tfrac{1-q}{1-p}\).
As a function of \(q\) it vanishes to first order at \(q=p\) and has second derivative
\(r^{-1}(1-r)^{-1}\) at an intermediate \(r\) between \(p\) and \(q\); since \(p\asymp q\asymp\rho\),
Taylor's theorem gives
\(D\{\operatorname{Bern}(q)\|\operatorname{Bern}(p)\}\asymp(q-p)^2/\{p(1-p)\}\asymp p(1-p)x^2\),
which is (L1). With the Hellinger affinity \(A(p,q)=\sqrt{pq}+\sqrt{(1-p)(1-q)}\),
\(2\{1-A(p,q)\}=(\sqrt p-\sqrt q)^2+(\sqrt{1-p}-\sqrt{1-q})^2\), so
\(1-A(p,q)\asymp(q-p)^2/\{p(1-p)\}\asymp p(1-p)x^2\); since \(p,q\le\tfrac12\), \(A\) is bounded away
from zero and \(-\log A\asymp 1-A\), which is (L2). The identity
\(\chi^2\{\operatorname{Bern}(q),\operatorname{Bern}(p)\}=(q-p)^2/\{p(1-p)\}\) with (L5) is (L3); the
interchanged bounds follow since \(p\asymp q\). Finally \(\log(q/p)=x-\log\{1+p(e^x-1)\}=O(x)\) and
\(\log\{(1-q)/(1-p)\}=-\log\{1+p(e^x-1)\}=O(px)\), so
\(\mathbb E Z^2\le C\{qx^2+(1-q)p^2x^2\}\le C\,p(1-p)x^2\), giving (L4).
\end{proof}

We now state the separation rate. Fix a snapshot and a dyad set \(\mathcal E_n\) of size
\(m_n\asymp n^2\). Let \(p_e\) and \(q_e\) be two sparse edge-probability arrays on the scale
\(\rho_n\), with per-dyad log-odds gap \(x_e=\operatorname{logit}q_e-\operatorname{logit}p_e\),
and define the aggregate Fisher signal
\[
        S_n
        =
        \sum_{e\in\mathcal E_n}p_e(1-p_e)x_e^2 .
\]
Under a uniform gap \(|x_e|\asymp\Delta_n\) and the sparse scale \(p_e(1-p_e)\asymp\rho_n\),
\(S_n\asymp m_n\rho_n\Delta_n^2\asymp n^2\rho_n\Delta_n^2\).

\begin{theorem}[Mechanism separation rate]
\label{thm:t2supp}
Assume the sparse scale \(c_0\rho_n\le p_e,q_e\le C_0\rho_n\le\tfrac12\) and the bounded gap
\(\sup_e|x_e|\le\eta_0\) of Lemma~\ref{lem:sparse-bernoulli-logit}, with the dyads conditionally
independent. Write \(P_n=\bigotimes_{e}\operatorname{Bern}(p_e)\) and
\(Q_n=\bigotimes_{e}\operatorname{Bern}(q_e)\).
\begin{enumerate}
\item[(a)] {\rm(Achievability.)} If \(S_n\to\infty\), there is a test \(\phi_n\) of \(P_n\) versus
\(Q_n\) with \(\mathbb E_{P_n}\phi_n+\mathbb E_{Q_n}(1-\phi_n)\to0\). In particular, the two
mechanism configurations are asymptotically distinguishable from a single snapshot whenever
\(n^2\rho_n\Delta_n^2\to\infty\).
\item[(b)] {\rm(Impossibility.)} If \(S_n=O(1)\), then
\(\liminf_n\inf_{\phi}\{\mathbb E_{P_n}\phi+\mathbb E_{Q_n}(1-\phi)\}>0\), where the infimum is over
all tests. The two configurations are mutually contiguous and cannot be separated with power
tending to one. In particular this holds whenever \(n^2\rho_n\Delta_n^2=O(1)\).
\end{enumerate}
Hence the critical separation rate is \(\Delta_n\asymp(n^2\rho_n)^{-1/2}\): detection is possible
above it and impossible below it.
\end{theorem}

\begin{proof}
The argument is the two-point method, with the per-dyad divergence bounds supplied by
Lemma~\ref{lem:sparse-bernoulli-logit} and aggregated across independent dyads.

\emph{Impossibility (b).} By the tensorization identity for the \(\chi^2\) divergence of product
measures,
\[
        1+\chi^2(Q_n\,\|\,P_n)
        =
        \prod_{e\in\mathcal E_n}
        \left\{1+\chi^2\big(\operatorname{Bern}(q_e)\,\|\,\operatorname{Bern}(p_e)\big)\right\}
        \le
        \exp\!\left\{\sum_{e\in\mathcal E_n}\chi^2_e\right\},
\]
where \(\chi^2_e=\chi^2(\operatorname{Bern}(q_e)\,\|\,\operatorname{Bern}(p_e))\) and we used
\(1+u\le e^u\). By Lemma~\ref{lem:sparse-bernoulli-logit}(L3),
\(\chi^2_e\le C\,p_e(1-p_e)x_e^2\), so
\(\sum_e\chi^2_e\le C\,S_n=O(1)\). Therefore \(\chi^2(Q_n\,\|\,P_n)\le e^{C\,O(1)}-1=O(1)\). The
total variation distance obeys
\(\mathrm{TV}(P_n,Q_n)\le\tfrac12\sqrt{\chi^2(Q_n\,\|\,P_n)}\), so \(\mathrm{TV}(P_n,Q_n)\) is bounded
away from one. Since for any test
\(\mathbb E_{P_n}\phi+\mathbb E_{Q_n}(1-\phi)\ge1-\mathrm{TV}(P_n,Q_n)\), with equality for the
likelihood-ratio test, the total testing error is bounded away from zero, proving (b).

\emph{Achievability (a).} Consider the log-likelihood-ratio statistic
\[
        \Lambda_n
        =
        \sum_{e\in\mathcal E_n}
        \left\{A_e\log\frac{q_e}{p_e}+(1-A_e)\log\frac{1-q_e}{1-p_e}\right\},
\]
and the test \(\phi_n=\1\{\Lambda_n>\tau_n\}\) with
\(\tau_n=\tfrac12\{\mathbb E_{Q_n}\Lambda_n+\mathbb E_{P_n}\Lambda_n\}\). Writing
\(D_n^Q=D(Q_n\|P_n)=\mathbb E_{Q_n}\Lambda_n\) and \(D_n^P=D(P_n\|Q_n)=-\mathbb E_{P_n}\Lambda_n\),
Lemma~\ref{lem:sparse-bernoulli-logit}(L1) and its interchanged form give
\[
        D_n^Q=\sum_e D\{\operatorname{Bern}(q_e)\|\operatorname{Bern}(p_e)\}\ge c\,S_n,
        \qquad
        D_n^P=\sum_e D\{\operatorname{Bern}(p_e)\|\operatorname{Bern}(q_e)\}\ge c\,S_n,
\]
so \(\mathbb E_{Q_n}\Lambda_n-\mathbb E_{P_n}\Lambda_n=D_n^Q+D_n^P\ge 2c\,S_n\to\infty\). By
Lemma~\ref{lem:sparse-bernoulli-logit}(L4) applied under \(Q_n\), and its interchanged form under
\(P_n\), the per-dyad summands have variance at most \(C\,p_e(1-p_e)x_e^2\), so by independence
\[
        \operatorname{Var}_{Q_n}(\Lambda_n)\le C\,S_n,
        \qquad
        \operatorname{Var}_{P_n}(\Lambda_n)\le C\,S_n .
\]
Under \(Q_n\), \(\mathbb E_{Q_n}\Lambda_n=D_n^Q\ge\tau_n+c\,S_n\) (since
\(\tau_n=\tfrac12(D_n^Q-D_n^P)\le\tfrac12 D_n^Q\le D_n^Q-c\,S_n\) for \(S_n\) large), and Chebyshev's
inequality gives
\[
        \mathbb E_{Q_n}(1-\phi_n)
        =
        \Prob_{Q_n}(\Lambda_n\le\tau_n)
        \le
        \frac{\operatorname{Var}_{Q_n}(\Lambda_n)}{(D_n^Q-\tau_n)^2}
        \le
        \frac{C\,S_n}{(c\,S_n)^2}
        =
        \frac{C}{c^2 S_n}\to0 .
\]
Symmetrically, under \(P_n\), \(\mathbb E_{P_n}\Lambda_n=-D_n^P\le\tau_n-c\,S_n\), so
\(\mathbb E_{P_n}\phi_n=\Prob_{P_n}(\Lambda_n>\tau_n)\le C/(c^2S_n)\to0\). Adding the two error
probabilities proves (a). The scaling statements follow from
\(S_n\asymp n^2\rho_n\Delta_n^2\) under a uniform gap.
\end{proof}

\begin{remark}
The lemma supplies the sparse Bernoulli-logit perturbation bound used in the separation proof; the separation rate of Theorem~\ref{thm:t2supp} is
obtained from it by the standard chi-square tensorization and two-point (Le~Cam) reduction. The
optimal-test ceiling means the achievability test in (a) can be taken to be the likelihood-ratio
test, and the impossibility in (b) applies to every test, so the rate \((n^2\rho_n)^{-1/2}\) is
two-sided.
\end{remark}

\section{Point-prequential tracking, switch recovery, and localization}
\label{supp:t3}

\paragraph{Loss, pseudo-truth, and information scale.}
Let \(d=J+1\).  At time \(t\), let \(\mathcal E_t\) be the held-out dyads used by the synthesis
update.  For \(e\in\mathcal E_t\), define
\[
  \eta_{t,e}(\bbeta)
  =
  o_{t,e}+z_{t,e}^{\top}\bbeta,
  \qquad
  q_{t,e}(\bbeta)
  =
  \sigma\{\eta_{t,e}(\bbeta)\},
  \qquad
  \sigma(x)=\frac{e^x}{1+e^x},
\]
where \(z_{t,e}\in\R^d\).  Define the Bernoulli-logistic negative log-likelihood
\[
  \ell_t(\bbeta)
  =
  \sum_{e\in\mathcal E_t}
  \left[
    \log\{1+\exp(\eta_{t,e}(\bbeta))\}
    -
    A_{t,e}\eta_{t,e}(\bbeta)
  \right],
\]
and its conditional mean
\(L_t(\bbeta):=\E\{\ell_t(\bbeta)\mid\F_{t-1}\}\).
The snapshot pseudo-true synthesis state is
\(\bbeta_t^\circ\in\argmin_{\bbeta\in\mathcal B}L_t(\bbeta)\),
and the unnormalized excess expected log-loss is
\(\mathcal R_t(\bbeta):=L_t(\bbeta)-L_t(\bbeta_t^\circ)\).
In the per-dyad normalization of the main text and Section~\ref{supp:t4}, with
\[
R_t(\bbeta)=m_t^{-1}\sum_{e\in\mathcal E_t}\KL\{\mathrm{Bern}(p^*_{t,e}),\mathrm{Bern}(q_{t,e}(\bbeta))\}
\]
and \(m_t=|\mathcal E_t|\), the excess risk is
\(R_t(\bbeta)-R_t(\bbeta_t^\circ)=m_t^{-1}\mathcal R_t(\bbeta)\); the tracking bounds below are stated in
the normalization-free excess log-loss \(\mathcal R_t\).
Define the snapshot information scale
\[
  N_t
  :=
  \sum_{e\in\mathcal E_t}
  q_{t,e}(\bbeta_t^\circ)\{1-q_{t,e}(\bbeta_t^\circ)\}.
\]

\begin{assumption}[Regularity for point-state dynamic synthesis]
\label{ass:t3supp}
The following conditions hold uniformly for \(t\le T\).
\begin{enumerate}
\item[(A1)] \textbf{Predictability and cross-fitting.}
There is a filtration \((\F_t)\) such that \(\mathcal E_t\), \(o_{t,e}\), \(z_{t,e}\), and
\(p^*_{t,e}\) are \(\F_{t-1}\)-measurable.  Conditionally on \(\F_{t-1}\), the held-out dyads are
independent and \(A_{t,e}\sim\mathrm{Bernoulli}(p^*_{t,e})\), \(e\in\mathcal E_t\).
\item[(A2)] \textbf{Bounded sparse-safe design and first-order pseudo-truth.}
There is a constant \(K<\infty\) such that \(\sup_{t,e}\norm{z_{t,e}}_\infty\le K\).
The parameter set \(\mathcal B\subset\R^d\) is compact and convex, all filtered states are projected
onto \(\mathcal B\), and \(\nabla L_t(\bbeta_t^\circ)=0\) (automatic if \(\bbeta_t^\circ\) is interior).
\item[(A3)] \textbf{Curvature and variance on the information scale.}
There exist \(0<m<M<\infty\) such that, for every \(\bbeta\in\mathcal B\),
\(mN_t I_d\preceq\nabla^2L_t(\bbeta)\preceq MN_t I_d\), and the conditional score variance satisfies
\(\sum_{e\in\mathcal E_t}p^*_{t,e}(1-p^*_{t,e})z_{t,e}z_{t,e}^{\top}\preceq MN_t I_d\).
\item[(A4)] \textbf{Implementation perturbation.}
Let the ideal discounted logistic filter be
\(\widetilde\bbeta_{t\mid t}\in\argmin_{\bbeta\in\mathcal B}\sum_{s\le t}\delta^{t-s}\ell_s(\bbeta)\),
\(0<\delta<1\). The implemented cross-fitted Laplace filter satisfies, on an event
\(\mathcal A_T\), \(\norm{\widehat\bbeta_{t\mid t}-\widetilde\bbeta_{t\mid t}}_2\le\varepsilon_t\),
\(t=1,\ldots,T\); for the actual algorithm this event is proved separately, e.g. with
\(\varepsilon_t\le a_{n,t}+r_{{\rm Lap},t}\).
\end{enumerate}
\end{assumption}

\begin{theorem}[Point-prequential tracking, switch recovery, and localization]
\label{thm:t3supp}
Let
\(N_t^\delta:=\sum_{s\le t}\delta^{t-s}N_s\),
\(\lambda_T(\alpha):=\log(2dT/\alpha)\), \(\alpha\in(0,1)\), and define
\[
  \eta_t(\alpha):=\sqrt{\frac{d\lambda_T(\alpha)}{N_t^\delta}}+\frac{d\lambda_T(\alpha)}{N_t^\delta},
  \qquad
  B_t^\delta:=\frac{\sum_{s\le t}\delta^{t-s}N_s\norm{\bbeta_s^\circ-\bbeta_t^\circ}_2}{N_t^\delta}.
\]
Then there exists \(C<\infty\), depending only on \(m,M,K\) and \(\mathcal B\), and an event
\(\mathcal E_T(\alpha)\) with \(\Pbb\{\mathcal E_T(\alpha)\}\ge 1-\alpha\), such that, on
\(\mathcal E_T(\alpha)\cap\mathcal A_T\), simultaneously for all \(t\le T\),
\[
  \norm{\widehat\bbeta_{t\mid t}-\bbeta_t^\circ}_2
  \le
  C\left\{\eta_t(\alpha)+B_t^\delta+\varepsilon_t\right\}.
  \tag{T3.1}
\]
Consequently the unconditional probability of {\rm(T3.1)} is at least \(1-\alpha-\Pbb(\mathcal A_T^c)\).
Let \(\theta_t\in\R^J\) be the structural-weight subvector of \(\bbeta_t\), and suppose the active
mechanism \(r_t\) satisfies the margin condition
\(\theta^\circ_{t,r_t}-\max_{j\ne r_t}\theta^\circ_{t,j}\ge\kappa>0\) {\rm(T3.2)}. If
\(C\{\eta_t(\alpha)+B_t^\delta+\varepsilon_t\}<\kappa/2\) {\rm(T3.3)}, then on the same event
\[
  \argmax_{1\le j\le J}\widehat\theta_{t,j}=r_t .
  \tag{T3.4}
\]
For the point one-step-ahead forecast \(\widehat\bbeta_{t\mid t-1}:=\widehat\bbeta_{t-1\mid t-1}\),
\(t\ge2\),
\[
  \sum_{t=2}^{T}\mathcal R_t(\widehat\bbeta_{t\mid t-1})
  \le
  C\sum_{t=2}^{T}N_t
  \left[\left\{\eta_{t-1}(\alpha)+B_{t-1}^\delta+\varepsilon_{t-1}\right\}^2
  +\norm{\bbeta_t^\circ-\bbeta_{t-1}^\circ}_2^2\right].
  \tag{T3.5}
\]
If \(c_NN\le N_t\le C_NN\), \(\Delta_t:=\norm{\bbeta_t^\circ-\bbeta_{t-1}^\circ}_2\),
\(V_2:=\sum_{t=2}^T\Delta_t^2\), and \(N\ge C d\lambda_T(\alpha)\), then
\[
  \sum_{t=2}^{T}\mathcal R_t(\widehat\bbeta_{t\mid t-1})
  \le
  C\left[d\lambda_T(\alpha)\{1+\log T+T(1-\delta)\}
  +\frac{NV_2}{(1-\delta)^2}+N\sum_{t=1}^{T-1}\varepsilon_t^2\right].
  \tag{T3.6}
\]
If the comparator path is piecewise stable with \(S\) switches of size at most \(D\), then
\(V_2\le SD^2\) and the switch contribution is at most \(CNSD^2/(1-\delta)^2\); if in addition the
switch times satisfy \(\tau_{\ell+1}-\tau_\ell\ge L_0/(1-\delta)\) {\rm(T3.7)}, the sharper
switch-tail contribution is \(C(L_0)NSD^2/(1-\delta)\) {\rm(T3.8)}. Finally, for an isolated
two-regime switch at \(\nu\) with \(\norm{\bbeta^+-\bbeta^-}_2\le D\), post-switch active mechanism
\(r^+\) and margin \(\theta^+_{r^+}-\max_{j\ne r^+}\theta^+_j\ge\kappa\), \(0<\kappa\le1\)
{\rm(T3.9)}, under \(N/(1-\delta)\ge C d\lambda_T(\alpha)/\kappa^2\) and \(\sup_t\varepsilon_t\le c\kappa\)
{\rm(T3.10)}, with
\[
  H_\alpha:=\left\lceil C\left[\frac{d\lambda_T(\alpha)}{N\kappa^2}\vee
  \frac{\log(1+CD/\kappa)}{|\log\delta|}\right]\right\rceil
  \tag{T3.11}
\]
and no second switch on \(\nu,\ldots,\nu+H_\alpha\le T\), the recovery delay \(h_{\rm rec}\) (first
time the filtered state localizes \(r^+\)) satisfies, with probability at least
\(1-\alpha-\Pbb(\mathcal A_T^c)\),
\[
  h_{\rm rec}\le H_\alpha .
  \tag{T3.12}
\]
\end{theorem}

\begin{proof}
Define
\[
\Phi_t(\bbeta):=\sum_{s\le t}\delta^{t-s}L_s(\bbeta),
\qquad
\widehat\Phi_t(\bbeta):=\sum_{s\le t}\delta^{t-s}\ell_s(\bbeta),
\]
with \(\bar\bbeta_t\in\argmin_{\bbeta\in\mathcal B}\Phi_t(\bbeta)\). By A3,
\(\nabla^2\Phi_t(\bbeta)=\sum_{s\le t}\delta^{t-s}\nabla^2L_s(\bbeta)\succeq mN_t^\delta I_d\), so
\(\Phi_t\) is \(mN_t^\delta\)-strongly convex and \(\bar\bbeta_t\) is unique.
\emph{Bias.} Since \(\nabla L_s(\bbeta_s^\circ)=0\),
\[
\nabla\Phi_t(\bbeta_t^\circ)=\sum_{s\le t}\delta^{t-s}\{\nabla L_s(\bbeta_t^\circ)-\nabla L_s(\bbeta_s^\circ)\},
\]
and the mean-value theorem with the A3 upper bound gives
\(\|\nabla L_s(\bbeta_t^\circ)-\nabla L_s(\bbeta_s^\circ)\|_2\le MN_s\|\bbeta_t^\circ-\bbeta_s^\circ\|_2\),
so \(\|\nabla\Phi_t(\bbeta_t^\circ)\|_2\le M\sum_{s\le t}\delta^{t-s}N_s\|\bbeta_s^\circ-\bbeta_t^\circ\|_2\).
By optimality of \(\bar\bbeta_t\) and strong convexity at \(\bbeta_t^\circ\),
\(\tfrac{mN_t^\delta}{2}\|\bar\bbeta_t-\bbeta_t^\circ\|_2^2\le\|\nabla\Phi_t(\bbeta_t^\circ)\|_2\|\bar\bbeta_t-\bbeta_t^\circ\|_2\),
hence
\[
  \norm{\bar\bbeta_t-\bbeta_t^\circ}_2\le C B_t^\delta .
  \tag{3}
\]
\emph{Fluctuation.} For the Bernoulli-logistic loss,
\(\ell_s(\bbeta)-L_s(\bbeta)=\sum_{e\in\mathcal E_s}(p^*_{s,e}-A_{s,e})\{o_{s,e}+z_{s,e}^\top\bbeta\}\),
so \(\widehat\Phi_t(\bbeta)-\Phi_t(\bbeta)=C_t+G_t^\top\bbeta\) with
\(G_t=\sum_{s\le t}\delta^{t-s}\sum_{e\in\mathcal E_s}(p^*_{s,e}-A_{s,e})z_{s,e}\). For fixed \(t,k\) the
summands \(\delta^{t-s}(p^*_{s,e}-A_{s,e})z_{s,e,k}\) are conditionally mean-zero martingale
differences bounded by \(K\), with predictable quadratic variation at most
\(M\sum_{s\le t}\delta^{2(t-s)}N_s\le M N_t^\delta\). Freedman's Bernstein inequality and a union
bound over the \(pT\) coordinate-time pairs give an event \(\mathcal E_T(\alpha)\),
\(\Pbb\{\mathcal E_T(\alpha)\}\ge1-\alpha\), on which, for all \(t\),
\(\|G_t\|_2\le C[\sqrt{dN_t^\delta\lambda_T(\alpha)}+d\lambda_T(\alpha)]\). With
\(\widetilde\bbeta_{t\mid t}\in\argmin\widehat\Phi_t\), optimality of \(\widetilde\bbeta_{t\mid t}\) and
\(\bar\bbeta_t\), and strong convexity, give
\(\|\widetilde\bbeta_{t\mid t}-\bar\bbeta_t\|_2\le C\|G_t\|_2/N_t^\delta\le C\eta_t(\alpha)\). On
\(\mathcal A_T\), \(\|\widehat\bbeta_{t\mid t}-\widetilde\bbeta_{t\mid t}\|_2\le\varepsilon_t\).
Combining with (3) proves (T3.1), and the probability bound follows from
\(\Pbb\{\mathcal E_T(\alpha)\cap\mathcal A_T\}\ge1-\alpha-\Pbb(\mathcal A_T^c)\).
\emph{Localization.} With \(\Delta_t^\theta:=\|\widehat\theta_t-\theta_t^\circ\|_\infty\le
\|\widehat\bbeta_{t\mid t}-\bbeta_t^\circ\|_2\), for \(j\ne r_t\),
\(\widehat\theta_{t,r_t}-\widehat\theta_{t,j}\ge\kappa-2\Delta_t^\theta\); under (T3.3),
\(\Delta_t^\theta<\kappa/2\), giving (T3.4).
\emph{Prequential bound.} Since \(\nabla L_t(\bbeta_t^\circ)=0\), Taylor and A3 give
\(\mathcal R_t(\bbeta)\le CN_t\|\bbeta-\bbeta_t^\circ\|_2^2\) (11). For \(t\ge2\),
\(\widehat\bbeta_{t\mid t-1}=\widehat\bbeta_{t-1\mid t-1}\), so
\(\|\widehat\bbeta_{t\mid t-1}-\bbeta_t^\circ\|_2^2\le2\|\widehat\bbeta_{t-1\mid t-1}-\bbeta_{t-1}^\circ\|_2^2
+2\|\bbeta_t^\circ-\bbeta_{t-1}^\circ\|_2^2\); substituting (T3.1) at \(t-1\) into (11) gives (T3.5).
\emph{Piecewise-stable consequences.} If \(c_NN\le N_t\le C_NN\) then
\(N_t^\delta\asymp N\{t\wedge(1-\delta)^{-1}\}\), so
\(\sum_{t=2}^T N_t/N_{t-1}^\delta\le C\{1+\log T+T(1-\delta)\}\). If \(N\ge Cd\lambda_T(\alpha)\) then
\(\eta_t^2(\alpha)\le Cd\lambda_T(\alpha)/N_t^\delta\), giving
\(\sum_{t=2}^T N_t\eta_{t-1}^2(\alpha)\le Cd\lambda_T(\alpha)\{1+\log T+T(1-\delta)\}\) (14). For the
bias, \(\|\bbeta_s^\circ-\bbeta_t^\circ\|_2\le\sum_{u=s+1}^t\Delta_u\) and the comparability of \(N_t\)
give \(B_t^\delta\le C\sum_{u=2}^t\delta^{t-u}\Delta_u\) (15): the denominator
\(\sum_{r\le t}\delta^{t-r}N_r\asymp N(1-\delta^t)/(1-\delta)\), while a fixed jump \(\Delta_u\)
contributes at most a constant times \(N\Delta_u\delta^{t-u}(1-\delta^t)/(1-\delta)\) to the
numerator. Young's convolution inequality gives
\(\sum_{t=1}^T(B_t^\delta)^2\le C(\sum_{h\ge0}\delta^h)^2\sum_u\Delta_u^2=CV_2/(1-\delta)^2\) (16), so
\(\sum_{t=2}^T N_t(B_{t-1}^\delta)^2\le CNV_2/(1-\delta)^2\) (17). Using \((x+y+z)^2\le3(x^2+y^2+z^2)\)
in (T3.5) with (14) and (17) proves (T3.6); the direct jump term
\(\sum_t N_t\Delta_t^2\le CNV_2\) is absorbed. With \(S\) switches, \(V_2\le SD^2\). Under the spacing
condition (T3.7): a single size-\(D\) switch contributes at most \(CD^2/(1-\delta^2)\le CD^2/(1-\delta)\)
to \(\sum_t(B_t^\delta)^2\), and two switches \(m\) apart contribute at most
\(CD^2\delta^m/(1-\delta^2)\); since \(|\log\delta|\ge1-\delta\), \(m\ge\ell L_0/(1-\delta)\) gives
\(\delta^m\le e^{-cL_0\ell}\), and summing over pairs gives
\(\sum_t(B_t^\delta)^2\le C(L_0)SD^2/(1-\delta)\), so multiplying by \(N\) gives (T3.8).
\emph{Isolated switch.} At \(t=\nu+h\) only pre-switch observations have differing pseudo-truth, so
\(B_{\nu+h}^\delta\le CD\,\{\sum_{s<\nu}\delta^{\nu+h-s}N_s\}/\{\sum_{s\le\nu+h}\delta^{\nu+h-s}N_s\}\)
(18). With \(N_s\asymp N\) this ratio is, up to constants, \(\delta^{h+1}(1-\delta^{\nu-1})/(1-\delta^{\nu+h})\),
which is at most \(2\delta^h\) in both cases \(\delta^h\ge\tfrac12\) and \(\delta^h<\tfrac12\); hence
\(B_{\nu+h}^\delta\le CD\delta^h\) (19). The post-switch observations alone give
\(N_{\nu+h}^\delta\ge cN\{(h+1)\wedge(1-\delta)^{-1}\}\) (20), so
\(\eta_{\nu+h}(\alpha)\le C[\sqrt{d\lambda_T(\alpha)/(N\{(h+1)\wedge(1-\delta)^{-1}\})}
+d\lambda_T(\alpha)/(N\{(h+1)\wedge(1-\delta)^{-1}\})]\) (21). The condition
\(N/(1-\delta)\ge Cd\lambda_T(\alpha)/\kappa^2\) makes the stochastic term \(\le c\kappa\) once
\(h+1\ge Cd\lambda_T(\alpha)/(N\kappa^2)\), and (19) makes \(B_{\nu+h}^\delta\le c\kappa\) once
\(h\ge C\log(1+CD/\kappa)/|\log\delta|\). With \(\sup_t\varepsilon_t\le c\kappa\), at \(h=H_\alpha\)
the bound \(C\{\eta_{\nu+h}(\alpha)+B_{\nu+h}^\delta+\varepsilon_{\nu+h}\}<\kappa/2\) holds, and since
no second switch occurs on the recovery window, (T3.4) at \(t=\nu+H_\alpha\) gives
\(h_{\rm rec}\le H_\alpha\), proving (T3.12). If the prequential sum also scores \(t=1\), (11)
contributes the initial-state term \(CN_1\|\widehat\bbeta_{1\mid0}-\bbeta_1^\circ\|_2^2\).
\end{proof}

\section{Projection, calibration, model-averaging collapse, and stacking gap}
\label{supp:t4}

Throughout this section,
\[
  \sigma(x)=\frac{1}{1+e^{-x}},
  \qquad
  \kl(p,q)=p\log\frac pq+(1-p)\log\frac{1-p}{1-q},
  \qquad
  \Delta_J=\Big\{w\in[0,1]^J:\textstyle\sum_{j=1}^Jw_j=1\Big\}.
\]

\begin{theorem}[Projection, calibration, BMA collapse, and stacking gap]
\label{thm:t4supp}
Fix one snapshot \(t\) and suppress \(t\) from notation. Let \(E_n\) be the evaluated dyads,
\(m_n=|E_n|\), and \(\mathcal F_n\) the sigma-field of the cross-fitted agent predictions.
Conditional on \(\mathcal F_n\), assume \(A_e\) independent, \(A_e\sim\Bern(p_e)\),
\(e\in E_n\). For \(j=1,\ldots,J\) let \(P_e^{(j)}\in(0,1)\) and
\(u_e^{(j)}=\logit P_e^{(j)}-\log\rho_n\), \(z_e=(1,u_e^{(1)},\ldots,u_e^{(J)})^\top\), \(d=J+1\).
For \(\bbeta=(\alpha,\theta_1,\ldots,\theta_J)^\top\in B\subset\mathbb R^d\), set
\(\eta_{\bbeta,e}=\log\rho_n+\bbeta^\top z_e\), \(q_{\bbeta,e}=\sigma(\eta_{\bbeta,e})\). For any
probability vector \(r\), \(R_n(r)=m_n^{-1}\sum_e\kl(p_e,r_e)\) and \(R_n(\bbeta)=R_n(q_{\bbeta})\).
Let \(\bbeta_n^\circ\in\arg\min_{\bbeta\in B}R_n(\bbeta)\), \(q_e^\circ=q_{\bbeta_n^\circ,e}\),
\(N_n=\sum_e q_e^\circ(1-q_e^\circ)\), and \(\|v\|_{2,n}=(m_n^{-1}\sum_e v_e^2)^{1/2}\). Assume:
\[
\begin{array}{ll}
{\rm(A1)}& J\text{ fixed},\ \rho_n\to0,\ N_n\asymp m_n\rho_n\to\infty .\\[1mm]
{\rm(A2)}& c\rho_n\le p_e,\ P_e^{(j)},\ q_{\bbeta,e}\le C\rho_n\ \ \forall e,j,\bbeta\in B,\ \ 0<c<C<\infty.\\[1mm]
{\rm(A3)}& B\text{ compact, convex, }\operatorname{dist}(\bbeta_n^\circ,\partial B)\ge c_B>0 .\\[1mm]
{\rm(A4)}& \gamma_j=(0,\ldots,1,\ldots,0)^\top\in B\text{ and }q_{\gamma_j,e}=P_e^{(j)}\ (j=1,\ldots,J).\\[1mm]
{\rm(A5)}& \sup_e\|z_e\|\le M,\ \ \lambda_{\min}(m_n^{-1}\sum_e z_ez_e^\top)\ge\lambda_z>0 .
\end{array}
\]
Let \(\widehat R_n(\bbeta)=m_n^{-1}\sum_e\{-A_e\log q_{\bbeta,e}-(1-A_e)\log(1-q_{\bbeta,e})\}\) and
\(\widehat\bbeta_n\in\arg\min_{\bbeta\in B}\widehat R_n(\bbeta)\). Then, conditionally on
\(\mathcal F_n\):

\emph{(i) Projection, dyadic rate, conditional CLT.} \(\bbeta_n^\circ\) is unique and
\(\|\widehat\bbeta_n-\bbeta_n^\circ\|=O_p(N_n^{-1/2})\). If
\(H_n=N_n^{-1}\sum_e q_e^\circ(1-q_e^\circ)z_ez_e^\top\to H\) and
\(V_n=N_n^{-1}\sum_e p_e(1-p_e)z_ez_e^\top\to V\) with \(H\) nonsingular, then
\(\sqrt{N_n}(\widehat\bbeta_n-\bbeta_n^\circ)\Rightarrow N(0,H^{-1}VH^{-1})\), and under correct
specification \(p_e=q_e^\circ\) one has \(V=H\). For an approximate loss \(\widetilde R_n\) and
probabilities \(\widetilde q_{\bbeta,e}\) that, with probability tending to one, are strongly convex
with curvature at least \(c_0\rho_n\),
\(\sup_{\bbeta}\|\nabla\widetilde R_n(\bbeta)-\nabla\widehat R_n(\bbeta)\|=O_p\{\rho_n(a_n+r_n)\}\), and
\(\sup_{\bbeta}\|\widetilde q_{\bbeta}-q_{\bbeta}\|_{2,n}=O_p(\rho_n a_n)\) with \(a_n,r_n=o(1)\),
\(\widetilde\bbeta_n\in\arg\min\widetilde R_n\) satisfies
\(\|\widetilde\bbeta_n-\bbeta_n^\circ\|=O_p(N_n^{-1/2}+a_n+r_n)\).

\emph{(ii) Sparse graphon projection.} If dyads have locations \(x_e\in\Omega\), empirical averages
converge uniformly to \(\mu\)-integration over the relevant class, and uniformly in \(e\)
\(\rho_n^{-1}p_e\to W^\circ(x_e)\), \(\rho_n^{-1}P_e^{(j)}\to W^{(j)}(x_e)\) bounded away from
\(0,\infty\), then uniformly over \(\bbeta\in B\),
\(\rho_n^{-1}q_{\bbeta,e}\to W_{\bbeta}(x_e):=e^\alpha\prod_j\{W^{(j)}(x_e)\}^{\theta_j}\), and
\[
  \rho_n^{-1}R_n(\bbeta)=\int_\Omega\{W_{\bbeta}-W^\circ\log W_{\bbeta}\}\,d\mu+C(W^\circ)+o(1)
\]
uniformly, with \(C(W^\circ)=\int_\Omega\{W^\circ\log W^\circ-W^\circ\}\,d\mu\) free of \(\bbeta\).
Hence every limit point of \(\bbeta_n^\circ\) minimizes
\(\mathcal R(\bbeta)=\int_\Omega\{W_{\bbeta}-W^\circ\log W_{\bbeta}\}\,d\mu\); if the minimizer
\(\bbeta^\dagger\) is unique, \(\bbeta_n^\circ\to\bbeta^\dagger\). The sparse limit is a Bernoulli/KL
(Poisson-type) projection, not an \(L^2\)-projection.

\emph{(iii) Calibration equations and finite-bin reliability.} The projection satisfies
\(\sum_e(p_e-q_e^\circ)=0\) and \(\sum_e(p_e-q_e^\circ)u_e^{(j)}=0\), \(j=1,\ldots,J\).
With \(\Delta_n=\inf_{\bbeta\in B}\|p-q_{\bbeta}\|_{2,n}\), \(\widehat q_e=q_{\widehat\bbeta_n,e}\)
evaluated on an independent calibration fold \(A_e^{\rm cal}\sim\Bern(p_e)\), prediction-measurable
bins \(B_1,\ldots,B_K\), \(I_k=\{e:\widehat q_e\in B_k\}\), and
\(\ECE_K(\widehat q)=\sum_k\frac{|I_k|}{m_n}|\,|I_k|^{-1}\sum_{e\in I_k}(A_e^{\rm cal}-\widehat q_e)|\),
if \(K=o(N_n)\) then
\(\ECE_K(\widehat q)=O_p[\Delta_n+\rho_nN_n^{-1/2}+\rho_n\sqrt{K/N_n}]\); for the approximate
estimator the bound gains an additive \(\rho_n(a_n+r_n)\).

\emph{(iv) BMA collapse and the BPS log-score gap.} Let BMA place priors \(\pi_j>0\) on the agents
\(P^{(j)}=q_{\gamma_j}\), with posterior \(\Pi_n(j\mid A)\propto\pi_j\prod_e
(P_e^{(j)})^{A_e}(1-P_e^{(j)})^{1-A_e}\) and predictive \(q_e^{\rm BMA}=\sum_j\Pi_n(j\mid A)P_e^{(j)}\).
If there is a unique \(j_\star\) and \(\delta>0\) with
\(\rho_n^{-1}\{R_n(\gamma_j)-R_n(\gamma_{j_\star})\}\ge\delta\) for \(j\ne j_\star\), then
\(\Pi_n(j_\star\mid A)\to1\) in probability and
\(R_n(q^{\rm BMA})-R_n(q_{\widehat\bbeta_n})=G_n+o_p(\rho_n)\), where
\(G_n=\min_j R_n(\gamma_j)-R_n(\bbeta_n^\circ)=\min_j m_n^{-1}\sum_e\kl(q_e^\circ,P_e^{(j)})\ge0\).
Under (A2)--(A5), \(G_n\ge c_G\rho_n\min_j\|\bbeta_n^\circ-\gamma_j\|^2\); hence BPS strictly improves
on BMA whenever \(\liminf_n\min_j\|\bbeta_n^\circ-\gamma_j\|>0\).

\emph{(v) Convex stacking: the gap.} Let
\(\mathcal C_n=\{r_w:r_{w,e}=\sum_j w_jP_e^{(j)},\ w\in\Delta_J\}\),
\(d_{\mathcal C,n}^2=\inf_{r\in\mathcal C_n}m_n^{-1}\sum_e(r_e-p_e)^2/\{p_e(1-p_e)\}\),
and \(\varepsilon_{\rm BPS,n}=\inf_{\bbeta\in B}R_n(\bbeta)=R_n(\bbeta_n^\circ)\). Then there is
\(c_{\rm kl}>0\), depending only on (A2), with
\(\inf_{r\in\mathcal C_n}R_n(r)-\inf_{\bbeta\in B}R_n(\bbeta)\ge c_{\rm kl}d_{\mathcal C,n}^2-\varepsilon_{\rm BPS,n}\).
Consequently, if \(\varepsilon_{\rm BPS,n}\le\tfrac12 c_{\rm kl}d_{\mathcal C,n}^2\), the gap is at least
\(\tfrac12 c_{\rm kl}d_{\mathcal C,n}^2\); and if \(p=q_{\bar\bbeta}\) for some \(\bar\bbeta\in B\)
(so \(\varepsilon_{\rm BPS,n}=0\)), the gap is at least \(c_{\rm kl}d_{\mathcal C,n}^2\). Without the
approximation condition, no universal dominance of BPS over convex stacking follows.
\end{theorem}

\begin{proof}
All probability statements are conditional on \(\mathcal F_n\); the only randomness is the
independent array \(A_e\sim\Bern(p_e)\). Let \(b(\eta)=\log(1+e^\eta)\), so up to a
\(\bbeta\)-free additive term \(R_n(\bbeta)=m_n^{-1}\sum_e\{b(\eta_{\bbeta,e})-p_e\eta_{\bbeta,e}\}\),
\(\nabla R_n(\bbeta)=m_n^{-1}\sum_e(q_{\bbeta,e}-p_e)z_e\), and
\(\nabla^2R_n(\bbeta)=m_n^{-1}\sum_e q_{\bbeta,e}(1-q_{\bbeta,e})z_ez_e^\top\). By (A2),
\(q_{\bbeta,e}(1-q_{\bbeta,e})\ge c_1\rho_n\), so with (A5)
\(\lambda_{\min}\{\nabla^2R_n(\bbeta)\}\ge c_1\lambda_z\rho_n\): \(R_n\) is strictly convex and
\(\bbeta_n^\circ\) is the unique, interior minimizer, with \(\nabla R_n(\bbeta_n^\circ)=0\), i.e.
\(\sum_e(q_e^\circ-p_e)z_e=0\).
\emph{(i).} At \(\bbeta_n^\circ\),
\(\nabla\widehat R_n(\bbeta_n^\circ)=m_n^{-1}\sum_e(q_e^\circ-A_e)z_e=m_n^{-1}\sum_e(p_e-A_e)z_e\)
by the score equation; the summands are independent, mean zero, bounded, and by (A2)
\(p_e(1-p_e)\asymp\rho_n\), so
\(\mathbb E\|\nabla\widehat R_n(\bbeta_n^\circ)\|^2\le C\rho_n/m_n=C\rho_n^2/N_n\), giving
\(\|\nabla\widehat R_n(\bbeta_n^\circ)\|=O_p(\rho_nN_n^{-1/2})\). The empirical Hessian has the same
curvature lower bound \(c_2\rho_n\), so strong convexity gives
\(\tfrac{c_2\rho_n}{2}\|\widehat\bbeta_n-\bbeta_n^\circ\|^2\le\|\nabla\widehat R_n(\bbeta_n^\circ)\|
\|\widehat\bbeta_n-\bbeta_n^\circ\|\), hence \(\|\widehat\bbeta_n-\bbeta_n^\circ\|=O_p(N_n^{-1/2})\), and
\(\widehat\bbeta_n\) is interior with probability tending to one. Expanding
\(\nabla\widehat R_n(\widehat\bbeta_n)=0\) about \(\bbeta_n^\circ\),
\((m_n/N_n)\nabla^2\widehat R_n(\bar\bbeta_n)=H_n+o_p(1)\), so
\(\sqrt{N_n}(\widehat\bbeta_n-\bbeta_n^\circ)=H_n^{-1}N_n^{-1/2}\sum_e(A_e-p_e)z_e+o_p(1)\). The
summands are bounded, Lindeberg holds, the conditional variance is \(V_n\), and if \(H_n\to H\),
\(V_n\to V\), \(H\) nonsingular, then
\(\sqrt{N_n}(\widehat\bbeta_n-\bbeta_n^\circ)\Rightarrow N(0,H^{-1}VH^{-1})\); under correct
specification \(V_n=H_n\). For the approximate loss,
\(\|\nabla\widetilde R_n(\bbeta_n^\circ)\|=O_p\{\rho_nN_n^{-1/2}+\rho_n(a_n+r_n)\}\) and strong
convexity give \(\|\widetilde\bbeta_n-\bbeta_n^\circ\|=O_p(N_n^{-1/2}+a_n+r_n)\).
\emph{(ii).} Uniformly in \(e\), \(P_e^{(j)}=\rho_nW^{(j)}(x_e)+o(\rho_n)\) gives
\(u_e^{(j)}=\log W^{(j)}(x_e)+o(1)\), so
\(\eta_{\bbeta,e}=\log\rho_n+\alpha+\sum_j\theta_j\log W^{(j)}(x_e)+o(1)\) uniformly on compact \(B\),
and \(q_{\bbeta,e}=\rho_n e^\alpha\prod_j\{W^{(j)}(x_e)\}^{\theta_j}\{1+o(1)\}\), i.e.
\(\rho_n^{-1}q_{\bbeta,e}\to W_{\bbeta}(x_e)\). For \(a,b\) bounded away from \(0,\infty\),
\(\kl(\rho_na,\rho_nb)=\rho_n\{a\log(a/b)+b-a\}+O(\rho_n^2)\); applying this with \(a=W^\circ\),
\(b=W_{\bbeta}\), averaging, and using uniform empirical-graphon convergence gives the displayed limit,
whose \(\bbeta\)-free term is \(C(W^\circ)\). Uniform convergence on compact \(B\) gives the argmin
conclusion, and uniqueness gives \(\bbeta_n^\circ\to\bbeta^\dagger\).
\emph{(iii).} The calibration equations are the intercept and slope coordinates of
\(\sum_e(q_e^\circ-p_e)z_e=0\). By (A2), Bernoulli KL is equivalent to sparse squared error on
\([c\rho_n,C\rho_n]\): \(c_3\rho_n^{-1}(p-q)^2\le\kl(p,q)\le C_3\rho_n^{-1}(p-q)^2\). Since
\(\bbeta_n^\circ\) minimizes \(R_n\), \(\|p-q^\circ\|_{2,n}\le C\Delta_n\); and
\(|q_{\widehat\bbeta_n,e}-q_e^\circ|\le C\rho_n\|\widehat\bbeta_n-\bbeta_n^\circ\|\) gives
\(\|\widehat q-q^\circ\|_{2,n}=O_p(\rho_nN_n^{-1/2})\), so
\(\|p-\widehat q\|_{2,n}=O_p(\Delta_n+\rho_nN_n^{-1/2})\). Conditional on the fitted
probabilities the bins are fixed for the independent calibration array, so
\(\ECE_K(\widehat q)\le\|p-\widehat q\|_{2,n}+m_n^{-1}\sum_k|\sum_{e\in I_k}(A_e^{\rm cal}-p_e)|\);
Jensen and Cauchy--Schwarz bound the conditional expectation of the second term by
\(C\sqrt{K\rho_n/m_n}=C\rho_n\sqrt{K/N_n}\), and Markov converts this to the stated \(O_p\) bound. The
approximate estimator adds \(\rho_n(a_n+r_n)\).
\emph{(iv).} For \(j\ne j_\star\),
\(\log\frac{\Pi_n(j\mid A)}{\Pi_n(j_\star\mid A)}=\log\frac{\pi_j}{\pi_{j_\star}}+\sum_e A_e\log
\frac{P_e^{(j)}}{P_e^{(j_\star)}}+\sum_e(1-A_e)\log\frac{1-P_e^{(j)}}{1-P_e^{(j_\star)}}\), whose
conditional mean is \(-m_n\{R_n(\gamma_j)-R_n(\gamma_{j_\star})\}\) and conditional variance is
\(O(m_n\rho_n)=O(N_n)\) (edge log-ratios bounded, non-edge log-ratios \(O(\rho_n)\)). The separation
assumption gives \(m_n\{R_n(\gamma_j)-R_n(\gamma_{j_\star})\}\ge\delta m_n\rho_n\asymp\delta N_n\to\infty\),
so the log-ratio is \(-c_\delta N_n+O_p(N_n^{1/2})\to-\infty\) and \(\Pi_n(j_\star\mid A)\to1\). Then
\(\sup_e|q_e^{\rm BMA}-P_e^{(j_\star)}|=o_p(\rho_n)\), and since \(q\mapsto\kl(p_e,q)\) has
uniformly bounded derivative on the sparse envelope,
\(R_n(q^{\rm BMA})=\min_j R_n(\gamma_j)+o_p(\rho_n)\); also
\(R_n(q_{\widehat\bbeta_n})=R_n(\bbeta_n^\circ)+o_p(\rho_n)\). For any \(\gamma\in B\), the score
equation kills the cross term and
\(R_n(\gamma)-R_n(\bbeta_n^\circ)=m_n^{-1}\sum_e\kl(q_e^\circ,q_{\gamma,e})\), giving the formula for
\(G_n\). Bernoulli KL is strongly convex in the natural parameter on the envelope, so
\(\kl(q_{\bbeta_n^\circ,e},q_{\gamma_j,e})\ge c\rho_n\{(\bbeta_n^\circ-\gamma_j)^\top z_e\}^2\), and (A5)
gives \(G_n\ge c_G\rho_n\min_j\|\bbeta_n^\circ-\gamma_j\|^2\).
\emph{(v).} For \(p,q\in[c\rho_n,C\rho_n]\), Taylor of \(q\mapsto\kl(p,q)\) with comparability gives
\(c_{\rm kl}>0\) with \(\kl(p,q)\ge c_{\rm kl}(p-q)^2/\{p(1-p)\}\). Every convex stack lies in
\([c\rho_n,C\rho_n]\), so \(\inf_{r\in\mathcal C_n}R_n(r)\ge c_{\rm kl}d_{\mathcal C,n}^2\), and
subtracting \(\inf_{\bbeta\in B}R_n(\bbeta)=\varepsilon_{\rm BPS,n}\) gives the stated lower bound. The
remaining conclusions are immediate.
\end{proof}

\begin{figure}[t]
\centering
\includegraphics[width=\textwidth]{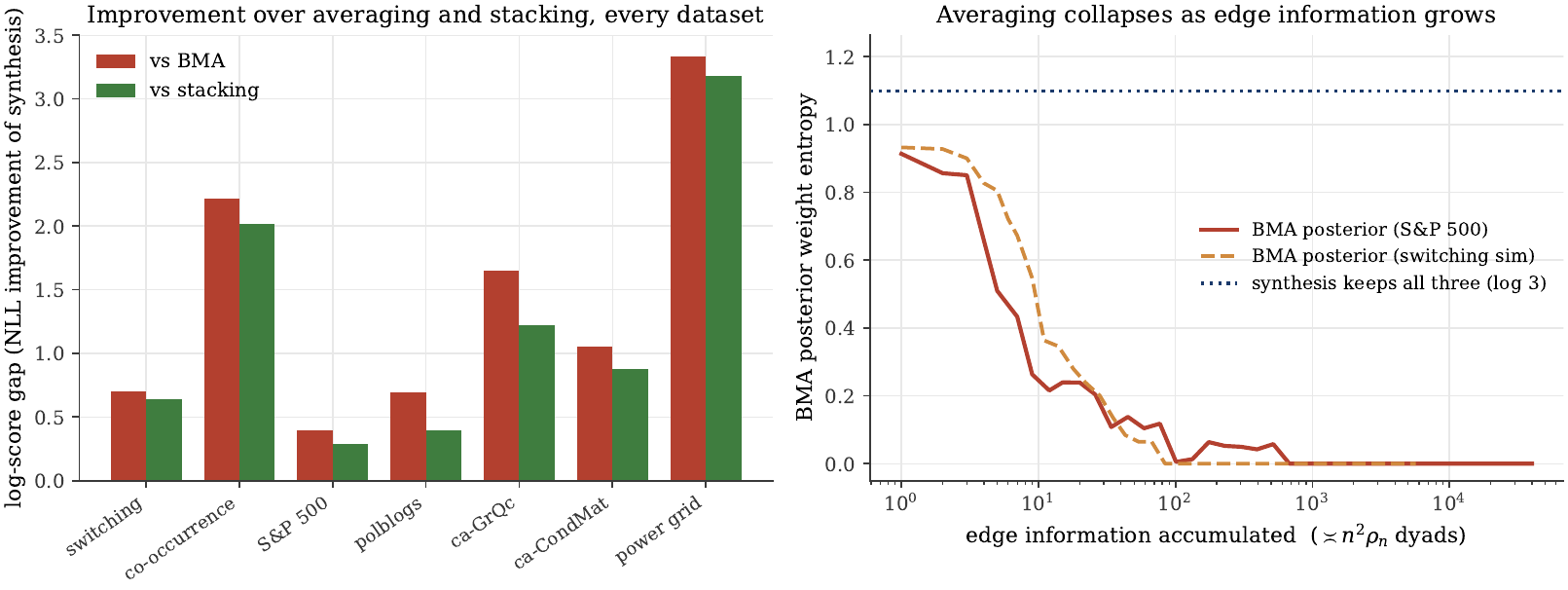}
\caption{Theorem~\ref{thm:t4}: dominance over model averaging under misspecification. (Left)
Log-score gap of the synthesis over Bayesian model averaging and stacking across seven theory-check
and benchmark studies (the switching and co-occurrence simulations, the S\&P 500 network, polblogs,
ca-GrQc, ca-CondMat, and the power-grid mesh), positive on each. (Right) The averaging posterior collapses as edge information accumulates.
Within a single snapshot, the entropy of the Bayesian model-averaging weights over the three
mechanism agents falls from $\log 3$ toward zero as the number of scored dyads grows to
$\asymp n^2\rho_n$, on both the switching simulation and the S\&P network, as the posterior
concentrates on the single Kullback--Leibler-closest agent. The synthesis keeps all three agents
active throughout; because its coefficients are predictive log-odds weights, not a posterior
distribution, the contrast is qualitative, and the reference line marks the maximum (uniform) entropy $\log 3$.}
\label{fig:t4}
\end{figure}

\section{Static cross-fitted affine-logit synthesis}
\label{supp:t5}

\begin{corollary}[Static cross-fitted affine-logit synthesis]
\label{cor:t5supp}
Let \(T=1\). Let \(\mathcal F_{\rm tr}\) be the sigma-field of the training data and fitted
cross-fitted agents. Conditional on \(\mathcal F_{\rm tr}\), let \(P_n\) be the evaluation law of a
validation/test dyad \((X,Y)\), \(Y\in\{0,1\}\), with \(\pi_n(x)=\mathbb P_n(Y=1\mid X=x)\),
\(\ell(y,q)=-y\log q-(1-y)\log(1-q)\), and \(R_n(q)=\mathbb E_n\{\ell(Y,q(X))\}\). Let
\(\widehat u_n(x)=\{\widehat u_{n1}(x),\ldots,\widehat u_{nJ}(x)\}^\top\),
\(\widetilde u_n(x)=(1,\widehat u_n(x)^\top)^\top\in\mathbb R^d\), \(d=J+1\), and for a known offset
\(b_n\) and \(\vartheta=(\alpha,\theta^\top)^\top\in\Theta_n\), set
\(\widehat q_\vartheta(x)=\sigma\{b_n+\vartheta^\top\widetilde u_n(x)\}\), \(\sigma(t)=1/(1+e^{-t})\).
Assume: (1) conditional on \(\mathcal F_{\rm tr}\), the validation dyads
\((X_i^{\rm val},Y_i^{\rm val})_{i=1}^{m_{\rm val}}\) are i.i.d.\ \(P_n\), with labels not used to build
\(\widehat u_n\); (2) on an event \(\mathcal E_n\) with \(\mathbb P(\mathcal E_n)\to1\), \(\Theta_n\) is
nonempty compact with \(\sup_{\vartheta\in\Theta_n}\|\vartheta\|_2\le M_n\),
\(\sup_x\|\widetilde u_n(x)\|_2\le K_n\sqrt d\),
\(\sup_{\vartheta\in\Theta_n,x}|b_n+\vartheta^\top\widetilde u_n(x)|\le B_n\), \(M_n,K_n,B_n\ge1\);
(3) the validation estimator is an approximate minimizer,
\(\widehat R_{\rm val}(\widehat q_{\widehat\vartheta})\le\inf_{\vartheta\in\Theta_n}\widehat R_{\rm val}
(\widehat q_\vartheta)+\varepsilon_{\rm opt,n}\), with \(\varepsilon_{\rm opt,n}=O_p(e_n)\).
Then, on \(\mathcal E_n\), for every \(0<\delta<1\), with conditional probability at least
\(1-\delta\),
\[
  R_n(\widehat q_{\widehat\vartheta})
  \le
  \inf_{\vartheta\in\Theta_n}R_n(\widehat q_\vartheta)
  +C\Big[M_nK_n\sqrt{\tfrac{J+1}{m_{\rm val}}}+B_n\sqrt{\tfrac{\log(2/\delta)}{m_{\rm val}}}\Big]
  +\varepsilon_{\rm opt,n},
\]
with \(C\) universal; hence
\(R_n(\widehat q_{\widehat\vartheta})\le\inf_{\vartheta}R_n(\widehat q_\vartheta)
+O_p[M_nK_n\sqrt{(J+1)/m_{\rm val}}+B_n/\sqrt{m_{\rm val}}+e_n]\). With a population-agent feature
\(u_n^0\) and \(q_\vartheta^0(x)=\sigma\{b_n+\vartheta^\top(1,u_n^0(x)^\top)^\top\}\), if
\(A_n=\sup_\vartheta|R_n(\widehat q_\vartheta)-R_n(q_\vartheta^0)|=O_p(a_n)\), then
\(R_n(\widehat q_{\widehat\vartheta})\le\inf_\vartheta R_n(q_\vartheta^0)+O_p[M_nK_n\sqrt{(J+1)/m_{\rm val}}
+B_n/\sqrt{m_{\rm val}}+e_n+a_n]\); for \(M_n,K_n,B_n=O(1)\) and
\(e_n=O[\sqrt{(J+1)/m_{\rm val}}+a_n]\) the rate is \(O_p[\sqrt{(J+1)/m_{\rm val}}+a_n]\). An independent
test sample adds \(B_n/\sqrt{m_{\rm test}}\).
If the class contains each single agent (\(\widehat q_{\vartheta^{(j)}}=\widehat q_j\)), then
\(\inf_\vartheta R_n(\widehat q_\vartheta)\le\min_j R_n(\widehat q_j)\). For nested spaces
\(\Theta_{J+1,n}\supseteq\{(\vartheta,0):\vartheta\in\Theta_{J,n}\}\), \(R_{J+1}^\star\le R_J^\star\)
where \(R_J^\star=\inf_{\vartheta\in\Theta_{J,n}}R_n(q_\vartheta^0)\); and with
\(q_J^\circ=q_{\vartheta_J^\circ}^0\), \(\mathcal L_J=\operatorname{span}\{1,u_1^0,\ldots,u_J^0\}\),
\(h\in L_2(P_n)\) bounded, \(h_\perp=h-\Pi_{\mathcal L_J}h\), if the enlarged space contains a
neighborhood of \((\vartheta_J^\circ,0)\) and \(\mathbb E_n[\{q_J^\circ(X)-\pi_n(X)\}h_\perp(X)]\ne0\),
then \(R_{J+1}^\star<R_J^\star\).
Under case-control sampling with \(\mathbb P(S=1\mid Y=1,X=x)=s_1\),
\(\mathbb P(S=1\mid Y=0,X=x)=s_0\in(0,1]\), one has
\(\operatorname{logit}\pi_{\rm cc}(x)=\operatorname{logit}\pi_{\rm pop}(x)+\log(s_1/s_0)\); a fitted
case-control forecast with \(\|\operatorname{logit}p_{{\rm cc},n}-\operatorname{logit}\pi_{\rm cc}\|_{L_r(\mu)}\to0\)
yields \(\|p_{{\rm pop},n}-\pi_{\rm pop}\|_{L_r(\mu)}\to0\) after the offset shift
\(b_{\rm pop,n}=b_{\rm cc,n}-\log(s_1/s_0)\), and the intercept enforces average calibration
\(\mathbb E_{\rm cc}\{Y-q_{\rm cc}^\circ(X)\}=0\).
Finally, under the signal-free condition \(\mathbb P_n(Y=1\mid\widehat u_n(X))=\pi_n=\mathbb P_n(Y=1)\),
with the base-rate forecast in the class (\(b_n+\alpha_\pi=\operatorname{logit}\pi_n\)) and nondegenerate
design (\(b^\top\widehat u_n(X)\) a.s.\ constant \(\Rightarrow b=0\)), every population minimizer has
\(q^\circ\equiv\pi_n\) and slope \(\theta^\circ=0\); if
\(M_nK_n\sqrt{(J+1)/m_{\rm val}}+B_n/\sqrt{m_{\rm val}}+e_n\to0\) then
\(\mathbb E_n|\widehat q_{\widehat\vartheta}(X)-\pi_n|\to0\) in probability, no \(\widehat u_n\)-measurable score ranks the dyads above chance, the probability that a positive dyad outscores a negative one being \(1/2\) (so its empirical counterpart \(\to1/2\)), and the fixed-bin
\(\ECE_K(\widehat q_{\widehat\vartheta})\to0\) in probability, with empirical-vs-population fixed-bin ECE
agreeing to \(O_p(\sqrt{K/m_{\rm test}})\).
\end{corollary}

\begin{proof}
Conditional on \(\mathcal F_{\rm tr}\), the validation criterion is an empirical process over a fixed
function class. With \(\eta_\vartheta(x)=b_n+\vartheta^\top\widetilde u_n(x)\),
\(\partial_\eta\ell\{y,\sigma(\eta)\}=\sigma(\eta)-y\in[-1,1]\), so \(\eta\mapsto\ell\{y,\sigma(\eta)\}\)
is \(1\)-Lipschitz, and on \(\mathcal E_n\), \(0\le\ell\{y,\widehat q_\vartheta(x)\}\le\log2+B_n\le CB_n\).
Let \(\Delta_{\rm val}=\sup_\vartheta|\widehat R_{\rm val}(\widehat q_\vartheta)-R_n(\widehat q_\vartheta)|\).
Symmetrization gives \(\mathbb E\{\Delta_{\rm val}\mid\mathcal F_{\rm tr}\}\le 2\mathbb E_\varepsilon
\sup_\vartheta|m_{\rm val}^{-1}\sum_i\varepsilon_i\ell\{Y_i^{\rm val},\widehat q_\vartheta(X_i^{\rm val})\}|\).
Fixing \(\vartheta_0\) and splitting off its constant contribution (expectation \(\le CB_n/\sqrt{m_{\rm val}}\)),
the centered part has \(1\)-Lipschitz increments \(\phi_i\) with \(\phi_i(0)=0\); the Ledoux--Talagrand
contraction inequality bounds it by \(C\mathbb E_\varepsilon\sup_\vartheta|m_{\rm val}^{-1}\sum_i\varepsilon_i
(\vartheta-\vartheta_0)^\top\widetilde u_n(X_i^{\rm val})|\), and Cauchy--Schwarz with
\(\|\vartheta-\vartheta_0\|_2\le2M_n\) and \(\|\widetilde u_n\|_2\le K_n\sqrt d\) bounds this by
\(2M_nK_n\sqrt{d/m_{\rm val}}\). Thus
\(\mathbb E\{\Delta_{\rm val}\mid\mathcal F_{\rm tr}\}\le C[M_nK_n\sqrt{d/m_{\rm val}}+B_n/\sqrt{m_{\rm val}}]\).
One validation point changes \(\Delta_{\rm val}\) by at most \(CB_n/m_{\rm val}\), so McDiarmid gives, with
conditional probability \(\ge1-\delta\),
\(\Delta_{\rm val}\le C[M_nK_n\sqrt{d/m_{\rm val}}+B_n\sqrt{\log(2/\delta)/m_{\rm val}}]\). Approximate
minimization gives, for every \(\vartheta\),
\(R_n(\widehat q_{\widehat\vartheta})\le\widehat R_{\rm val}(\widehat q_{\widehat\vartheta})+\Delta_{\rm val}
\le\widehat R_{\rm val}(\widehat q_\vartheta)+\Delta_{\rm val}+\varepsilon_{\rm opt,n}
\le R_n(\widehat q_\vartheta)+2\Delta_{\rm val}+\varepsilon_{\rm opt,n}\); taking the infimum and
integrating over \(\mathcal F_{\rm tr}\) gives the oracle inequality and its \(O_p\) form. The
population-agent version follows from
\(\inf_\vartheta R_n(\widehat q_\vartheta)\le\inf_\vartheta R_n(q_\vartheta^0)+A_n\); the bounded case and
the test-sample term are immediate (Hoeffding for bounded test losses).
Single-agent inclusion gives \(\inf_\vartheta R_n(\widehat q_\vartheta)\le\min_j R_n(\widehat q_j)\).
Nestedness gives \(R_{J+1}^\star\le R_J^\star\). For strict improvement, with
\(q_\gamma(x)=\sigma\{b_n+(\vartheta_J^\circ)^\top(1,u_1^0,\ldots,u_J^0)^\top+\gamma h(x)\}\),
differentiation under the expectation gives
\(\frac{d}{d\gamma}R_n(q_\gamma)|_{0}=\mathbb E_n[\{q_J^\circ(X)-\pi_n(X)\}h(X)]\); the interior
first-order conditions give \(\mathbb E_n[\{q_J^\circ-\pi_n\}g]=0\) for \(g\in\mathcal L_J\), so the
derivative equals \(\mathbb E_n[\{q_J^\circ-\pi_n\}h_\perp]\), and if nonzero a small \(\gamma\) of the
opposite sign strictly lowers the risk, giving \(R_{J+1}^\star<R_J^\star\).
For case-control, Bayes' rule gives
\(\pi_{\rm cc}(x)=s_1\pi_{\rm pop}(x)/\{s_1\pi_{\rm pop}(x)+s_0(1-\pi_{\rm pop}(x))\}\), hence the logit
shift \(\log(s_1/s_0)\); since \(\sigma\) is \(1\)-Lipschitz the \(L_r\) consistency transfers, and the
intercept score equation \(\mathbb E_{\rm cc}\{Y-q_{\rm cc}^\circ(X)\}=0\) is average calibration.
Under the signal-free condition, for any \(\widehat u_n\)-measurable \(q\),
\(R_n(q)=H(\pi_n)+\mathbb E_n[\operatorname{KL}(\Bern(\pi_n)\|\Bern(q(X)))]\) with
\(H(\pi)=-\pi\log\pi-(1-\pi)\log(1-\pi)\); the KL term is nonnegative and zero iff \(q\equiv\pi_n\), so
every minimizer has \(q^\circ\equiv\pi_n\), and nondegeneracy forces \(\theta^\circ=0\). The oracle
inequality gives \(R_n(\widehat q_{\widehat\vartheta})-H(\pi_n)\to0\), and Pinsker
(\(\operatorname{KL}(\Bern(\pi_n)\|\Bern(q))\ge2(q-\pi_n)^2\)) gives
\(\mathbb E_n[\{\widehat q_{\widehat\vartheta}-\pi_n\}^2]\to0\), hence
\(\mathbb E_n|\widehat q_{\widehat\vartheta}-\pi_n|\to0\) by Cauchy--Schwarz. Signal-freeness gives
\(Y\perp S(X)\) for \(\widehat u_n\)-measurable \(S\), so the conditional score laws coincide and
\(\mathbb P(S_1>S_0)+\tfrac12\mathbb P(S_1=S_0)=\tfrac12\); the empirical version is the
two-sample \(U\)-statistic and converges to \(1/2\). Finally, since
\(\mathbb E_n[Y\mid q(X)]=\pi_n\), \(\ECE_K(q)\le\mathbb E_n|q(X)-\pi_n|\to0\); and for bounded
\(Z_{i,k}=\{Y_i^{\rm test}-q(X_i^{\rm test})\}\1\{q(X_i^{\rm test})\in I_k\}\) with disjoint bins,
\(\sum_k\operatorname{Var}_n(Z_{1,k})\le1\), so Cauchy--Schwarz gives
\(|\widehat{\ECE}_K(q)-\ECE_K(q)|=O_p(\sqrt{K/m_{\rm test}})\).
\end{proof}

\section{Cross-fold orthogonal unconditional limit (Theorem 1c)}
\label{supp:s6}
This section gives the full statement and proof of Theorem~1c of the main text. It specializes
the single-snapshot setting of Section~\ref{supp:t1} to a single snapshot $t$, suppressed from the notation,
and constructs a cross-fitted, Neyman-orthogonal version of the synthesis estimator whose normal
limit holds unconditionally over the dyad draw, the random fold split, and the agent refits.
Throughout, $\sigma$ is the logistic link and the sparse offset $\ell=\log\rho_n$ is carried
explicitly.

Fix a snapshot \(t\) and suppress \(t\) from the notation.  Let
\(\mathcal E_n\) be the dyad set used in the one-snapshot synthesis
update, \(m_n=|\mathcal E_n|\), and let
\[
  I_1,\ldots,I_K
\]
be a random balanced partition of \(\mathcal E_n\), with fixed
\(K<\infty\).  Let
\[
  \mathcal G_n=\sigma(\mathcal D_n,I_1,\ldots,I_K)
\]
be the sigma-field generated by the network design and the dyad split.
The design \(\mathcal D_n\) contains the node-level latent quantities,
the dyad labels, and any nonrandom oracle agent limits.
Conditional on \(\mathcal G_n\), assume that
\[
  A_e\sim {\rm Bernoulli}(p_e),
  \qquad e\in\mathcal E_n,
\]
independently over dyads.  For each fold \(k\), let
\[
  \mathcal T_{n,k}
  =
  \sigma\bigl(\mathcal G_n,\{A_f:f\notin I_k\}\bigr)
\]
be the training sigma-field for fold \(k\).  The cross-fitted nuisance
estimate used on fold \(k\) is denoted by
\(\widehat\xi^{-k}\),
assumed \(\mathcal T_{n,k}\)-measurable, so that the
validation dyads in \(I_k\) are not used to construct their own nuisance
estimate.
The nuisance object is
\[
  \xi=(\eta,\ell),
  \qquad
  \eta=
  \left(
    u^{(1)}_{\cdot},\ldots,u^{(J)}_{\cdot}
  \right),
  \qquad
  \ell=\log\rho_n .
\]
If the sparse offset \(\ell\) is fixed or conditioned upon, it may be
omitted from \(\xi\).  If \(\ell\) is estimated from the same snapshot,
then the fold-specific estimate \(\widehat\ell^{-k}\) must be included
in \(\widehat\xi^{-k}\), or its influence function must be appended to
the orthogonal score below.
For
\[
  \bbeta=(\alpha,\theta_1,\ldots,\theta_J)^\top\in\mathbb R^d,
  \qquad
  d=J+1,
\]
define
\(
  z_e(\xi)
  =
  ( 1,u^{(1)}_e,\ldots,u^{(J)}_e )^\top
\)
and the sparse-offset BPS probability
\[
  q_e(\bbeta,\xi)
  =
  \sigma\left\{
    \ell+\bbeta^\top z_e(\xi)
  \right\}.
\]
Let \(\xi^\circ=(\eta^\circ,\ell^\circ)\)
be the oracle nuisance, and let \(\bbeta^\circ\) be the ordinary
sparse-offset BPS projection target.  Write
\[
  q^\circ_e=q_e(\bbeta^\circ,\xi^\circ),
  \qquad
  w^\circ_e=q^\circ_e(1-q^\circ_e),
  \qquad
  N_n=\sum_{e\in\mathcal E_n}w^\circ_e .
\]
For each fold,
\(
  N_{n,k}=\sum_{e\in I_k}w^\circ_e
\)
with \(N_{n,k}/N_n\to\pi_k\in(0,1)\).
The ordinary BPS score is
\(
  s_e(\bbeta,\xi)
  =
  z_e(\xi)\{A_e-q_e(\bbeta,\xi)\}
\),
the ordinary BPS population moment is
\[
  \Psi^s_n(\bbeta,\xi)
  =
  N_n^{-1}
  \sum_{e\in\mathcal E_n}
  E\{s_e(\bbeta,\xi)\mid\mathcal G_n\},
\]
and the target \(\bbeta^\circ\) satisfies \(\Psi^s_n(\bbeta^\circ,\xi^\circ)=0\).
For each fold \(k\), let \(c_{k,e}(\bbeta,\xi)\)
be a fold-specific correction, and define the foldwise orthogonal BPS
score
\[
  \psi^\perp_{k,e}(\bbeta,\xi)
  =
  s_e(\bbeta,\xi)-c_{k,e}(\bbeta,\xi),
  \qquad e\in I_k.
\]
Define the foldwise and aggregate orthogonal population moments
\[
  \Psi^\perp_{n,k}(\bbeta,\xi)
  =
  N_{n,k}^{-1}
  \sum_{e\in I_k}
  E\{\psi^\perp_{k,e}(\bbeta,\xi)\mid\mathcal G_n\},
  \qquad
  \Psi^\perp_n(\bbeta,\xi)
  =
  N_n^{-1}
  \sum_{k=1}^K
  \sum_{e\in I_k}
  E\{\psi^\perp_{k,e}(\bbeta,\xi)\mid\mathcal G_n\}.
\]
The cross-fold orthogonal estimating equation is
\[
  S^\perp_n(\bbeta)
  =
  \sum_{k=1}^K
  \sum_{e\in I_k}
  \psi^\perp_{k,e}(\bbeta,\widehat\xi^{-k}).
\]

{\renewcommand{\thetheorem}{1c}%
\begin{theorem}[Cross-fold orthogonal unconditional BPS limit]
\label{thm:crossfold-orthogonal-unconditional-bps}
Assume the following conditions.
\begin{enumerate}
\item[\textnormal{(C1)}]
\textnormal{Sparse boundedness and information growth.}
There exist constants \(0<c<C<\infty\) such that, uniformly for
\(\bbeta\) in a neighborhood of \(\bbeta^\circ\) and \(\xi\) in a
neighborhood of \(\xi^\circ\),
\(
  p_e\le C\rho_n
\),
\(
  c\rho_n
  \le
  q_e(\bbeta,\xi)
  \le
  C\rho_n
  \le
  1/2
\),
and \(\sup_{e,\xi}\|z_e(\xi)\|\le C\).
Moreover,
\[
  m_n\asymp n^2,
  \qquad
  N_n=\sum_e w^\circ_e\asymp m_n\rho_n\asymp n^2\rho_n,
  \qquad
  N_n\to\infty .
\]
\item[\textnormal{(C2)}]
\textnormal{Target preservation.}
For every \(\bbeta\) in a neighborhood of \(\bbeta^\circ\),
\(
  \Psi^\perp_n(\bbeta,\xi^\circ)
  =
  \Psi^s_n(\bbeta,\xi^\circ)
\),
equivalently
\(
  N_n^{-1}
  \sum_{k}
  \sum_{e\in I_k}
  E\{c_{k,e}(\bbeta,\xi^\circ)\mid\mathcal G_n\}
  =
  0
\).
Thus the orthogonal score has the same population target
\(\bbeta^\circ\) as the ordinary BPS score.
\item[\textnormal{(C3)}]
\textnormal{Foldwise Neyman orthogonality and second-order smoothness.}
For every fold \(k\) and every admissible nuisance direction \(h\),
\(
  \partial_\xi
  \Psi^\perp_{n,k}(\bbeta^\circ,\xi^\circ)[h]=0
\).
Define the sparse Fisher nuisance norm
\[
  \|h\|_{\mathcal H_n}^2
  =
  \max_{1\le k\le K}
  N_{n,k}^{-1}
  \sum_{e\in I_k}
  w^\circ_e\|h_e\|^2 .
\]
There exists \(C<\infty\) such that, uniformly in \(k\),
\[
  \left\|
    \Psi^\perp_{n,k}(\bbeta^\circ,\xi)
    -
    \Psi^\perp_{n,k}(\bbeta^\circ,\xi^\circ)
    -
    \partial_\xi
    \Psi^\perp_{n,k}(\bbeta^\circ,\xi^\circ)
    [\xi-\xi^\circ]
  \right\|
  \le
  C\|\xi-\xi^\circ\|_{\mathcal H_n}^2 .
\]
\item[\textnormal{(C4)}]
\textnormal{Construction of the correction.}
The correction \(c_{k,e}\) satisfies \textnormal{(C2)} and
\textnormal{(C3)}.  A sufficient finite-dimensional construction is the
following.  Suppose \(\xi=\xi(\gamma)\), and let
\(g_{k,e}(\gamma)\) be a foldwise nuisance score satisfying
\(
  \Gamma_{n,k}(\gamma^\circ)
  =
  N_{n,k}^{-1}
  \sum_{e\in I_k}
  E\{g_{k,e}(\gamma^\circ)\mid\mathcal G_n\}
  =
  0
\),
with nonsingular derivative
\(
  \dot\Gamma_{\gamma,n,k}
  =
  \partial_\gamma\Gamma_{n,k}(\gamma^\circ)
\).
Let
\(
  \dot\Psi^s_{\gamma,n,k}
  =
  \partial_\gamma
  \Psi^s_{n,k}(\bbeta^\circ,\gamma^\circ)
\),
where
\(
  \Psi^s_{n,k}(\bbeta,\gamma)
  =
  N_{n,k}^{-1}
  \sum_{e\in I_k}
  E\{s_e(\bbeta,\gamma)\mid\mathcal G_n\}
\).
Then
\[
  c_{k,e}(\bbeta,\gamma)
  =
  \dot\Psi^s_{\gamma,n,k}
  \dot\Gamma_{\gamma,n,k}^{-1}
  g_{k,e}(\gamma)
\]
is target-preserving at \(\gamma^\circ\) and satisfies the foldwise
orthogonality equation.
\item[\textnormal{(C5)}]
\textnormal{Cross-fitted nuisance rate.}
For
\(
  r_n=
  \max_{k}
  \|\widehat\xi^{-k}-\xi^\circ\|_{\mathcal H_n}
\),
we have \(r_n=o_p(1)\) and \(\sqrt{N_n}\,r_n^2=o_p(1)\).
If the fold-specific nuisance population target differs from
\(\xi^\circ\), the difference is included in \(r_n\).  If
\(c_{k,e}\) is estimated, the estimation error of the correction is also
included in \(r_n\).
\item[\textnormal{(C6)}]
\textnormal{Cross-fit stochastic equicontinuity.}
Let
\(
  \Delta_{k,e}
  =
  \psi^\perp_{k,e}(\bbeta^\circ,\widehat\xi^{-k})
  -
  \psi^\perp_{k,e}(\bbeta^\circ,\xi^\circ)
\).
Then
\[
  N_n^{-1/2}
  \sum_{k=1}^K
  \sum_{e\in I_k}
  \left[
    \Delta_{k,e}
    -
    E\{\Delta_{k,e}\mid\mathcal T_{n,k}\}
  \right]
  =
  o_p(1).
\]
\item[\textnormal{(C7)}]
\textnormal{Local root, local smoothness, and no aliasing.}
There exists a measurable local root \(\widetilde\bbeta\) such that
\(
  S^\perp_n(\widetilde\bbeta)=0
\)
and \(\widetilde\bbeta-\bbeta^\circ=o_p(1)\).
Let
\(
  H^\perp_n
  =
  -
  \partial_{\bbeta}
  \Psi^\perp_n(\bbeta^\circ,\xi^\circ)
\).
There are constants \(0<\lambda<\Lambda<\infty\) such that
\(
  \lambda
  \le
  \lambda_{\min}(H^\perp_n)
  \le
  \lambda_{\max}(H^\perp_n)
  \le
  \Lambda
\).
Moreover, with
\[
  \bar H_n
  =
  -
  N_n^{-1}
  \int_0^1
  \partial_{\bbeta}
  S^\perp_n
  \left(
    \bbeta^\circ+s(\widetilde\bbeta-\bbeta^\circ)
  \right)
  ds,
\]
we have \(\bar H_n-H^\perp_n=o_p(1)\).
\item[\textnormal{(C8)}]
\textnormal{Oracle dyadic CLT.}
Define the centered oracle orthogonal score
\(
  \bar\psi^\perp_{k,e}
  =
  \psi^\perp_{k,e}(\bbeta^\circ,\xi^\circ)
  -
  E\{
    \psi^\perp_{k,e}(\bbeta^\circ,\xi^\circ)
    \mid
    \mathcal G_n
  \}
\),
and let
\[
  \Omega^\perp_n
  =
  N_n^{-1}
  \sum_{k=1}^K
  \sum_{e\in I_k}
  \operatorname{Var}
  \{
    \psi^\perp_{k,e}(\bbeta^\circ,\xi^\circ)
    \mid
    \mathcal G_n
  \}.
\]
The eigenvalues of \(\Omega^\perp_n\) are bounded away from zero and
infinity, and for every \(\epsilon>0\),
\[
  N_n^{-1}
  \sum_{k=1}^K
  \sum_{e\in I_k}
  E\left[
    \|\bar\psi^\perp_{k,e}\|^2
    \mathbf 1
    \{
      \|\bar\psi^\perp_{k,e}\|>\epsilon\sqrt{N_n}
    \}
    \mid
    \mathcal G_n
  \right]
  \to 0
\]
in probability.
\item[\textnormal{(C9)}]
\textnormal{Variance estimation.}
Let
\(
  \widehat H^\perp_n
  =
  -
  N_n^{-1}
  \sum_{k}
  \sum_{e\in I_k}
  \partial_{\bbeta}
  \psi^\perp_{k,e}(\widetilde\bbeta,\widehat\xi^{-k})
\),
and assume \(\widehat H^\perp_n-H^\perp_n=o_p(1)\)
and that a foldwise covariance estimator
\(\widehat\Omega^\perp_n\) satisfies
\(\widehat\Omega^\perp_n-\Omega^\perp_n=o_p(1)\).
Under correct Bernoulli-BPS specification with per-dyad mean-zero
scores, the usual outer-product estimator may be used.  Under
fixed-design misspecification, the naive outer product is not generally
consistent; then \(\widehat\Omega^\perp_n\) must consistently estimate
the conditional Bernoulli variance \(\Omega^\perp_n\).
\end{enumerate}
Then
\[
  \sqrt{N_n}
  (\widetilde\bbeta-\bbeta^\circ)
  =
  (H^\perp_n)^{-1}
  N_n^{-1/2}
  \sum_{k=1}^K
  \sum_{e\in I_k}
  \bar\psi^\perp_{k,e}
  +
  o_p(1).
\]
Consequently, if \(H^\perp_n\to H^\perp\) and \(\Omega^\perp_n\to\Omega^\perp\),
then, unconditionally over the dyad draw, the random split, and the
agent refits,
\[
  \sqrt{N_n}
  (\widetilde\bbeta-\bbeta^\circ)
  \Rightarrow
  N\left(
    0,
    (H^\perp)^{-1}
    \Omega^\perp
    (H^\perp)^{-\top}
  \right).
\]
Equivalently, without deterministic matrix limits,
\(
  [
    (H^\perp_n)^{-1}
    \Omega^\perp_n
    (H^\perp_n)^{-\top}
  ]^{-1/2}
  \sqrt{N_n}
  (\widetilde\bbeta-\bbeta^\circ)
  \Rightarrow
  N(0,I_d)
\).
Moreover, with
\(
  \widehat V_n
  =
  (\widehat H^\perp_n)^{-1}
  \widehat\Omega^\perp_n
  (\widehat H^\perp_n)^{-\top}
\),
for every fixed contrast \(a\in\mathbb R^d\) with
\(
  0<\liminf_n a^\top
  (H^\perp_n)^{-1}\Omega^\perp_n(H^\perp_n)^{-\top}a
\),
\[
  \frac{
    \sqrt{N_n}\,a^\top(\widetilde\bbeta-\bbeta^\circ)
  }{
    \sqrt{a^\top\widehat V_n a}
  }
  \Rightarrow
  N(0,1),
\]
so the Wald interval
\(
  a^\top\widetilde\bbeta
  \pm
  z_{1-\alpha/2}
  \sqrt{
    a^\top\widehat V_n a/N_n
  }
\)
has coverage \(1-\alpha+o(1)\).  If
\(\textnormal{(C1)}\)--\(\textnormal{(C9)}\) hold uniformly over a
class \(\mathcal P_n\), including the conditional CLT and variance
consistency, then the coverage is uniform over \(\mathcal P_n\).
\end{theorem}%
\addtocounter{theorem}{-1}}

\begin{proof}
By \(\textnormal{(C2)}\) and the definition of \(\bbeta^\circ\),
\(
  \Psi^\perp_n(\bbeta^\circ,\xi^\circ)
  =
  \Psi^s_n(\bbeta^\circ,\xi^\circ)
  =
  0
\),
equivalently
\(
  \sum_{k}
  \sum_{e\in I_k}
  E\{
    \psi^\perp_{k,e}(\bbeta^\circ,\xi^\circ)
    \mid
    \mathcal G_n
  \}
  =
  0
\).
This is an aggregate projection equation; it does not require
per-dyad mean-zero scores.
We first expand the estimating equation at the target.  With
\(
  \Delta_{k,e}
  =
  \psi^\perp_{k,e}(\bbeta^\circ,\widehat\xi^{-k})
  -
  \psi^\perp_{k,e}(\bbeta^\circ,\xi^\circ)
\),
\[
  N_n^{-1/2}S^\perp_n(\bbeta^\circ)
  =
  N_n^{-1/2}
  \sum_{k=1}^K
  \sum_{e\in I_k}
  \bar\psi^\perp_{k,e}
  +
  R_{1n}
  +
  R_{2n},
\]
where
\(
  R_{1n}
  =
  N_n^{-1/2}
  \sum_{k}
  \sum_{e\in I_k}
  [
    \Delta_{k,e}
    -
    E\{\Delta_{k,e}\mid\mathcal T_{n,k}\}
  ]
\)
and
\(
  R_{2n}
  =
  N_n^{-1/2}
  \sum_{k}
  \sum_{e\in I_k}
  E\{\Delta_{k,e}\mid\mathcal T_{n,k}\}
\).
The oracle deterministic mean vanishes because of the aggregate moment
equation above.
By \(\textnormal{(C6)}\), \(R_{1n}=o_p(1)\).
It remains to control \(R_{2n}\).  Conditional on \(\mathcal T_{n,k}\),
the nuisance estimate \(\widehat\xi^{-k}\) is fixed and the validation
dyads in \(I_k\) are independent Bernoulli draws.  Hence
\[
  \sum_{e\in I_k}
  E\{\Delta_{k,e}\mid\mathcal T_{n,k}\}
  =
  N_{n,k}
  \left\{
    \Psi^\perp_{n,k}
    (\bbeta^\circ,\widehat\xi^{-k})
    -
    \Psi^\perp_{n,k}
    (\bbeta^\circ,\xi^\circ)
  \right\}.
\]
By the second-order expansion in \(\textnormal{(C3)}\),
\[
  \Psi^\perp_{n,k}
    (\bbeta^\circ,\widehat\xi^{-k})
    -
    \Psi^\perp_{n,k}
    (\bbeta^\circ,\xi^\circ)
  =
    \partial_\xi
    \Psi^\perp_{n,k}(\bbeta^\circ,\xi^\circ)
    [\widehat\xi^{-k}-\xi^\circ]
    +
    O_p
    (
      \|\widehat\xi^{-k}-\xi^\circ\|_{\mathcal H_n}^2
    ).
\]
The first-order term is zero by foldwise Neyman orthogonality, so
\(
  \|
    \Psi^\perp_{n,k}
    (\bbeta^\circ,\widehat\xi^{-k})
    -
    \Psi^\perp_{n,k}
    (\bbeta^\circ,\xi^\circ)
  \|
  =
  O_p(r_n^2)
\)
uniformly in \(k\).  Since \(\sum_k N_{n,k}=N_n\),
\[
  \|R_{2n}\|
  \le
  N_n^{-1/2}
  \sum_{k=1}^K
  N_{n,k}O_p(r_n^2)
  =
  O_p(\sqrt{N_n}r_n^2)
  =
  o_p(1)
\]
by \(\textnormal{(C5)}\).  Consequently,
\(
  N_n^{-1/2}S^\perp_n(\bbeta^\circ)
  =
  N_n^{-1/2}
  \sum_{k}
  \sum_{e\in I_k}
  \bar\psi^\perp_{k,e}
  +
  o_p(1)
\).
By \(\textnormal{(C8)}\), conditional on \(\mathcal G_n\),
\(
  N_n^{-1/2}
  \sum_{k}
  \sum_{e\in I_k}
  \bar\psi^\perp_{k,e}
  \Rightarrow
  N(0,\Omega^\perp_n)
\)
in the triangular-array sense; if \(\Omega^\perp_n\to\Omega^\perp\) this is
\(N(0,\Omega^\perp)\), and the same holds unconditionally because the
conditional characteriztic functions converge in probability and are
uniformly bounded.
Next, since \(S^\perp_n(\widetilde\bbeta)=0\), the integral form of
Taylor's theorem gives
\[
  0
  =
  S^\perp_n(\bbeta^\circ)
  +
  \left[
    \int_0^1
    \partial_{\bbeta}
    S^\perp_n
    (
      \bbeta^\circ+s(\widetilde\bbeta-\bbeta^\circ)
    )
    ds
  \right]
  (\widetilde\bbeta-\bbeta^\circ)
  =
  S^\perp_n(\bbeta^\circ)
  -
  N_n\bar H_n(\widetilde\bbeta-\bbeta^\circ),
\]
so
\(
  \sqrt{N_n}(\widetilde\bbeta-\bbeta^\circ)
  =
  \bar H_n^{-1}
  N_n^{-1/2}S^\perp_n(\bbeta^\circ)
\).
By \(\textnormal{(C7)}\), \(\bar H_n=H^\perp_n+o_p(1)\) with \(H^\perp_n\)
uniformly nonsingular, so \(\bar H_n^{-1}=(H^\perp_n)^{-1}+o_p(1)\), and
\[
  \sqrt{N_n}
  (\widetilde\bbeta-\bbeta^\circ)
  =
  (H^\perp_n)^{-1}
  N_n^{-1/2}
  \sum_{k=1}^K
  \sum_{e\in I_k}
  \bar\psi^\perp_{k,e}
  +
  o_p(1).
\]
If \(H^\perp_n\to H^\perp\) and \(\Omega^\perp_n\to\Omega^\perp\), Slutsky's
theorem yields the stated normal limit; otherwise the asymptotic linear
representation and \(\textnormal{(C8)}\) give the self-normalized form.
Finally, by \(\textnormal{(C9)}\), \(\widehat V_n=(H^\perp_n)^{-1}\Omega^\perp_n
(H^\perp_n)^{-\top}+o_p(1)\), so for every fixed contrast \(a\) with
nondegenerate limiting variance Slutsky's theorem gives the scalar normal
limit, and the Wald interval has pointwise coverage \(1-\alpha+o(1)\).  If
the assumptions hold uniformly over \(\mathcal P_n\), the same argument with
uniform versions of the CLT, the nuisance-rate bounds, nonsingularity, and
variance consistency gives uniform coverage.
\end{proof}

\begin{remark}[Node-clustered variance under shared-node dependence]
\label{rmk:dyadic-robust}
Condition \textnormal{(C9)} takes the meat \(\widehat\Omega^\perp_n\) of the sandwich to be a consistent
estimator of the score covariance \(\Omega^\perp_n\). When the dyads are conditionally independent the
per-dyad outer product \(N_n^{-1}\sum_e \widehat s_e \widehat s_e^\top\), with
\(\widehat s_e=(Y_e-\widehat p_e)z_e\) the estimated score contribution of dyad \(e\), is consistent. When
two dyads that share a node are dependent, as for the correlation-thresholded networks of the main text
where edges sharing an asset are dependent, the score covariances between node-sharing dyads do not vanish
and the outer product understates \(\Omega^\perp_n\). The consistent replacement is the node-clustered
(dyadic-robust) estimator
\[
  \widehat\Omega^{\,\mathrm{dy}}_n
  =
  N_n^{-1}
  \sum_{e}\sum_{e'}
  \mathbf 1\{e\cap e'\neq\varnothing\}\,
  \widehat s_e\,\widehat s_{e'}^\top ,
\]
summing the outer products of all dyad pairs that share at least one node. Under the dyadic-dependence model
in which dyads with no common node are conditionally independent,
\(\widehat\Omega^{\,\mathrm{dy}}_n-\Omega^\perp_n=o_p(1)\) by the cluster-robust argument with the nodes as
overlapping clusters, so \textnormal{(C9)} holds with \(\widehat\Omega^\perp_n=\widehat\Omega^{\,\mathrm{dy}}_n\)
and the Wald interval retains its coverage. This is the variance reported for the S\&P~500 network in the main
text; it reduces to the per-dyad outer product when node-sharing dyads are uncorrelated.
\end{remark}

\begin{remark}[Verification of the orthogonal construction]
\label{rmk:verify-s6}
The single load-bearing step in Theorem~1c is that the explicit correction of
\(\textnormal{(C4)}\),
\(c_{k,e}=\dot\Psi^s_{\gamma,n,k}\,\dot\Gamma_{\gamma,n,k}^{-1}\,g_{k,e}\), satisfies the
target-preservation condition \(\textnormal{(C2)}\) and the foldwise Neyman-orthogonality
condition \(\textnormal{(C3)}\); the remainder of the proof is the standard asymptotic-linearity
argument. Both follow in one line from the zero-conditional-mean nuisance score
\(\Gamma_{n,k}(\gamma^\circ)=0\). For \(\textnormal{(C2)}\), the correction inherits that zero
mean,
\[
  E\{c_{k,e}(\bbeta^\circ,\gamma^\circ)\mid\mathcal G_n\}
  =\dot\Psi^s_{\gamma,n,k}\,\dot\Gamma_{\gamma,n,k}^{-1}\,\Gamma_{n,k}(\gamma^\circ)=0,
\]
so the orthogonal score keeps the BPS target \(\bbeta^\circ\). For \(\textnormal{(C3)}\),
differentiating \(\Psi^\perp_{n,k}=\Psi^s_{n,k}-E\{c_{k,e}\mid\mathcal G_n\}\) in \(\gamma\) at
\(\gamma^\circ\) and using
\(\partial_\gamma[\dot\Gamma_{\gamma,n,k}^{-1}\Gamma_{n,k}]_{\gamma^\circ}=
\dot\Gamma_{\gamma,n,k}^{-1}\dot\Gamma_{\gamma,n,k}=I\) gives
\(\partial_\gamma\Psi^\perp_{n,k}=\dot\Psi^s_{\gamma,n,k}-\dot\Psi^s_{\gamma,n,k}=0\), the
partialling-out identity of two-step semiparametric estimation. The construction can therefore be
checked by verifying only that \(g_{k,e}\) is a valid nuisance score, of zero conditional mean at
\(\gamma^\circ\) with nonsingular \(\dot\Gamma_{\gamma,n,k}\); the orthogonality is then
automatic.
\end{remark}

\begin{remark}[Which agents satisfy the nuisance-rate condition (C5)]
\label{rmk:c5-agents}
Condition (C5) asks for \(\sqrt{N_n}\,r_n^2\to0\) in the score-scale norm
\(\|\cdot\|_{\mathcal H_n}\). Since a logit-feature deviation equals the relative probability
deviation on the sparse scale \(p_e\asymp\rho_n\), \(r_n^2\) is the probability-scale dyadic
mean-squared error of the agent divided by \(\rho_n^2\). An agent with a fixed number of parameters,
each informed by \(\Theta(n^2)\) dyads, has probability-scale error \(O(\rho_n/n^2)\) and hence
\(r_n^2=O(1/(n^2\rho_n))\), so \(\sqrt{N_n}\,r_n^2=O(K^2/(n\sqrt{\rho_n}))\to0\); the \(K\)-block
model with fixed \(K\) is of this type. An agent with \(\Theta(n)\) node-level parameters, each
informed by only \(\Theta(n)\) dyads, has relative error \(\asymp(n\rho_n)^{-1/2}\) per dyad and so
\(r_n^2\asymp 1/(n\rho_n)\) up to dimension and polylog factors, giving
\(\sqrt{N_n}\,r_n^2\asymp 1/\sqrt{\rho_n}\), which does not vanish in the sparse regime; the
adjacency-spectral (GRDPG) and expected-degree (Chung--Lu) agents are of this type. Theorem~1c
therefore delivers the unconditional limit for finite-dimensional agents. For \(\Theta(n)\)-parameter
agents the second-order condition is not available in the sparse regime, and the relevant guarantee is
the conditional limit of Section~\ref{supp:t1}, which conditions on the fitted features and requires
no agent-rate condition.
\end{remark}

\section{Local minimax lower bound for mechanism tracking (Theorem 3b)}
\label{supp:s7}
This section gives the full statement and proof of the local minimax lower bound matching the tracking upper bound of Section~\ref{supp:t3}, together with its separated-switch
and switch-regret corollaries. The argument is a sparse Bernoulli-logit reduction: the
per-dyad divergence bound of Lemma~1 (Section~\ref{supp:t2}) controls the information in one snapshot, and
the lower bounds follow by Fano, Bretagnolle--Huber, and a stopped-likelihood argument across
the conditionally independent post-switch snapshots.

{\renewcommand{\thetheorem}{3b}%
\begin{theorem}[Local minimax lower bound for BPS mechanism tracking]
\label{thm:bps-tracking-lower}
Let \(J\ge 2\), and define
\[
        L_J=\max\{1,\log J\},
        \qquad
        L_{J,\alpha}=L_J\vee \log(1/\alpha),
        \qquad
        0<\alpha\le \alpha_0<1/4 .
\]
Consider the sparse-offset Bernoulli BPS model
\[
A_{t,e}\mid \bbeta_t^\circ
\sim
\operatorname{Bernoulli}\!\left[
\sigma\left\{
\log\rho_n+\alpha_t^\circ+
\sum_{j=1}^J \theta_{t,j}^\circ u^{(j)}_{t,e}
\right\}
\right],
\qquad
\bbeta_t^\circ=(\alpha_t^\circ,\theta_t^\circ),
\]
with conditionally independent dyads. For
\(\bbeta=(\alpha,\theta_1,\ldots,\theta_J)\), write
\(
        q_{\bbeta,t,e}
        =
        \sigma\{
        \log\rho_n+\alpha+
        \sum_{j}\theta_j u^{(j)}_{t,e}
        \}
\),
\(
        z_{t,e}
        =
        (1,u^{(1)}_{t,e},\ldots,u^{(J)}_{t,e})^\top
\),
\(
        N_t(\bbeta)
        =
        \sum_{e\in E_t}
        q_{\bbeta,t,e}\{1-q_{\bbeta,t,e}\}
\),
and
\[
        H_t(\bbeta)
        =
        N_t(\bbeta)^{-1}
        \sum_{e\in E_t}
        q_{\bbeta,t,e}\{1-q_{\bbeta,t,e}\}
        z_{t,e}z_{t,e}^\top .
\]
Assume that the dynamic BPS class
\(\mathcal P_{n,T}(1,\kappa,N)\) contains the following one-switch
least-favorable family. For every switch time \(s\) with sufficient
post-switch horizon, there are \(J\) laws \(P_1,\ldots,P_J\) satisfying:
\begin{enumerate}
\item \(P_1,\ldots,P_J\) are identical up to time \(s\).
\item After time \(s\), law \(P_r\) has constant BPS state
\(
        \bbeta^{(r)}=(\alpha_0,\kappa e_r),
        \ r=1,\ldots,J,
\)
where \(e_r\) is the \(r\)-th coordinate vector in \(\mathbb R^J\).
\item In this least-favorable family, the agent arrays are deterministic or
oracle-fixed and satisfy \(\max_{t,e,j}|u^{(j)}_{t,e}|\le B\).
The post-switch snapshots are conditionally independent over time given the
states and these agent arrays.
\item There is a compact convex BPS prediction class \(\Theta_0\) containing
\(\{\bbeta^{(1)},\ldots,\bbeta^{(J)}\}\), and, for every \(\bbeta\in\Theta_0\),
\(
        c_q\rho_n
        \le
        q_{\bbeta,t,e}
        \le
        C_q\rho_n
        \le 1/2
\)
uniformly in \(t,e\).
\item For every \(\bbeta\in\Theta_0\), \(N_t(\bbeta)\asymp N\asymp n^2\rho_n\).
\item The BPS design is uniformly non-aliased on \(\Theta_0\):
\(
        0<\lambda I_{J+1}
        \preceq
        H_t(\bbeta)
        \preceq
        \Lambda I_{J+1}
        <\infty
\)
for all \(\bbeta\in\Theta_0\).
\end{enumerate}
The active post-switch mechanism under \(P_r\) is
\(
        r_{s+}
        =
        \arg\max_{j}\theta_j^{(r)}
        =
        r
\),
with margin
\(
        \theta^{(r)}_r-\max_{j\neq r}\theta^{(r)}_j
        =
        \kappa
\).
Assume \(0<\kappa\le\kappa_0\), where \(\kappa_0\) is small enough that the
above sparse-scale bounds hold uniformly.
Then there exist constants \(c,c_0>0\), depending only on
\((B,c_q,C_q,\lambda,\Lambda,\alpha_0,\kappa_0)\), such that the following
bounds hold even if the procedure is told \(s\), the agent arrays, \(N\), and
\(\kappa\).

\noindent
\textbf{(i) Fixed-horizon localization.}
For any estimator \(\widehat r_{s+h}\) measurable with respect to \(A_{1:s+h}\),
\[
        \inf_{\widehat r_{s+h}}
        \sup_{P\in\mathcal P_{n,T}(1,\kappa,N)}
        P\{\widehat r_{s+h}\neq r_{s+}\}
        \ge c
\]
whenever \(hN\kappa^2\le c_0L_J\) and \(s+h\le T\).
Thus constant-probability localization before order
\(
        L_J/(N\kappa^2)
        \asymp
        \log J/(n^2\rho_n\kappa^2)
\)
post-switch snapshots is impossible.

\noindent
\textbf{(ii) \(\alpha\)-reliable declaration delay.}
Let \(\tau\ge s\) be a finite stopping time, and suppose a procedure declares
\(\widehat r_\tau\) with local error at most \(\alpha\):
\(
        \sup_{P}
        P\{\widehat r_\tau\neq r_{s+}\}
        \le \alpha
\).
Then
\[
        \inf_{\substack{\tau,\widehat r_\tau:\\
        \sup_{P}P\{\widehat r_\tau\neq r_{s+}\}\le\alpha}}
        \sup_{P\in\mathcal P_{n,T}(1,\kappa,N)}
        E_P(\tau-s)
        \ge
        c\frac{L_{J,\alpha}}{N\kappa^2}
        \asymp
        \frac{\log(eJ/\alpha)}{n^2\rho_n\kappa^2}.
\]
On a finite horizon, the same statement holds whenever the post-switch regime
lasts at least a constant multiple of \(L_{J,\alpha}/(N\kappa^2)\); otherwise
the lower bound is truncated by the available post-switch run length.
\end{theorem}%
\addtocounter{theorem}{-1}}

\begin{proof}
It is enough to prove the result on the least-favorable family
\(\{P_1,\ldots,P_J\}\subset\mathcal P_{n,T}(1,\kappa,N)\). All lower bounds proved
for the oracle problem in which \(s\), the agent arrays, \(N\), and \(\kappa\) are
known therefore also hold for the original online tracking problem.
For \(r\neq r'\), define \(q^{(r)}_{t,e}=q_{\bbeta^{(r)},t,e}\) and the dyadic BPS
logit gap
\(
        \Delta^{rr'}_{t,e}
        =
        \operatorname{logit}q^{(r)}_{t,e}
        -
        \operatorname{logit}q^{(r')}_{t,e}
\).
Because \(\bbeta^{(r)}=(\alpha_0,\kappa e_r)\),
\(
        \Delta^{rr'}_{t,e}
        =
        \kappa\{u^{(r)}_{t,e}-u^{(r')}_{t,e}\}
\),
and the bounded-agent assumption gives \(|\Delta^{rr'}_{t,e}|\le 2B\kappa\).
By the sparse-scale bounds and \(N_t(\bbeta^{(r)})\asymp N\),
\(
        \sum_{e\in E_t}
        q^{(r)}_{t,e}\{1-q^{(r)}_{t,e}\}
        \{\Delta^{rr'}_{t,e}\}^2
        \le
        C N\kappa^2
\).
By the sparse Bernoulli-logit perturbation lemma (Lemma~1), for all \(r\neq r'\),
\[
\operatorname{KL}\!\left\{
\operatorname{Bern}(q^{(r)}_{t,e}),
\operatorname{Bern}(q^{(r')}_{t,e})
\right\}
\le
C
q^{(r)}_{t,e}\{1-q^{(r)}_{t,e}\}
\{\Delta^{rr'}_{t,e}\}^2 ,
\]
so summing over dyads gives the one-snapshot bound
\(
        \operatorname{KL}(P_{r,t},P_{r',t})
        \le
        C N\kappa^2
\).
Since the post-switch snapshots are conditionally independent,
\[
        \operatorname{KL}(P_r^{h},P_{r'}^{h})
        =
        \sum_{\ell=1}^{h}
        \operatorname{KL}(P_{r,s+\ell},P_{r',s+\ell})
        \le
        C hN\kappa^2 ,
\]
where \(P_r^h\) denotes the law of \((A_{s+1},\ldots,A_{s+h})\) under \(P_r\).
Let \(R\sim\operatorname{Unif}\{1,\ldots,J\}\) and \(Y_h=(A_{s+1},\ldots,A_{s+h})\)
be generated from \(P_R^h\). The mutual information satisfies
\(
        I(R;Y_h)
        \le
        J^{-2}
        \sum_{r,r'}
        \operatorname{KL}(P_r^h,P_{r'}^h)
        \le
        C hN\kappa^2
\).
For \(J\ge3\), Fano's inequality gives
\[
        \inf_{\widehat r}
        \frac1J\sum_{r=1}^{J}
        P_r\{\widehat r(Y_h)\neq r\}
        \ge
        1-
        \frac{C hN\kappa^2+\log2}{\log J},
\]
so for a sufficiently small \(c_0>0\), \(hN\kappa^2\le c_0\log J\) implies the
average error is at least \(c\). For \(J=2\), the Bretagnolle--Huber inequality
gives, for any \(r\neq r'\),
\(
        P_r\{\widehat r\neq r\}
        +
        P_{r'}\{\widehat r\neq r'\}
        \ge
        \frac12
        \exp\{-\operatorname{KL}(P_r^h,P_{r'}^h)\}
\),
hence
\(
        \max[
        P_r\{\widehat r\neq r\},
        P_{r'}\{\widehat r\neq r'\}
        ]
        \ge
        \frac14\exp\{-C hN\kappa^2\}
\),
and \(hN\kappa^2\le c_0\) proves the same constant-error bound. Since
\(L_J=\max\{1,\log J\}\), we have the Bayes-average inequality
\(
        \inf_{\widehat r_{s+h}}
        J^{-1}\sum_{r}
        P_r\{\widehat r_{s+h}\neq r\}
        \ge c
\)
whenever \(hN\kappa^2\le c_0L_J\). The fixed-horizon minimax statement follows
because the supremum over \(\mathcal P_{n,T}(1,\kappa,N)\) is at least the maximum
over \(\{P_1,\ldots,P_J\}\), which is at least the average.
It remains to prove the reliable declaration-delay bound. Let \(P_r^\tau\) denote
the law, under \(P_r\), of the stopped experiment \((\tau,A_{s+1},\ldots,A_\tau)\).
For any finite stopping time, the likelihood-ratio chain rule gives
\(
        \operatorname{KL}(P_r^\tau,P_{r'}^\tau)
        =
        E_r
        \sum_{\ell\ge1}
        \mathbf 1\{\tau-s\ge \ell\}
        \operatorname{KL}(P_{r,s+\ell},P_{r',s+\ell})
        \le
        C N\kappa^2 E_r(\tau-s)
\).
If \(E_r(\tau-s)=\infty\) the bound is trivial, so assume it finite. Suppose
\(\sup_r P_r\{\widehat r_\tau\neq r\}\le \alpha\). Fix \(r\neq r'\); under \(P_r\),
\(P_r\{\widehat r_\tau=r\}\ge 1-\alpha\), and under \(P_{r'}\),
\(P_{r'}\{\widehat r_\tau=r\}\le \alpha\). By data processing for the event
\(\{\widehat r_\tau=r\}\),
\(
        \operatorname{KL}(P_r^\tau,P_{r'}^\tau)
        \ge
        \operatorname{kl}(1-\alpha,\alpha)
        \ge
        c\log(1/\alpha)
\),
uniformly for \(0<\alpha\le\alpha_0<1/4\). Combining with the stopped-KL upper
bound, \(E_r(\tau-s)\ge c\log(1/\alpha)/(N\kappa^2)\), hence
\(\sup_r E_r(\tau-s)\ge c\log(1/\alpha)/(N\kappa^2)\). For the \(J\)-way piece, let
\(R\sim\operatorname{Unif}\{1,\ldots,J\}\); since \(\widehat r_\tau\) has Bayes error
at most \(\alpha\), Fano gives \(I(R;\tau,A_{s+1:\tau})\ge cL_J\) for \(J\ge3\),
while
\(
        I(R;\tau,A_{s+1:\tau})
        \le
        J^{-2}
        \sum_{r,r'}
        \operatorname{KL}(P_r^\tau,P_{r'}^\tau)
        \le
        C N\kappa^2
        \sup_r E_r(\tau-s)
\),
so \(\sup_r E_r(\tau-s)\ge cL_J/(N\kappa^2)\). For \(J=2\), \(L_J\) is absorbed by
\(\log(1/\alpha)\) since \(\alpha<1/4\). Combining the multiplicity and confidence
pieces, \(\sup_r E_r(\tau-s)\ge c(L_J\vee\log(1/\alpha))/(N\kappa^2)
=cL_{J,\alpha}/(N\kappa^2)\), and the supremum over the full class is at least the
supremum over the least-favorable family.
\end{proof}

{\renewcommand{\thecorollary}{3b.1}\renewcommand{\thetheorem}{3b.1}%
\begin{corollary}[Separated-switch reliable-delay lower bound]
\label{cor:bps-S-switch-delay-lower}
Assume the conditions of Theorem~\ref{thm:bps-tracking-lower}. Suppose the
dynamic class contains \(S\) separated hard switch episodes: after each switch
\(s_\ell\), conditional on the past up to \(s_\ell\), the post-switch law contains
the same \(J\)-point local BPS packing as in
Theorem~\ref{thm:bps-tracking-lower}, the corresponding post-switch testing
windows are disjoint, and each post-switch regime lasts at least
\(
        1+\lceil
        C L_{J,\alpha}/(N\kappa^2)
        \rceil
\)
snapshots. If the procedure is required to make an \(\alpha\)-reliable local
declaration after each switch,
\(
        \sup_P
        P\{\widehat r_{\tau_\ell}\neq r_{s_\ell+}\mid \mathcal F_{s_\ell}\}
        \le \alpha
\)
almost surely for each \(\ell\), then
\[
        \inf_{\widehat r}
        \sup_{P\in\mathcal P_{n,T}(S,\kappa,N)}
        E_P\sum_{\ell=1}^S(\tau_\ell-s_\ell)
        \ge
        c\frac{S L_{J,\alpha}}{N\kappa^2}.
\]
\end{corollary}%
\addtocounter{theorem}{-1}}

\begin{proof}
Condition on the history \(\mathcal F_{s_\ell}\) immediately before switch
\(s_\ell\). By assumption, the conditional post-switch experiment contains the
same least-favorable \(J\)-point BPS packing as in
Theorem~\ref{thm:bps-tracking-lower}. Applying the reliable declaration-delay
part of that theorem conditionally on \(\mathcal F_{s_\ell}\) gives
\(
        E_P(\tau_\ell-s_\ell\mid \mathcal F_{s_\ell})
        \ge
        cL_{J,\alpha}/(N\kappa^2)
\).
Taking expectations and summing over the \(S\) disjoint episodes gives the claim,
and the supremum over \(P\) with the infimum over admissible procedures completes
the proof.
\end{proof}

{\renewcommand{\thecorollary}{3b.2}\renewcommand{\thetheorem}{3b.2}%
\begin{corollary}[Switch-regret lower bound for BPS coefficient predictors]
\label{cor:bps-switch-regret-lower}
Assume the conditions of Theorem~\ref{thm:bps-tracking-lower} and the separated
hard-episode condition of Corollary~\ref{cor:bps-S-switch-delay-lower}. Restrict
attention to predictable BPS coefficient predictors
\(
        \widehat\bbeta_{t|t-1}\in\Theta_0
\)
measurable with respect to \(A_{1:t-1}\). For \(N_t=N_t(\bbeta_t^\circ)\), define
the information-normalized BPS risk
\[
        \bar R_t(\bbeta)
        =
        \frac{1}{N_t}
        \sum_{e\in E_t}
        \operatorname{KL}\!\left\{
        \operatorname{Bern}(q_{\bbeta_t^\circ,t,e}),
        \operatorname{Bern}(q_{\bbeta,t,e})
        \right\}.
\]
Then
\[
        \inf_{\widehat\bbeta}
        \sup_{P\in\mathcal P_{n,T}(S,\kappa,N)}
        E_P
        \sum_{t=1}^T
        N_t\{
        \bar R_t(\widehat\bbeta_{t|t-1})
        -
        \bar R_t(\bbeta_t^\circ)
        \}
        \ge
        cS L_J .
\]
Since \(\bar R_t(\bbeta_t^\circ)=0\) in the least-favorable family, the displayed
quantity is the cumulative excess one-step BPS log-score, an ordinary expected
predictive-regret lower bound. The factor \(\log(1/\alpha)\) belongs to reliable
hard localization and is not part of unconstrained expected predictive regret;
the stationary parametric term of order \((J+1)\log(TN)\) requires a separate
stationary BPS regret lower-bound lemma.
\end{corollary}%
\addtocounter{theorem}{-1}}

\begin{proof}
Fix one hard switch at time \(s\), and write \(a=N\kappa^2\), \(L=L_J\). For the
post-switch packing \(\{\bbeta^{(1)},\ldots,\bbeta^{(J)}\}\), define the
nearest-neighbor decoder
\(
        d(\bbeta)
        =
        \arg\min_{r'}
        \|\bbeta-\bbeta^{(r')}\|_2
\).
Because \(\|\bbeta^{(r)}-\bbeta^{(r')}\|_2=\sqrt{2}\kappa\) for \(r\neq r'\), the
event \(d(\bbeta)\neq r\) implies \(\|\bbeta-\bbeta^{(r)}\|_2\ge\kappa/\sqrt{2}\). By
the uniform sparse-scale bounds and the uniform non-aliasing assumption, the
information-normalized logistic risk is uniformly strongly convex on \(\Theta_0\),
so \(\bar R_t(\bbeta)-\bar R_t(\bbeta^{(r)})\ge c\|\bbeta-\bbeta^{(r)}\|_2^2\), and
therefore \(d(\bbeta)\neq r\) implies \(\bar R_t(\bbeta)-\bar R_t(\bbeta^{(r)})\ge
c\kappa^2\). At prediction time \(s+m+1\), the prequential predictor
\(\widehat\bbeta_{s+m+1|s+m}\) has seen only \(m\) post-switch snapshots, so the
induced decoder \(d(\widehat\bbeta_{s+m+1|s+m})\) is an estimator of \(r\) from
\(m\) post-switch snapshots. By the Bayes-average fixed-horizon bound of
Theorem~\ref{thm:bps-tracking-lower}, whenever \(ma\le c_0L\),
\(
        J^{-1}\sum_r
        P_r\{d(\widehat\bbeta_{s+m+1|s+m})\neq r\}
        \ge c
\),
hence
\(
        J^{-1}\sum_r
        E_r[
        \bar R_{s+m+1}(\widehat\bbeta_{s+m+1|s+m})
        -
        \bar R_{s+m+1}(\bbeta^{(r)})
        ]
        \ge
        c\kappa^2
\)
for every \(m\) with \(ma\le c_0L\). Let \(M=1+\lfloor c_0L/a\rfloor\); then
\(ma\le c_0L\) for \(m=0,\ldots,M-1\), and \(Ma\ge cL\), so
\[
        J^{-1}\sum_r
        E_r
        \sum_{m=0}^{M-1}
        N[
        \bar R_{s+m+1}(\widehat\bbeta_{s+m+1|s+m})
        -
        \bar R_{s+m+1}(\bbeta^{(r)})
        ]
        \ge
        c M N\kappa^2
        \ge
        cL_J .
\]
The supremum over \(r\) is at least the Bayes average, so one hard switch costs at
least \(cL_J\) in expected scaled BPS regret. By the separated-hard-episode
assumption the same argument applies on each of the \(S\) disjoint episodes with
non-overlapping prediction windows; summing the \(S\) one-switch bounds gives the
claim.
\end{proof}

\begin{remark}[Verification of the lower-bound constant and the switch-regret scope]
\label{rmk:verify-s7}
Two steps carry Theorem~3b and are worth isolating. First, the constant in the lower bound is the
one supplied by Lemma~1: the per-dyad Kullback--Leibler and Hellinger divergences between the
competing post-switch laws are at most \(C\,p(1-p)x^2\le C'\kappa^2\), with \(C'\) depending only
on the design bounds \([\lambda,\Lambda]\), so one snapshot carries at most \(C'N\kappa^2\) units
of information; the Fano and Bretagnolle--Huber reductions then forbid constant-probability
localization until \(hN\kappa^2\gtrsim\log J\), and the constant in
\(\log J/(C'N\kappa^2)\) is exactly the Lemma~1 constant. This is what makes the bound match the
discounted-filter upper bound up to constants, not only in order. Second, the switch-regret bound \(S\,L_J\) is the cost of identifying the
post-switch mechanism among \(J\) alternatives at each of \(S\) switches, and it deliberately
\emph{excludes} the stationary within-regime estimation term \((J{+}1)\log(TN)\): the
least-favorable family fixes the within-regime coefficients at \((\alpha_0,\kappa e_r)\) and
varies only the active index \(r\), so no within-regime estimation cost is charged, and the bound
is per switch, the matching companion to the \(O(SD^2)\) switch term of the upper bound, with
neither term scaling in \(T\).
\end{remark}

\section{Mechanism localization on a Bitcoin trust network}
\label{sec:e2}
The financial network of Section~\ref{sec:localize} is one real dynamic network on which the
filter localizes a documented regime change. We add a second of an entirely different kind, a
who-trusts-whom network, not a correlation network. We use the Bitcoin-OTC trust network
\citep{kumar2016edge,leskovec2016snap}, a standard temporal benchmark in which members of a trading
platform rate one another from $-10$ to $+10$, observed natively as timestamped ratings among
$5{,}881$ users from November 2010 to January 2016 (we verify the file against its canonical
checksum). We restrict to the active core of $483$ users with at least thirty ratings,
comparable in size to the financial panel, and form one undirected binary snapshot per quarter, placing an edge between two users who
exchanged at least one rating in the quarter regardless of its sign or direction, giving $18$ snapshots on which we run the same one-step-ahead protocol and competitors. The
window contains a documented external shock: the February 2014 collapse of Mt.\ Gox, which had
handled the majority of global Bitcoin volume.

The calibration result is sharper here than on the financial panel. The snapshots are sparse,
so the raw mechanism agents emit edge probabilities near zero and are badly miscalibrated
against a balanced evaluation, and every competitor inherits this. Averaged over the seventeen
one-step forecasts the synthesis attains NLL $0.678$, PIT-KS $0.075$, and ECE $0.097$, against
NLL between $2.8$ and $4.9$ and ECE near $0.48$ for model averaging, stacking, sequential
stacking, equal weights, and the dynamic block and latent-space competitors alike; the
intercept restores the overall level that no convex combiner can reach. The paired bootstrap
places every NLL gap far from zero, for instance $+2.99$ ($[2.55,3.45]$) against the dynamic
latent-space model, with $p<10^{-4}$. The ranking pattern is again the same: the synthesis gives up some separation to the convex combiners, which rank more sharply while remaining badly miscalibrated, because pulling the badly scaled agents back to a calibrated level costs some
separation. The decision experiment resolves which matters. With the threshold set from the
cost ratio, the miscalibrated competitors incur substantially higher decision cost, and the synthesis incurs $49\%$
lower cost at $r=2$, $78\%$ at $r=5$, and $87\%$ at $r=10$, a far wider margin than on the
financial panel precisely because the alternatives are so poorly calibrated.

The localization reads off the documented shock. Figure~\ref{fig:btcloc} plots the synthesis
state across quarters. The degree mechanism (Chung-Lu) holds the largest weight throughout, the
expected signature of a trust network in which reputation concentrates on a few much-rated
hubs, while the community weight decays from $0.78$ to near zero as the early club structure
dissolves. The regime change appears in the intercept, which corrects the agents' overall level
and therefore inversely tracks density: it falls as the network densifies through the 2013
Bitcoin run-up, reaching its minimum at the third-quarter-2013 activity peak, and rises steadily
through the contraction that follows the Mt.\ Gox collapse. Across the financial and trust
networks the filtered weights localize a documented structural break by the same mechanism that
Theorem~\ref{thm:t3} controls.

\begin{figure}[t]
\centering
\includegraphics[width=0.92\textwidth]{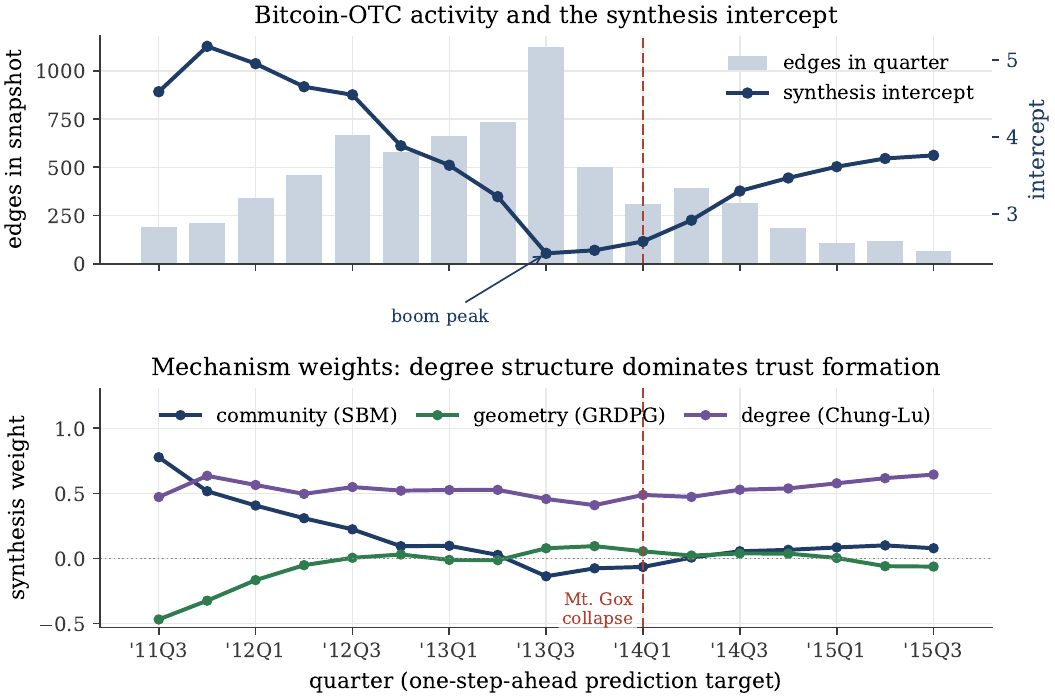}
\caption{Mechanism localization on the real Bitcoin-OTC trust network
\citep{kumar2016edge}, one-step-ahead over $18$ quarterly snapshots
(Theorem~\ref{thm:t3}). Top: per-quarter activity (bars) and the synthesis intercept (line),
which corrects the agents' level and inversely tracks density, bottoming at the 2013 boom peak
and rising through the post-Mt.-Gox contraction (dashed line, February 2014). Bottom: synthesis
weights, with the degree (Chung-Lu) weight stably largest, the expected structure of a
trust network, and the community mechanism decaying. The filtered state tracks a documented
external shock without it being supplied to the model.}
\label{fig:btcloc}
\end{figure}

\section{Static link prediction: the single-snapshot specialization}
\label{sec:e3}

The $T=1$ specialization is static link prediction. We hold out $15\%$ of edges of each of
four real networks for testing, fit agents on a core, learn the synthesis on a validation
fold, and evaluate on the held-out edges against an equal number of sampled non-edges. We
report two agent pools: the three generative agents, and a four-agent pool that adds the
local Adamic--Adar predictor, demonstrating that the synthesis absorbs a heterogeneous
object type. Table~\ref{tab:e3} reports the four-agent results; the effect of adding the
local agent is discussed below.

\begin{table}[t]
\centering
\caption{Link prediction on four real networks (four-agent pool: SBM, GRDPG, Chung--Lu,
Adamic--Adar). Held-out edges versus sampled non-edges. Best PIT-KS in each
block in bold.}
\label{tab:e3}
\small
\begin{tabular}{llccc}
\toprule
Network & Method & NLL & PIT-KS & ECE \\
\midrule
\multirow{4}{*}{polblogs ($n{=}1222$)}
 & \textbf{Static BPS} & \textbf{0.309} & \textbf{0.016} & \textbf{0.014} \\
 & BMA       & 0.999 & 0.157 & 0.382 \\
 & Stacking  & 0.702 & 0.184 & 0.391 \\
 & Equal weight & 0.957 & 0.311 & 0.430 \\
\midrule
\multirow{4}{*}{ca-GrQc ($n{=}4158$)}
 & \textbf{Static BPS} & \textbf{0.314} & \textbf{0.010} & \textbf{0.028} \\
 & BMA       & 1.967 & 0.226 & 0.479 \\
 & Stacking  & 1.535 & 0.257 & 0.464 \\
 & Equal weight & 1.785 & 0.338 & 0.460 \\
\midrule
\multirow{4}{*}{ca-CondMat ($n{=}21363$)}
 & \textbf{Static BPS} & \textbf{0.220} & \textbf{0.005} & \textbf{0.030} \\
 & BMA       & 1.269 & 0.187 & 0.499 \\
 & Stacking  & 1.100 & 0.207 & 0.496 \\
 & Equal weight & 1.522 & 0.384 & 0.495 \\
\midrule
\multirow{4}{*}{power-grid ($n{=}4941$)}
 & \textbf{Static BPS} & \textbf{0.629} & \textbf{0.014} & \textbf{0.010} \\
 & BMA       & 3.963 & 0.477 & 0.498 \\
 & Stacking  & 3.807 & 0.461 & 0.499 \\
 & Equal weight & 4.064 & 0.440 & 0.498 \\
\bottomrule
\end{tabular}
\end{table}

Each row instantiates a clause of Corollary~\ref{cor:t1}. The synthesis has the best predictive log-score and the best calibration on all four networks, with PIT-KS between $0.005$ and $0.016$ against $0.16$--$0.49$ for every competitor; the intercept absorbs the constant case-control offset $\log(s_1/s_0)$, so the reported probabilities are calibrated to the evaluation sample. On ranking it is best on polblogs and competitive on the collaboration networks, where stacking ranks marginally more sharply while being badly miscalibrated, the intercept-invariance of Remark~\ref{rmk:final}. Agent-set monotonicity holds: adding the local Adamic--Adar predictor strictly improves the four-agent forecast on ca-CondMat and ca-GrQc, the strict gain the corollary attributes to an agent carrying residual log-score information. The power-grid mesh is the graceful-failure case: a near-regular planar grid carries essentially no community, geometry, or degree signal, so the synthesis carries no usable ranking signal yet remains the best-calibrated method (PIT-KS $0.014$, ECE $0.010$), reporting calibrated uncertainty, not a confident wrong ordering, the $\theta^\circ\to0$ regime of the corollary.

\section{Held-out validation of the mechanism weights}
\label{sec:supp-weightval}

The leading synthesis weights reported in the main text are validated against signals the agents never
see. This section records the deterministic statistics behind those readings, each recomputed from the
raw contact logs or the S\&P returns.

\paragraph{High school: community, validated by class.}
The aggregate high-school contact graph is assortative by the nine class labels, with attribute
assortativity $0.65$. Fitting the community agent to the aggregate graph and comparing its blocks to
the held-out labels gives adjusted Rand index $0.993$ and normalized mutual information $0.994$; the
comparison uses a seeded spectral fit ($n_{\mathrm{init}}=25$) and is stable, with standard deviation
$0.000$ across seeds. Within the school day the community weight tracks the within-class share of
contacts, correlation $0.49$ ($p=0.0015$).

\paragraph{Hospital: degree, validated by role.}
The ward is not assortative by its four role labels: the attribute assortativity is $-0.12$, slightly
disassortative, so edges run across roles and there is no role-community structure for the community
agent to recover. Block recovery for the roles is seed-unstable on this graph, with adjusted Rand index
ranging from about $0.2$ to $0.5$ depending on the eigensolver start and the clustering initialization,
which is why the deterministic assortativity, not a block-recovery index, is the reported
community-structure measure. What organizes the ward is contact volume: the mean aggregate degrees are
$39.1$ for nurses, $35.5$ for doctors, $29.1$ for administrative staff, and $20.7$ for patients, so the
staff carry roughly $1.7$ times the patient degree and the degree weight leads.

\paragraph{S\&P 500: geometry tracks the regime.}
Realized volatility of the equal-weight index, computed from the same returns but independent of the
network construction, agrees with the network's top-eigenvalue share on the regime (correlation
$0.93$). Against this external series the latent-geometry weight is anticorrelated ($-0.78$ against
volatility, $-0.91$ against the eigenvalue share), the degree weight is positively correlated
($+0.66$), and the community weight is mildly positive ($+0.33$). Counting the single largest weight per
quarter, the community agent leads six of the nine stress quarters and the degree agent three, while in
the calm quarters the lead splits between the geometry and community agents; the reallocation toward
block and degree structure in stress is what the geometry anticorrelation records. No competitor returns
any of these weights: model averaging collapses onto a single agent (Theorem~4 of the main text), and a
recalibrated forecaster returns calibrated probabilities but no mechanism weights.

\section{Single-snapshot versus smoothed-adjacency agents on the large networks}
\label{sec:supp-variant}

The large-network study of the main text fits the three mechanism agents on the discounted-average adjacency through
$t-1$, so that they are the dynamic block, latent-geometry, and degree predictors. A natural
alternative fits each agent on the single previous snapshot $A_{t-1}$ alone, as the released
\texttt{run\_dynamic} code does. Table~\ref{tab:supp-variant} reports both on the same full node sets
and all monthly windows. The choice leaves the calibration conclusions untouched on every network:
the synthesis PIT-KS and ECE move by at most $0.004$. Its effect is confined to the forecast sharpness, and only
where the mechanisms match the structure. On the citation network the smoothed-adjacency agents cut
the synthesis one-step log-score from $0.631$ to $0.413$;
on the e-mail network the forecast is essentially unchanged
(near the base rate under either choice).
Smoothed-history agents are therefore what sharpen matched mechanisms at scale, while the
calibration and mechanism-weight conclusions are a property of the synthesis itself, not of the
agent-fitting window.

\begin{table}[t]
\centering
\caption{Synthesis on the two large networks with agents fit on the single previous snapshot
versus on the discounted-average adjacency, on the full node sets and all monthly windows. Lower
NLL/PIT-KS/ECE better. Smoothed-adjacency agents sharpen the forecast on the citation
network while leaving the calibration unchanged.}
\label{tab:supp-variant}
\small
\begin{tabular}{llccc}
\toprule
Network & Agents & NLL & PIT-KS & ECE \\
\midrule
\multirow{2}{*}{arXiv HEP-PH} & single-snapshot & 0.631 & 0.027 & 0.027 \\
 & smoothed-adjacency & 0.413 & 0.031 & 0.041 \\
\midrule
\multirow{2}{*}{Enron e-mail} & single-snapshot & 0.576 & 0.057 & 0.073 \\
 & smoothed-adjacency & 0.597 & 0.061 & 0.077 \\
\bottomrule
\end{tabular}
\end{table}

\section{Robustness to a modern machine-learning baseline}
\label{sec:supp-gnn}

The comparisons in the main text are to other statistical combination rules, which carry inferential
targets and calibration guarantees. As a robustness check we also place the synthesis beside a modern
machine-learning forecaster that carries neither: a two-layer graph auto-encoder (Dynamic GCN) in the
manner of \citet{kipf2016variational}, trained online by gradient descent on the discounted-average
adjacency and scored on the identical held-out dyads. Table~\ref{tab:supp-gnn} reports both across all
five networks. The pattern is the one the repositioning predicts. The neural forecaster is competitive
or better on ranking and, on several networks, on the log-score, but it is the
less-calibrated of the two on four of the five networks, with calibration error $0.12$ to $0.19$ against the synthesis's $0.02$ to
$0.14$, while on the S\&P network its PIT-KS is slightly better and its calibration error matches; it states no coverage; and it returns none of the mechanism weights or switch localization
that are the contribution of this paper. On the citation network, where the mechanisms match the
structure, the synthesis attains the lower log-score outright. A black-box forecaster is therefore a
useful sharpness reference, not a competitor on the terms the paper claims, and it is reported here, not in the main text, for exactly that reason.

\begin{table}[t]
\centering
\caption{The synthesis against a graph auto-encoder (Dynamic GCN), a non-statistical forecaster with
no inferential guarantees, across all five networks; both are fit on the past-smoothed adjacency and
scored on the identical held-out dyads. Lower NLL/PIT-KS/ECE better. The neural
forecaster ranks well and is often sharp, but is less calibrated on four of the five networks and returns no
mechanism weights, switch localization, or coverage.}
\label{tab:supp-gnn}
\small
\begin{tabular}{llccc}
\toprule
Network & Method & NLL & PIT-KS & ECE \\
\midrule
\multirow{2}{*}{S\&P 500 ($n=470$)} & Dynamic BPS & 0.761 & 0.092 & 0.123 \\
 & Dynamic GCN & 0.619 & 0.074 & 0.123 \\
\midrule
\multirow{2}{*}{High school ($n=327$)} & Dynamic BPS & 0.613 & 0.107 & 0.140 \\
 & Dynamic GCN & 0.470 & 0.152 & 0.177 \\
\midrule
\multirow{2}{*}{Hospital ($n=75$)} & Dynamic BPS & 0.558 & 0.119 & 0.139 \\
 & Dynamic GCN & 0.543 & 0.136 & 0.176 \\
\midrule
\multirow{2}{*}{arXiv HEP-PH ($n=28{,}093$)} & Dynamic BPS & 0.413 & 0.031 & 0.041 \\
 & Dynamic GCN & 0.522 & 0.139 & 0.188 \\
\midrule
\multirow{2}{*}{Enron ($n=86{,}978$)} & Dynamic BPS & 0.597 & 0.061 & 0.077 \\
 & Dynamic GCN & 0.501 & 0.148 & 0.164 \\
\bottomrule
\end{tabular}
\end{table}

\section{Protocols, scoring, interval taxonomy, and notation}
\label{supp:clarify}
This section makes explicit several conventions the main text uses but states compactly: the two estimation protocols, the exact scoring definitions and probability scale, the three interval types and where each is used, the calibration of the predicted-state forecast, the conditions behind the unconditional community interval, the two recovery-delay regimes, and the symbol conventions. Nothing here changes a theorem; the statements are definitions and consequences of results already proved.

\subsection*{Two estimation protocols}
The method runs two operations on the snapshots, answering different questions, and the boxed procedure of the main text contains both.

\emph{Prequential forecasting.} Each agent forecasts the next snapshot from the past, $\widehat P^{(j)}_t=\mathcal A_j(A_{1:t-1})$, and snapshot $A_t$ enters only after the forecast is formed. The discounted filter produces the predicted state $\widehat\bbeta_{t\mid t-1}$, measurable with respect to $A_{1:t-1}$. Because the fitted agents and the state are functions of the past alone, the conditional weight inference of Theorem~\ref{thm:t1} holds with no same-snapshot refitting: there is no own-observation feedback to orthogonalize. This is the regime for the one-step forecast and the displayed weight trajectory.

\emph{Same-snapshot cross-fitting.} For the unconditional weight interval the snapshot is split into $K$ folds; each agent surface is fit on the dyads outside a fold and read off on the held-out fold, which enters the synthesis only. Here the agents are refit on the same snapshot that scores them, so the first-order effect of the agent nuisance must be removed; the orthogonalized cross-fold construction of Theorem~\ref{thm:crossfold-orthogonal-unconditional-bps} does this and delivers the unconditional interval. The cross-fold split is a property of the interval, not of the forecast. The forecast and trajectory come from the filter; the unconditional interval comes from the cross-fold split.

\subsection*{Scoring and resampling}
Write $D_t$ for the scored dyads at snapshot $t$ and $\widehat q_{t,e}$ for the forecast probability. The negative log-likelihood is
\[
\mathrm{NLL}=-\frac{1}{|D_t|}\sum_{e\in D_t}\big[A_{t,e}\log\widehat q_{t,e}+(1-A_{t,e})\log(1-\widehat q_{t,e})\big],
\]
averaged over scored snapshots. The probability integral transform is randomized for the binary outcome: with $U\sim\mathrm{Unif}(0,1)$ drawn independently per dyad, the value is $U(1-\widehat q_{t,e})$ when $A_{t,e}=0$ and $(1-\widehat q_{t,e})+U\,\widehat q_{t,e}$ when $A_{t,e}=1$, uniform on $(0,1)$ under correct conditional calibration; PIT--KS is the Kolmogorov--Smirnov distance between the pooled transform values and the uniform law. The expected calibration error uses $B$ equal-mass bins of $\{\widehat q_{t,e}\}$, $\mathrm{ECE}=\sum_{b=1}^{B}(|D_t\cap b|/|D_t|)\,|\bar q_b-\bar A_b|$, with $\bar q_b$ the mean forecast and $\bar A_b$ the empirical edge frequency in bin $b$; we use $B=15$.

\emph{Probability scale.} On a balanced evaluation set, all edges with an equal-size sample of non-edges, the scored probability is the case-control probability $p_{\mathrm{cc}}(z)=\Prob(A=1\mid z,S=1)$, where $S=1$ marks inclusion in the balanced set. The population edge probability is recovered by the constant offset of Corollary~\ref{cor:t1}, $\logit p_{\mathrm{pop}}(z)=\logit p_{\mathrm{cc}}(z)+\log(s_1/s_0)$ with sampling fractions $s_1,s_0$. Scores on a balanced set are reported on $p_{\mathrm{cc}}$; where population edge forecasting is the target the offset is applied first and the scale is stated with the table. The two scales are not mixed within a table.

\emph{Resampling unit.} Intervals from resampling use the unit that carries the dependence: nodes for networks whose dyads share an endpoint and are dependent, in particular the correlation networks, where a node-clustered bootstrap matches the clustered score covariance behind the sandwich; and forecast rounds for the prequential scores, resampled in blocks to preserve serial dependence.

\subsection*{Three interval types}
Three intervals appear, with different validity,
\[
\underbrace{H_t^{-1}}_{\text{Laplace credible}},\qquad
\underbrace{H_t^{-1}V_tH_t^{-1}}_{\text{conditional Wald}},\qquad
\underbrace{(H_t^{\perp})^{-1}\Omega_t^{\perp}(H_t^{\perp})^{-1}}_{\text{cross-fold unconditional}} .
\]
The Laplace credible covariance $H_t^{-1}$ is the curvature the filter propagates as its internal posterior, exact only under correct specification. The conditional Wald covariance $H_t^{-1}V_tH_t^{-1}$ is the sandwich of Theorem~\ref{thm:t1}, with $V_t=H_t$ only under correct Bernoulli-logit specification and otherwise the clustered score covariance; it is valid conditional on the fitted agents and is what the coverage check uses. The cross-fold covariance $(H_t^{\perp})^{-1}\Omega_t^{\perp}(H_t^{\perp})^{-1}$, from the orthogonalized same-snapshot cross-fit of Theorem~\ref{thm:crossfold-orthogonal-unconditional-bps}, is valid unconditionally for finite-dimensional agents. The displayed weight bands use the conditional Wald sandwich; the unconditional claims use the cross-fold interval; the Laplace covariance is the filter's propagation step and is not reported as a coverage statement.

\subsection*{Calibration of the predicted-state forecast}
Theorem~\ref{thm:t4} establishes the calibration identities at the snapshot projection $\bbeta^\circ_t$. The one-step forecast uses the predicted state $\widehat\bbeta_{t\mid t-1}$, which differs from $\bbeta^\circ_t$ by the tracking error. The reliability functional is Lipschitz in the state on the sparse scale: for any bin the mean forecast $\bar q_b(\bbeta)$ moves with $\bbeta$ through $\sigma(\bbeta^\top z)$, whose gradient has size of order $\rho_n$ there, while the empirical frequency does not move; hence
\[
\mathrm{ECE}\big(\widehat q_{t\mid t-1}\big)\ \le\ \mathrm{ECE}\big(q_t(\bbeta^\circ_t)\big)\ +\ C\rho_n\,\big\|\widehat\bbeta_{t\mid t-1}-\bbeta^\circ_t\big\| .
\]
The first term is the projection calibration of Theorem~\ref{thm:t4}, controlled by the approximation error of the set of agents and the sampling error; the displacement is the tracking error, controlled by the regret bound of Theorem~\ref{thm:t3}. The predicted-state forecast is therefore intercept-calibrated up to approximation error and tracking error, not exactly calibrated at every time, and the displayed coverage reflects both.

\subsection*{Conditions for the unconditional community interval}
The unconditional interval of Theorem~\ref{thm:crossfold-orthogonal-unconditional-bps} for the community weight is stated for a finite-dimensional community nuisance, which requires a fixed number of blocks $K$; block labels known or recovered exactly with probability $1-o(1)$, so label error does not enter at first order; and a community feature carrying no node-level degree parameter. Under these conditions the block-probability nuisance is estimated at the second-order rate and the orthogonalized weight is unconditionally valid. When the implemented community agent adds a degree correction with node-level effects, its leading weight is on the same footing as the degree and latent-space agents, whose per-node nuisance is informed by only $\asymp n\rho_n$ edges; for those the conditional interval is the one claimed. The settled-snapshot leader is the community weight in the regime where these conditions are met.

\subsection*{Two recovery-delay regimes}
The post-switch recovery delay has two regimes. For $\delta$ bounded away from $1$ the binding constraint is information: each snapshot carries $N\asymp n^2\rho_n$, localization completes in $O(1)$ snapshots, and the discount enters only through a constant, since the discounted post-switch information $N(1-\delta^h)/(1-\delta)\asymp hN$ over the $O(1)$ snapshots coincides with the undiscounted $hN$. As $\delta\to1$ the binding constraint is forgetting: the delay $\log(D/\kappa)/|\log\delta|\asymp 1/(1-\delta)$ grows and the discount enters the rate. The empirical delay plateaus near six snapshots at fixed $\delta$, the information-limited regime; the $1/(1-\delta)$ growth is the forgetting-limited worst case.

\subsection*{Symbol conventions}
A few symbols carry a context-dependent meaning, fixed here.
\begin{center}
\begin{tabular}{l p{0.70\linewidth}}
\hline
Symbol & Meaning by context \\
\hline
$\kappa(H_t)$ & condition number of the information matrix, written with its argument; the aliasing diagnostic \\
$\kappa$ & localization margin $\theta^\circ_{t,r_t}-\max_{j\ne r_t}\theta^\circ_{t,j}$ in the tracking and lower-bound results \\
$K$ & number of blocks for the community agent; separately, the number of cross-fitting folds \\
$d$ & latent dimension of the geometry agent; the synthesis dimension is $J+1$ \\
$\alpha$ & calibration intercept $\alpha_t$ in the state; separately, the level $\alpha$ of a confidence or test statement \\
\hline
\end{tabular}
\end{center}
The two readings of $\kappa$ are distinguished by the matrix argument, and the two readings of $K$ never share an expression.

\subsection*{Further conventions}
\emph{Information scale of a balanced fit.} The dyadic rate $(n^2\rho_n)^{-1/2}$ is the information from a full-snapshot weight fit, where the scored dyad count is $m_t\asymp n^2$ and $N_t\asymp m_t\rho_n\asymp n^2\rho_n$. If the weights are fit on a balanced case-control sample of $m_{\mathrm{val}}\asymp n^2\rho_n$ dyads, the unweighted Bernoulli information is $N_t\asymp m_{\mathrm{val}}\rho_n\asymp n^2\rho_n^2$ and the rate is $(n^2\rho_n^2)^{-1/2}$; the population scale is then recovered by the offset of Corollary~\ref{cor:t1} or by inverse-probability weighting.

\emph{Composite score for dependent dyads.} On networks whose dyads are dependent, in particular the correlation networks where edges sharing an asset are functions of the same return series, the product-Bernoulli NLL is a composite, marginal log-score for the edge-probability forecast, not the joint graph density; the same dependence is why the independent-dyad Wald interval is anti-conservative and the node-clustered variance is used.

\emph{Residualized aliasing diagnostic.} The condition number $\kappa(H_t)$ includes the intercept and so also reflects density and intercept scaling. Aliasing among the agents alone is isolated by the condition number of the intercept-residualized weighted Gram $\widetilde H_t=N_t^{-1}\sum_e w_{t,e}(u_{t,e}-\bar u_{w,t})(u_{t,e}-\bar u_{w,t})^\top$; we report $\kappa(H_t)$ for continuity with the theory and note $\mathrm{cond}(\widetilde H_t)$ as the purely agent-level measure.

\emph{Clip regime.} The bounded-feature property $|u^{(j)}_{t,e}|\le\log(1/\epsilon)+O(\rho_n)$ requires the upper sparse cap $\rho_n/\epsilon$ to bind, that is $\rho_n/\epsilon<1-\epsilon$; for $\rho_n\ge\epsilon(1-\epsilon)$, the dense stress snapshots, the cap is the fixed $1-\epsilon$ and the bound widens to $\log((1-\epsilon)/\rho_n)$. The sparse-regime theory assumes $\rho_n/\epsilon\to0$; the dense snapshots are outside it, used as a finite-sample stress check.

\emph{Qualification of optimality.} The recovery-delay lower bound of order $\log(eJ/\alpha)/(n^2\rho_n\kappa^2)$ and the upper bound of order $\log(JT/\alpha)/(n^2\rho_n\kappa^2)$ match in rate up to the logarithmic horizon factor, $\log(JT)$ against $\log(J)$, and the constants from the Freedman and two-point reductions; minimax optimality here means minimax-rate optimality up to logarithmic horizon factors and constants, for the local separated-switch family.

\emph{Predictive mean.} The implementation is a Laplace filter: the predictive mean is the plug-in $\sigma(\widehat\bbeta_{t\mid t-1}^\top z_{t,e})$ and the curvature supplies the intervals. The name predictive synthesis refers to the modeling construction, the synthesis density over agent forecasts, not to an integrated predictive. The integrated predictive $\mathbb{E}[\sigma(\bbeta_t^\top z_{t,e})\mid A_{1:t-1}]$, approximated by the logistic-normal form $\sigma(m_{t\mid t-1}^\top z_{t,e}/\sqrt{1+\pi z_{t,e}^\top C_{t\mid t-1}z_{t,e}/8})$, is available at extra cost and gives the same conclusions.

\end{document}